\documentclass{article}
\usepackage{amssymb}
\usepackage{amsmath}
\usepackage{amsfonts}
\usepackage{graphicx}

\newtheorem{theorem}{Theorem}

\newtheorem{corollary}[theorem]{Corollary}

\newtheorem{lemma}[theorem]{Lemma}

\newtheorem{remark}[theorem]{Remark}

\newenvironment{proof}[1][Proof]{\noindent\textbf{#1.} }{\ \rule{0.5em}{0.5em}}
\input{tcilatex}

\begin{document}

\title{Heteroclinic for a 6-dimensional reversible system occurring in
orthogonal domain walls in convection}
\date{}
\author{G\'{e}rard Iooss \\
%EndAName
{\footnotesize Laboratoire J.A.Dieudonn\'e, I.U.F., Universit\'e C\^ote
d'Azur, CNRS,}\\
{\footnotesize Parc Valrose, 06108 Nice cedex 2, France} \\
{\footnotesize iooss.gerard@orange.fr}}
\maketitle

\begin{abstract}
A six-dimensional reversible normal form system occurs in B\'{e}%
nard-Rayleigh convection between parallel planes, when we look for domain
walls intersecting orthogonally (see Buffoni et al \cite{BHI}). This leads
to \ study analytically the system%
\begin{eqnarray*}
\frac{d^{4}A}{dx^{4}} &=&A(1-A^{2}-gB^{2}) \\
\frac{d^{2}B}{dx^{2}} &=&\varepsilon ^{2}B(-1+gA^{2}+B^{2}),
\end{eqnarray*} for $x\in \mathbb{R}
,$ and looking for a heteroclinic connection between the two equilibria $%
M_{-}:(A,B)=(1,0)$ and $M_{+}:(A,B)=(0,1)$, each corresponding to a system
of convective rolls. \ In \cite{BHI} such a heteroclinic is shown to exist,
on which $0\leq B\leq 1,$ with no uniqueness result and no possibility to
use it for a persistence result under a reversible perturbation. The lack of
normal hyperbolicity in $(A,B)=(0,1/\sqrt{g})$ of equilibria obtained at the
limit $\varepsilon =0,$ is the main problem. The 3-dimensional unstable
manifold of $M_{-}$ is built for $0\leq B\leq \frac{1-c\varepsilon ^{4/5}}{%
\sqrt{g}},$ while, in solving a certain 4th-order differential equation
independent of $\varepsilon $ (occuring in \cite{Man-Pom}, \cite{Buff}), we
overcome the lack of hyperbolicity in building the stable manifold of $M_{+}$
for $\frac{1-c\varepsilon ^{4/5}}{\sqrt{g}}\leq B\leq 1.$ We use \cite{BHI}
\ for proving that the two manifold intersect. Then the two 3-dimensional
manifolds intersect transversally, leading to the existence, uniqueness and
analyticity in $(\varepsilon ,g)$ of the heteroclinic, for which we give
estimates of $A(x),B(x)$ and their derivatives. We finally study the
properties of the linearized operator along the heteroclinic, allowing to
prove (in \cite{Io24}) the persistence of the heteroclinic under
perturbation, corresponding to the existence of orthogonal domain walls in
the B\'{e}nard-Rayleigh convection problem.
\end{abstract}


\bigskip

\section{Introduction and Results}

In this work we study the following 6th order reversible system%
\begin{eqnarray}
A^{(4)} &=&A(1-A^{2}-gB^{2})  \label{truncated syst} \\
B^{\prime \prime } &=&\varepsilon ^{2}B(-1+gA^{2}+B^{2}),  \notag
\end{eqnarray}%
where $A$ and $B$ are real functions of $x\in 
%TCIMACRO{\U{211d} }%
%BeginExpansion
\mathbb{R}
%EndExpansion
.$ This system occurs in the search for domain walls intersecting
orthogonally, in a fluid dynamic problem such as the B\'{e}nard-Rayleigh
convection between parallel horizontal plates (see subsection \ref{sect
origin} and all details in \cite{BHI}). The heteroclinic we are looking for,
corresponds to the connection between rolls on one side and rolls oriented
orthogonally on the other side. The system (\ref{truncated syst}) has been
also introduced by Manneville and Pomeau in \cite{Man-Pom}, obtained after
formal physical considerations using symmetries.

We would like to find analytically a heteroclinic connection ($g>1,$ $%
\varepsilon $ small) such that 
\begin{eqnarray*}
A_{\ast }(x),B_{\ast }(x) &>&0, \\
(A_{\ast }(x),B_{\ast }(x)) &\rightarrow &\left\{ 
\begin{array}{c}
(1,0)\text{ as }x\rightarrow -\infty \\ 
(0,1)\text{ as }x\rightarrow +\infty%
\end{array}%
\right. .
\end{eqnarray*}%
By a variational argument Boris Buffoni et al \cite{BHI} prove the existence
of such an heteroclinic orbit, for any $g>1,$ and $\varepsilon $ small
enough. This type of elegant proof does not unfortunately allow to prove the
persistence of such heteroclinic curve under reversible perturbations of the
vector field. This is our motivation for producing analytic arguments,
proving such an existence, uniqueness and smoothness in parameters $%
(\varepsilon ,g)$ of this orbit (in particular analyticity in $g)$, however
for limited values $10/9<g\leq 2,$ fortunately including physical
interesting ones. Then we study the linearized operator along the
heteroclinic curve, allowing to attack the problem of existence of
orthogonal domain walls in convection (see \cite{Io24} and Remark \ref{rmk
persist}).

\subsection{Origin of system (\protect\ref{truncated syst})\label{sect
origin}}

The B\'{e}nard-Rayleigh convection problem is a classical problem in fluid
mechanics. It concerns the flow of a three-dimensional viscous fluid layer
situated between two horizontal parallel plates and heated from below. Upon
increasing the difference of temperature between the two plates, the simple
conduction state looses stability at a critical value of the temperature
difference corresponding to a critical value $\mathcal{R}_{c}$ of the
Rayleigh number. Beyond the instability threshold, a convective regime
develops in which patterns are formed, such as convective rolls, hexagons,
or squares. Observed patterns are often accompanied by defects.

We start with the Navier-Stokes-Boussinesq (N-S-B) \emph{steady} system of
PDE's, applying spatial dynamics with $x$ as "time" (as introduced by
K.Kirchg\"{a}ssner in \cite{Kirch82}, adapted for N-S equations in \cite%
{IoMiDe}, and more generally in \cite{HIbook}) and considering solutions $%
2\pi /k$ periodic in $y$ (coordinate parallel to the wall). We show in \cite%
{BHI} that near criticality a 12-dimensional center manifold reduction to a
reversible system applies for $(\mathcal{R},k)$ close to $(\mathcal{R}%
_{c},k_{c}),$ $\mathcal{R}$ being the Rayleigh number, and $k_{c}$ the
critical wave number. This high dimension of the center manifold may be
explained as follows. Due to the equivariance of the system under horizontal
shifts, the eigenvectors of the linearized problem are of the form $\exp
i(\pm k_{1}x\pm k_{2}y),$ the factor being only function of $%
k^{2}=k_{1}^{2}+k_{2}^{2}$ (invariance under rotations). It results that,
for eigenvectors independent of $x$ corresponding to a $0$ eigenvalue in the
spatial dynamics formulation, the eigenvalue is double in general (make $\pm
k_{1}\rightarrow 0).$ Now, at criticality, $k=k_{c}$ corresponds to two
different values of $k_{2}$ merging towards $k_{c},$ which doubles the
dimension, making a quadruple $0$ eigenvalue with complex and
complex-conjugate eigenvectors. Hence we already have a dimension 8
invariant subspace for the $0$ eigenvalue, with two $4\times 4$ Jordan
blocks. This corresponds to convective rolls of amplitude $A$ and $\overline{%
A}$ at $x=-\infty .$ Now for eigenvectors independent of $y$ corresponding
to eigenvalues $\pm ik$ in the spatial dynamics formulation it is shown in 
\cite{HI Arma} that they are simple in general, and give double eigenvalues $%
\pm ik_{c}$ for $k=k_{c}$ with amplitudes $B$ and $\overline{B}$ respectively%
$.$ Hence this adds 4 dimensions to the central space, so finally obtaining
a 12-dimensional central space. Now we restrict the study to solutions
invariant under reflection $y\rightarrow -y$ (the change $y$ into $-y$
changing $A$ in $\overline{A}$ and not changing $B)$, which constitutes an 
\emph{invariant subspace for the full system}. This restricts the study to 
\emph{real} amplitudes $A$ and the full system reduces to a 8-dimensional
sub-center manifold, such that $A\in 
%TCIMACRO{\U{211d} }%
%BeginExpansion
\mathbb{R}
%EndExpansion
$ and $B\in 
%TCIMACRO{\U{2102} }%
%BeginExpansion
\mathbb{C}
%EndExpansion
$ are the amplitudes of the rolls respectively at $x=-\infty ,$ and $%
x=+\infty .$ Moreover, for the full system, we keep

i) the reversibility symmetry: $(x,A,B)\rightarrow (-x,A,\overline{B}),$

ii) the equivariance under shifts by half of a period in $y$ direction,
leading to the symmetry: $(A,B)\rightarrow (-A,B).$

Now, in \cite{BHI} we use a normal form reduction up to cubic order, and
rewrite the system as one real 4th order differential equation for $A$, and
a second order complex differential equation for $B.$ In addition to the
above symmetries, the normal form commutes in particular with the symmetry: $%
(A,B)\rightarrow (A,Be^{i\phi }),$ for any $\phi \in 
%TCIMACRO{\U{211d} }%
%BeginExpansion
\mathbb{R}
%EndExpansion
.$

Handling the full N-S-B equations, in \cite{BHI} the authors show that the
study leads to a small perturbation of the reduced system of amplitude
equations (\ref{truncated syst}). More precisely, after a suitable scaling
(see \cite{BHI} and more details in \cite{Io24}), and denoting by $%
(\varepsilon ^{2}A_{0},\varepsilon ^{2}B_{0})$ rescaled amplitudes $%
(A,Be^{-ik_{c}x}),$ and after a rescaling of the coordinate $x,$ we obtain
the system%
\begin{eqnarray}
A_{0}^{(4)} &=&k_{-}A_{0}^{\prime \prime }+A_{0}(1-\frac{k_{-}^{2}}{4}%
-A_{0}^{2}-g|B_{0}|^{2})+\widehat{f},  \notag \\
B_{0}^{\prime \prime } &=&\varepsilon ^{2}B_{0}(-1+gA_{0}^{2}+|B_{0}|^{2})+%
\widehat{g},  \label{new reduced syst}
\end{eqnarray}%
where $\varepsilon ^{4}$ is proportional to $\mathcal{R}-\mathcal{R}_{c},$
the coefficient $g>1$ is function of the Prandtl number and is the same as
introduced and computed in \cite{HI Arma}, $k_{-}$ comes from the freedom
left to the wave number of the rolls at $-\infty ,$ defined as%
\begin{equation*}
k=k_{c}(1+\varepsilon ^{2}k_{-}),
\end{equation*}%
and $\widehat{f}$ and $\widehat{g}$ are perturbation terms, smooth functions
of their arguments, coming

i) from the rest of the cubic normal form, at least of order $\varepsilon
^{2}$ for $\widehat{f},$ and at least of order $\varepsilon ^{3}$ for $%
\widehat{g};$

ii) from higher order terms not in normal form, and not autonomous (because
of the introduction of $Be^{-ik_{c}x}$ rescaled as $\varepsilon ^{2}B_{0}$
in (\ref{new reduced syst})), and of order $\varepsilon ^{4}$ for $\widehat{f%
},$ and of order $\varepsilon ^{6}$ for $\widehat{g}.$ Without $k_{-},$ $%
\widehat{f},$ and $\widehat{g},$ this is the system (\ref{truncated syst}),
with $B_{0}\in 
%TCIMACRO{\U{2102} }%
%BeginExpansion
\mathbb{C}
%EndExpansion
$ replacing $B,$ and $|B_{0}|^{2}$ replacing $B^{2}$. The truncation leading
to (\ref{truncated syst}) allows to take $B$ real, since the phase of $B_{0}$
does not play any role in the dynamics for (\ref{truncated syst}). The two
different wave numbers of the rolls, close to the critical value $k_{c}$ are
left free for the full problem, however they do not appear in the present
proof of the heteroclinic, even though they are important for the final
proof of existence of the orthogonal domain walls (see Remark \ref{rmk
persist} in section \ref{studylinop}). It should be noticed that the system (%
\ref{new reduced syst}), without $\widehat{f}$ and $\widehat{g},$ was
obtained a long time ago by Pomeau-Manneville in \cite{Man-Pom}, however
they did not deal with the full N-S-B system, and only considered cases with
identical wave numbers at infinities, while it is shown in \cite{Io24} that
some cubic terms, not existent in \cite{Man-Pom}, as $\varepsilon
^{2}(A_{0}^{2}A_{0}^{\prime \prime }-A_{0}A_{0}^{\prime 2})$ in $\widehat{f}$
and $i\varepsilon ^{3}B_{0}A_{0}A_{0}^{\prime }$ in $\widehat{g}$ \ are
crucial for the determination of the solutions of the full problem, with
different wave numbers at infinities (see Remark \ref{rmk persist}).

\subsection{Sketch of the method and results}

From now on let us consider the system (\ref{truncated syst}). The
equilibrium $(A,B)=(0,1)$ of the system (\ref{truncated syst}) gives an
approximation of convection rolls parallel to the wall (periodic in the $x$
direction, with fixed phase) bifurcating for Rayleigh numbers $\mathcal{R}>%
\mathcal{R}_{c}$ close to $\mathcal{R}_{c}$, whereas the equilibrium $%
(A,B)=(1,0)$ of the system (\ref{truncated syst}) gives the same convection
rolls (periodic in the $y$ direction) rotated by an angle $\pi /2$ with the
phase fixed by the imposed reflection symmetry. A heteroclinic orbit
connecting these two equilibria provides then an approximation of orthogonal
domain walls (see Figure \ref{fig-wall}).

We set $\delta =(g-1)^{1/2}$. The idea here might be to use the arc of
equilibria $A^{2}+B^{2}=1,$ which exists for $\delta =0,$ connecting end
points $M_{-}=(1,0)$ and $M_{+}=(0,1),$ and to prove that for suitable
values of $\delta $ ($>0$ but close to 0), the 3-dimensional unstable
manifold of $M_{-}$ intersects transversally the 3-dimensional stable
manifold of $M_{+},$ both staying on a 5 dimensional invariant manifold $%
\mathcal{W}_{\varepsilon ,\delta }.$ However, for $\delta =0$ the situation
in $M_{+}$ is very degenerated, with a quadruple $0$ eigenvalue for the
linearized operator, while it is a double eigenvalue for $M_{-}.$ Then for $%
\delta $ close to 0, a 5-dimensional center-stable invariant manifold
starting from $M_{+}$ needs to intersect a four-dimensional center-unstable
manifold starting from $M_{-}$. We are not able to prove this. Moreover, for 
$\delta \neq 0$ but close to 0, we cannot prove that the 3-dimensional
unstable manifold of $M_{-}$ exists from $B=0$ until $B$ reaches a value
close enough to 1. In fact, we may fortunately notice that the physically
interesting values of $\delta $ are not close to 0 (see Remark \ref{Physval}%
). So that we prefer to play with $\varepsilon .$

We may observe that, after changing the coordinate $x$ in $\overline{x}%
=\varepsilon x,$ we obtain the new system%
\begin{eqnarray}
\varepsilon ^{4}\frac{d^{4}A}{d\overline{x}^{4}} &=&A(1-A^{2}-gB^{2})
\label{singular syst} \\
\frac{d^{2}B}{d\overline{x}^{2}} &=&B(B^{2}+gA^{2}-1),  \notag
\end{eqnarray}%
where the limit $\varepsilon \rightarrow 0$ is singular, and gives indeed a
non smooth heteroclinic solution such that

(i) for $x$ running from $-\infty $ to $0$, then $(A,B)$ varies from $(1,0)$
to $(0,\frac{1}{\sqrt{g}})$ on the ellipse $A^{2}+gB^{2}=1,$ while

(ii) for $x$ running from $0$ to $+\infty $, then $(A,B)$ varies from $(0,%
\frac{1}{\sqrt{g}})$ to $(0,1),$ satisfying, in the original coordinate $x,$
the differential equation%
\begin{equation*}
\frac{dB}{dx}=\frac{\varepsilon }{\sqrt{2}}(1-B^{2}).
\end{equation*}%
\begin{figure}[th]
\begin{center}
\includegraphics[width=6cm]{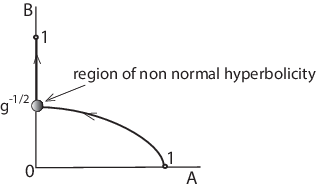}
\end{center}
\caption{Critical manifold}
\label{fig-wall}
\end{figure}
The two manifolds $A=\widetilde{A_{\ast }}=(1-gB^{2})^{1/2},$ and $A=0$ are
named "slow manifolds" in litterature (see \cite{Feni},\cite{Krupa}). We
might then think to use Fenichel's theorems \cite{Feni} on the system (\ref%
{truncated syst}) for $\varepsilon $ close to 0. For the part (i) of the
curve, where $x\in (-\infty ,0]$, the set of equilibria, here $A^{2}+gB^{2}=1
$, is not normally hyperbolic at the end point $(A,B)=(0,1/\sqrt{g})$ (see
in section \ref{sect unstable manif new} eigenvalues of the linear operator $%
\boldsymbol{L}_{\delta }$ corresponding to $\widetilde{A_{\ast }}=0$). For
the second part (ii) of the curve, where $x\in \lbrack 0,+\infty ),$ the set
of equilibria $(A,B)=(0,B)$ is also not normally hyperbolic for $B=1/\sqrt{g}
$ (the 4 remaining eigenvalues are such that $\lambda ^{4}=1-gB^{2}$ which
cancel for $B=1/\sqrt{g}).$ The normal hyperbolicity is essential in
Fenichel's theorems, so we cannot use them directly. However we may use
normal hyperbolicity up to a small neighborhood of $(A,B)=(0,1/\sqrt{g}),$
as this is done in sections \ref{sect unstable manif new} for finding the
unstable manifold of $(A,B)=(1,0)$ in a neighborhoof of the slow manifold $A=%
\widetilde{A_{\ast }},$ and in section \ref{sec 4.4} for finding the stable
manifold of $(A,B)=(0,1)$ in a neighborhood of the slow manifold $A=0.$

The neighborhood of $(A,B)=(0,1/\sqrt{g})$ not reached by the method above
has a size of order $\mathcal{O}(\varepsilon ^{4/5}).$ We could think to use
a geometric analysis, as Krupa et al did in \cite{Krupa}, where a blow-up
method is used for getting a system independent of $\varepsilon $. Indeed
the scaling%
\begin{eqnarray*}
A &=&K^{2}\varepsilon ^{2/5}\overline{A}, \\
B &=&\frac{1}{\sqrt{g}}(1+\frac{K^{4}}{2}\varepsilon ^{4/5}\overline{B}), \\
x &=&\frac{z}{K\varepsilon ^{1/5}},\text{ }
\end{eqnarray*}%
with $\overline{B}=z,$ since at main order%
\begin{equation*}
\overline{B}^{\prime \prime }=0,\text{ }\overline{B}(0)=0,
\end{equation*}%
leads, at main order, to%
\begin{equation}
\frac{d^{4}\overline{A}}{dz^{4}}=-\overline{A}(\overline{A}^{2}+z),\text{ }%
z\in \lbrack -a_{-},+a_{+}],  \label{intermediate}
\end{equation}%
which is independent of $\varepsilon .$ However, the work of \cite{Krupa} is
made in 2 dimensions, while we have here the 6-dimensional system (\ref%
{truncated syst}). It results that the nice pictures of \cite{Krupa} would
be very hard to transpose here. In addition, we need to satisfy boundary
values (also independent of $\varepsilon )$ coming on the left side from the
connection with the unstable manifold, and from the right side from the
connection with the stable manifold.

Moreover we need to provide \emph{precise estimates} (in function of $%
\varepsilon )$ of the interval of values for $B,$ between the value reached
by the unstable manifold, via the standard method, and the value reached
(backwards) by $B$ for the stable manifold. The critical value $1/\sqrt{g}$
is included in this finite interval, and this finiteness is essential for
extending the existence of the stable manifold on the full interval, until
it meets the unstable manifold.

In section \ref{sect linear} we see that there are 3 unstable
eigendirections starting from $M_{-}=(1,0),$ and 3 stable eigendirections in 
$M_{+}=(0,1).$ The difficulty in the proof of Theorem \ref{theorem} is to
obtain a precise estimate for the existence of the 3-dimensional unstable manifold
of $M_{-},$ where the coordinate $B$ varies from $0$ to a neighborhood of $1/%
\sqrt{g},$ and to obtain a precise estimate for the existence of the 3-dimensional
stable manifold of $M_{+}$ until $B$ varies from $1$ (backwards) to a
neighborhood of $1/\sqrt{g}=1/\sqrt{1+\delta ^{2}},$ while $A$ stays close
to $0.$ For approaching the closest possible to $B=1/\sqrt{g},$ we use the
first integral of (\ref{truncated syst}), which implies that both invariant
manifolds are included in a 5-dimensional invariant manifold. We are able to
obtain the unstable manifold of $M_{-}$ for $0\leq B\leq \frac{%
1-c\varepsilon ^{4/5}}{\sqrt{g}},$ while we first obtain the stable manifold
of $M_{+}$ for $\frac{1+c^{\prime }\varepsilon ^{4/5}}{\sqrt{g}}\leq B\leq 1.
$ For extending the existence of the stable manifold in the gap of size of
order $\varepsilon ^{4/5},$ we need to solve the 4th order differential
equation (\ref{intermediate}), independent of $\varepsilon ,$ also found in 
\cite{Man-Pom} and \cite{Buff}, after rescaling, where the boundary
conditions, also independent of $\varepsilon ,$ come from the 2 times 2
parameters introduced by each invariant manifolds arriving in $\pm a_{\pm }.$

We use a precise estimate on $a_{+}$ for being able to extend the domain of
existence of the stable manifold, for $B$ in the interval $\frac{%
1-c\varepsilon ^{4/5}}{\sqrt{g}}\leq B\leq 1.$ Using results of \cite{BHI}
the two manifolds intersect. We prove the following

\begin{theorem}
\label{theorem}Let us choose $1/3\leq \delta \leq 1$, then for $\varepsilon $
small enough, the 3-dim unstable manifold of $M_{-}$ intersects
transversally the 3-dim stable manifold of $M_{+},$ except maybe for a
finite set of values of $\delta .$ The connecting curve which is obtained is
unique (see Remark \ref{rmk sym}). Moreover its dependency in parameters $%
(\varepsilon ,\delta )$ is analytic. In addition we have $B(x)>0$ and $%
B^{\prime }(x)>0$ on $(-\infty ,+\infty ).$ For $x\rightarrow -\infty $ we
have $(A-1,A^{\prime },A^{\prime \prime },A^{\prime \prime \prime
},B,B^{\prime })\rightarrow 0$ at least as $e^{\varepsilon \delta x},$ while
for $x\rightarrow +\infty ,$ $(A,A^{\prime },A^{\prime \prime },A^{\prime
\prime \prime })\rightarrow 0$ at least as $e^{-\sqrt{\frac{\delta }{2}}x},$
and $(B-1,B^{\prime })\rightarrow 0$ at least as\ $e^{-\sqrt{2}\varepsilon
x}.$
\end{theorem}

Moreover we also have important estimates as follows, extensively used in 
\cite{Io24}.

\begin{corollary}
\label{Corollary unstab} For $x\in (-\infty ,0]$ and choosing $\delta ^{\ast
}<\delta $, there exists $c>0$ independent of $\varepsilon $ small enough,
such that the heteroclinic curve satisfies 
\begin{eqnarray*}
|A(x)-\sqrt{1-(1+\delta ^{2})B(x)}| &\leq &c\varepsilon
^{2/5}B(x)e^{\varepsilon \delta ^{\ast }x} \\
|A^{(m)}(x)| &\leq &c\varepsilon ^{3/5}B(x)e^{\varepsilon \delta ^{\ast }x},%
\text{ }m=1,2,3.
\end{eqnarray*}
\end{corollary}

\begin{corollary}
\label{Corollary stab} For $x\in \lbrack 0,+\infty )$ and $\delta _{\ast }=%
\frac{1}{10}\delta ^{2/5}$, there exists $c>0$ independent of $\varepsilon $
small enough, such that the heteroclinic curve satisfies%
\begin{equation*}
|A^{(m)}(x)|\leq c\varepsilon ^{2/5}e^{-\delta _{\ast }\varepsilon ^{1/5}x},%
\text{ }m=0,1,2,3.
\end{equation*}
\end{corollary}

\begin{remark}
\label{rmk intersec}It should be noticed that we show at Lemma \ref%
{intersection of manifolds} that, in the middle of the heteroclinic, $A(0)=%
\mathcal{O}(\varepsilon ^{2/5})$ and for $x\in (0,+\infty ),$ $A(x)$
oscillates, staying of order $\mathcal{O}(\varepsilon ^{2/5})$, while $%
B(0)=1/\sqrt{g}$ and $B(x)$ grows monotonically until $1.$
\end{remark}

\begin{remark}
\label{rmk sym}Using symmetries of the system: $A\mapsto \pm A,$ $B\mapsto
\pm B$ and reversibility symmetry: $(A(x),B(x))\mapsto (A(-x),B(-x)),$ we
find 8 heteroclinics. Two are connecting $M_{-}$ to $M_{+}$ with opposite
dynamics, two others connect $-M_{-}$ to $M_{+},$ two connect $M_{-}$ to $%
-M_{+},$ and two connect $-M_{-}$ to $-M_{+}.$ The one which interests us is
the only one connecting $M_{-}$ to $M_{+}$ with the dynamics running from $%
M_{-}$ to $M_{+}.$
\end{remark}

\begin{remark}
It should be noticed that the study made in \cite{Man-Pom} on the
heteroclinic solution for the system (\ref{truncated syst}) uses asymptotic
analysis, suggesting the existence of the heteroclinic, later proved
mathematically in \cite{BHI}. Contrary to these previous works, using
asymptotic analysis on the full real line, the precise estimate which is
obtained for $a_{+}$(see (\ref{intermediate})) is essential here, for
getting a rigorous result.
\end{remark}

\begin{remark}
\label{Physval} Values of $\delta $ such that $0.476\leq \delta $ include
values obtained for $\delta $ in the B\'{e}nard-Rayleigh convection problem
where $g=1+\delta ^{2}$ is function of the Prandtl number $\mathcal{P}$ (as
computed in \cite{HI Arma}). With rigid-rigid, rigid-free, or free-free
boundaries the minimum values of $g$ are respectively $(g_{\min }=1.227,$ $%
1.332,$ $1.423)$ corresponding to $\delta _{\min }=0.476,$ $0.576,$ $0.650.$
The restriction in Theorem \ref{theorem} corresponds to $1<g\leq 2.$ Then,
the eligible values for the Prandtl number are respectively $\mathcal{P}%
>0.5308,>0.6222,>0.8078$. 
\begin{figure}[th]
\begin{center}
\includegraphics[width=4cm]{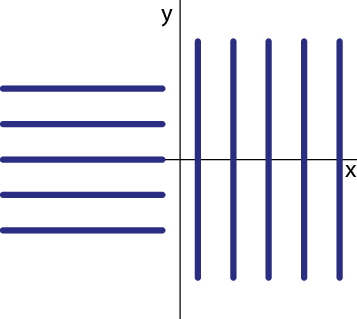}
\end{center}
\caption{Orthogonal domain wall}
\label{fig-wall}
\end{figure}
\end{remark}

The Schedule of the paper is as follows: in section \ref{sect unstable manif
new} we prove at Lemma \ref{unstablemanifM-} the existence of the
3-dimensional unstable manifold of $M_{-}=(1,0)$ for $B\in \lbrack
0,(1-\alpha _{-}^{2}\delta ^{2})/\sqrt{1+\delta ^{2}}]$ with $\varepsilon
=\nu _{-}\alpha _{-}^{5/2},$ $\nu _{-}$ being independent of $\varepsilon ,$
for $x\in (-\infty ,-x_{\ast }].$

In section \ref{sect stabmanifold} we prove at Lemma \ref{Lemma stable
manifM+} the existence of the 3-dimensional stable manifold of $M_{+}=(0,1)$
for $B\in \lbrack (1+\alpha _{+}^{2}\delta ^{2})/\sqrt{1+\delta ^{2}},1],$
and $x\in \lbrack x_{\ast }^{+},+\infty ).$

In section \ref{sect intersection}, we solve a certain 4th-order
differential equation, on the finite interval $[-x_{\ast },x_{\ast }^{+}]$
which, once rescaled, is independent of $\varepsilon ,$ so that we are able
to extend the existence of the stable manifold for $B\in \lbrack (1-\alpha
_{-}^{2}\delta ^{2})/\sqrt{1+\delta ^{2}}),(1+\alpha _{+}^{2}\delta ^{2})/%
\sqrt{1+\delta ^{2}})],$ $A_{0}$ still of order $\mathcal{O}(\alpha _{+}).$
Then we use results of \cite{BHI} to control the existence of the
intersection, and then prove the transverse intersection of the two
manifolds. This ends the proof of Theorem \ref{theorem}.

In section \ref{studylinop} we give, in Lemma \ref{Lem Linearoperator},
properties of the linearized operator along the heteroclinic, which are
necessary to prove a persistence result under a reversible perturbation for
the heteroclinic in the 8-dimensional space (with $B\in \mathbb{C}$). This
allows to prove the existence of orthogonal domain walls in convection as
made in \cite{Io24}.

In summary, what is new in this paper?

i) Existence of the unstable manifold of $M_{-},$ analytic in $\delta ,$
while coordinate $B$ is an increasing function for $x\in (-\infty ,-x_{\ast
}]$, varying from $0$ to a value $\mathcal{O}(\varepsilon ^{4/5})-$ close to 
$1/\sqrt{g}.$ Existence of the stable manifold of $M_{+}$, analytic in $%
\delta ,$ while coordinate $B$ is an increasing function for $x\in \lbrack
x_{\ast }^{+},+\infty )$, varying from a value $\mathcal{O}(\varepsilon
^{4/5})-$ close to $1/\sqrt{g},$ to 1.

ii) Justification and resolution backwards of the intermediate 4th order
differential equation (\ref{intermediate}) independent of $\varepsilon ,$
already introduced in \cite{Man-Pom} and \cite{Buff}, but now on a bounded
interval with boundary conditions independent of $\varepsilon $, the
solution being analytic in $\delta .$

iii) On the heteroclinic, $B_{0}(x)$ satisfies $B_{0}^{\prime }(x)>0$. 
Estimates for coordinates in $\mathbb{R}^{6}$ are established, which are essential for a further study on the
persistence under perturbations of the heteroclinic, as for the B\'{e}nard-Rayleigh convection
problem.

iv) Study of the conditions for the invertibility of the linearized
operator, along the heteroclinic, useful for any perturbation result.

\begin{remark}
There are many lengthy awful calculations in this work. However, they are necessary for getting precise estimates.
\end{remark}

\textbf{Acknowledgement} The author warmly thanks Mariana Haragus for her
help in section \ref{studylinop}, and her constant encouragements.  Warm thanks also to the referees who gave the author additional references and urged him to clarify some points of the proofs.

\section{General properties of the system}

\subsection{Global invariant manifold $\mathcal{W}_{\protect\varepsilon ,%
\protect\delta }$}

Let us define coordinates in $%
%TCIMACRO{\U{211d} }%
%BeginExpansion
\mathbb{R}
%EndExpansion
^{6}$ as 
\begin{equation*}
(A_{0},A_{1},A_{2},A_{3},B_{0},B_{1})=(A,A^{\prime },A^{\prime \prime
},A^{\prime \prime \prime },B,B^{\prime }).
\end{equation*}%
The first observation is that we have the first integral 
\begin{equation}
W=\varepsilon ^{2}(A^{\prime 2})^{\prime \prime }-3\varepsilon ^{2}A^{\prime
\prime 2}-B^{\prime 2}+\frac{\varepsilon ^{2}}{2}(A^{2}+B^{2}-1)^{2}+%
\varepsilon ^{2}\delta ^{2}A^{2}B^{2},  \label{first integ}
\end{equation}%
as noticed in \cite{Man-Pom}, where $W$ is used in an energy functional,
used later in \cite{BHI}. Then, for containing the end points $M_{\pm }$,
our heteroclinic should satisfy 
\begin{equation}
2\varepsilon ^{2}A_{1}A_{3}-\varepsilon ^{2}A_{2}^{2}-B_{1}^{2}+\frac{%
\varepsilon ^{2}}{2}(A_{0}^{2}+B_{0}^{2}-1)^{2}+\varepsilon ^{2}\delta
^{2}A_{0}^{2}B_{0}^{2}=0.  \label{Wg}
\end{equation}%
Since our purpose is to find $B_{0}$ growing from 0 to 1, we extract the
positive square root (needs to be justified later):%
\begin{equation*}
B_{1}=\{2\varepsilon ^{2}A_{1}A_{3}-\varepsilon ^{2}A_{2}^{2}+\frac{%
\varepsilon ^{2}}{2}(A_{0}^{2}+B_{0}^{2}-1)^{2}+\varepsilon ^{2}\delta
^{2}A_{0}^{2}B_{0}^{2}\}^{1/2},
\end{equation*}%
which defines a \emph{5-dimensional invariant maniford }$\mathcal{W}%
_{\varepsilon ,\delta }$ valid for any $\delta >0,$ which should contain the
heteroclinic curve that we are looking for.

For $\delta >0,$ we find the singular points (where a tangent hyperplane is
not defined)%
\begin{eqnarray}
(A_{0},B_{0}) &=&(\pm 1,0),\text{ }A_{1}=A_{2}=A_{3}=B_{1}=0
\label{singular points} \\
(A_{0},B_{0}) &=&(0,\pm 1),\text{ }A_{1}=A_{2}=A_{3}=B_{1}=0.  \notag
\end{eqnarray}%
For $\delta =0,$ singular points constitute the circle 
\begin{equation}
A_{0}^{2}+B_{0}^{2}=1,\text{ }A_{1}=A_{2}=A_{3}=B_{1}=0.  \label{sing g=1}
\end{equation}

\begin{remark}
We do not emphasize here on the hamiltonian structure of system (\ref%
{truncated syst}) since this does not help our understanding. On the
contrary, the reversibility property is inherited from the original physical
problem and is still valid for the perturbed system (\ref{new reduced syst}%
). Moreover, if we consider perturbation terms as $\varepsilon
^{2}(A_{0}^{2}A_{0}^{\prime \prime }-A_{0}A_{0}^{\prime 2})$ in $\widehat{f}$
and $i\varepsilon ^{3}B_{0}A_{0}A_{0}^{\prime }$ in $\widehat{g},$ we cannot
find a new first integral analogue to (\ref{Wg}), while the system is still
reversible.
\end{remark}

\subsection{Linear study of the dynamics\label{sect linear}}

\subsubsection{Neighborhood of $M_{-}=(1,0)$}

The eigenvalues of the linearized operator at $M_{-}$ are such that $\lambda
^{4}=-2$ or $\lambda ^{2}=\varepsilon ^{2}\delta ^{2},$ hence they are $\pm
2^{-1/4}(1\pm i)$ and $\pm \varepsilon \delta .$ This gives a 3-dimensional
unstable manifold, and a 3-dimensional stable manifold, originating from $%
M_{-}.$

\subsubsection{Neighborhood of $M_{+}=(0,1)$}

The eigenvalues of the linearized operator at $M_{+}$ are such that $\lambda
^{4}=-\delta ^{2}$ or $\lambda ^{2}=2\varepsilon ^{2},$ hence defining $%
\delta ^{\prime }=\sqrt{\delta },$ the eigenvalues are $\pm 2^{-1/2}(1\pm
i)\delta ^{\prime },$ and $\pm \varepsilon \sqrt{2}.$ This gives again a
3-dimensional unstable manifold and a 3-dimensional stable manifold
originating from $M_{+}.$

All this implies that the 3-dimensional unstable manifold starting at $M_{-}$
and the 3-dimensional stable manifold starting at $M_{+}$ which are both
included into the $5-$dimensional manifold $\mathcal{W}_{\varepsilon ,\delta
},$ give a good hope for these two manifolds to intersect along a
heteroclinic curve...provided that they still exist as graphs with respect
to $B,$ "far" from the end points $M_{+}$ and $M_{-}.$ The idea is to show
that this occurs when $\delta $ is not too small and at most 1.

\begin{remark}
The limit points $M_{-}=(1,0)$ and $M_{+}=(0,1)$ have a degenerate situation
for $\delta =0,$ because of the multiple $0$ eigenvalue for the linearized
operator. For $\delta =0,$ it is possible to build a family of 2-dim
unstable invariant manifolds and a family of 2-dim stable manifolds along
the arc of equilibria $A^{2}+B^{2}=1.$ For $\delta >0$ and small, the
perturbation gives two new 3-dim invariant manifolds, however their
transversality is weaker and weaker as $B\rightarrow 1$ (so that Fenichel's
theorem cannot apply). A more "serious" study would then be needed. However
the physical interest is for values of $\delta >0$ not too small, which
cancels the physical interest of such a difficult question (see Remark \ref%
{Physval}).
\end{remark}

\section{Unstable manifold of $M_{\_}\label{sect unstable manif new}$}

\subsection{Choice of coordinates\label{sec changecoord}}

Let us assume in this section $1/3\leq \delta \leq 1$ and define $\eta _{0}$
and $\alpha _{-}$ such that%
\begin{eqnarray*}
0 &\leq &B_{0}\leq \sqrt{1-\eta _{0}^{2}\delta ^{2}}=\left( \frac{1-\alpha
^{2}\delta ^{2}}{1+\delta ^{2}}\right) ^{1/2},\text{ \ }\eta _{0}^{2}=\frac{%
1+\alpha ^{2}}{1+\delta ^{2}}, \\
\frac{1}{\sqrt{g}} &=&\frac{1}{\sqrt{1+\delta ^{2}}}<\eta _{0}<\frac{1}{%
\delta },\text{ }\alpha \delta <1.
\end{eqnarray*}%
$\alpha $ will be determined later, as a power of $\varepsilon .$ Now, we
define%
\begin{equation}
\widetilde{A_{\ast }}^{2}\overset{def}{=}1-(1+\delta ^{2})B_{0}^{2},
\label{def A*}
\end{equation}%
then%
\begin{equation}
\alpha \delta \leq \widetilde{A_{\ast }}\leq 1.  \label{lower bound A*}
\end{equation}%
Let us define the following coordinates in $%
%TCIMACRO{\U{211d} }%
%BeginExpansion
\mathbb{R}
%EndExpansion
^{6}:$%
\begin{equation}
Z=(\widetilde{A_{\ast }}+\widetilde{A_{0}}%
,A_{1},A_{2},A_{3},B_{0},B_{1})^{t}.  \label{new coord 1}
\end{equation}

\begin{remark}
The assumption $B_{0}\geq 0$ comes from the result of \cite{BHI} where $%
0\leq B_{0}(x)\leq 1$ along the heteroclinic; $\widetilde{A_{\ast }}$ is
just the first part of the "singular" heteroclinic found for the system
singular for $\varepsilon =0$ (\ref{singular syst}). The occurence of $%
\widetilde{A_{\ast }}$ is also linked with a formal computation of an
expansion of the heteroclinic in powers of $\varepsilon ,$ which gives $%
\widetilde{A_{\ast }}$ as the principal part of $A_{0},$ valid for $%
B_{0}<(1+\delta ^{2})^{-1/2}=1/\sqrt{g}.$ We expect to build the unstable
manifold until this limit value.
\end{remark}

\begin{remark}
We put a condition $1/3\leq \delta $ in the purpose to have not too small
values for $\delta ,$ and to include known computed values of the
coefficient $g=1+\delta ^{2},$ in the convection problems, with different
boundary conditions (see Remark \ref{Physval} and \cite{HI Arma}). The
restriction $\delta \leq 1$ made in this section simplifies few estimates,
and is not really useful.
\end{remark}

We prove below the main result of this section:

\begin{lemma}
\label{unstablemanifM-} For $1/3\leq \delta \leq 1,$ and $\varepsilon $
small enough, the 3-dimensional unstable manifold $\mathcal{W}_{\varepsilon
,\delta }^{(u)}$ of $M_{-}$ exists for%
\begin{equation*}
0\leq B_{0}(x)\leq (1-\eta _{0}^{2}\delta ^{2})^{1/2}=\left( \frac{1-\alpha
_{-}^{2}\delta ^{2}}{1+\delta ^{2}}\right) ^{1/2},\text{ }x\in (-\infty
,-x_{\ast }],
\end{equation*}%
where $x_{\ast }\geq 0$ is arbitrary, and there exists $\nu _{-}>0$ such
that $\eta _{0}$ and $\alpha _{-}$ satisfy%
\begin{equation*}
(1+\delta ^{2})\eta _{0}^{2}=1+\alpha _{-}^{2},\text{ }\varepsilon =\nu
_{-}\alpha _{-}^{5/2}.
\end{equation*}%
The manifold $\mathcal{W}_{\varepsilon ,\delta }^{(u)}$ sits in $\mathcal{W}%
_{\varepsilon ,\delta }$, is analytic in $(\varepsilon ,\delta )$,
parametrized by ($X(-x_{\ast }),B_{0}(-x_{\ast })$) where $X(x)$ is a
2-dimensional coordinate on the $B_{0}-$dependent unstable directions
defined by (\ref{new linop 1}). Moreover, for any $\delta ^{\ast }<\delta $,
there exist a number $k_{0}>0$ independent of $\varepsilon ,\nu _{-}$ such
that for%
\begin{equation*}
B_{0}(-x_{\ast })=\sqrt{1-\eta _{0}^{2}\delta ^{2}},\text{ }\widetilde{%
A_{\ast }}(x)\geq \alpha \delta ,\text{ }|X(-x_{\ast })|\leq k_{0}\delta
\alpha _{-}^{3/2}=\frac{k_{0}}{\nu _{-}^{3/5}}\delta \varepsilon ^{3/5},
\end{equation*}%
we have%
\begin{eqnarray*}
A_{0}(x) &=&\widetilde{A_{\ast }}(x)+B_{0}(x)\mathcal{O}\left( \frac{%
|X(-x_{\ast })|}{\widetilde{A_{\ast }}(x)^{1/2}}e^{\varepsilon \delta ^{\ast
}x}\right) \\
A_{1}(x) &=&B_{0}(x)e^{\varepsilon \delta ^{\ast }x}\mathcal{O}(\nu
^{2/5}\varepsilon ^{3/5}+|X(-x_{\ast })|) \\
A_{2}(x) &=&B_{0}(x)\widetilde{A_{\ast }}(x)^{1/2}\mathcal{O}(|X(-x_{\ast
})|e^{\varepsilon \delta ^{\ast }x}) \\
A_{3}(x) &=&B_{0}(x)\widetilde{A_{\ast }}(x)\mathcal{O}(|X(-x_{\ast
})|e^{\varepsilon \delta ^{\ast }x}),
\end{eqnarray*}%
\begin{equation*}
0\leq 1-\widetilde{A_{\ast }}\leq cB_{0}^{2},\text{ }\widetilde{A_{\ast }}%
|_{B_{0}=0}=1.
\end{equation*}
\end{lemma}

\begin{remark}
We observe that when $x\rightarrow -x_{\ast },$ $A_{0}$ reaches a value
close to $0$ since $\widetilde{A_{\ast }}$ reaches $\mathcal{O}(\varepsilon
^{2/5})$ which is close to $0$, while $B_{0}$ reaches $(1-\eta
_{0}^{2}\delta ^{2})^{1/2}$ which is expected as close as possible to $%
1/(1+\delta ^{2})^{1/2}=1/\sqrt{g}$. The numbers $k_{0}$ and $\nu _{-}$ will
later be imposed, small enough, independently of $\varepsilon .$ Later $%
x_{\ast }$ will be chosen of order $\varepsilon ^{-1/5}$ (see section \ref%
{sect x*}).
\end{remark}

\textbf{Strategy in section \ref{sect unstable manif new}}.

The strategy is first in sections \ref{sec changecoord}, \ref{sec 3.2}, \ref%
{sec 3.3} to write the system (\ref{truncated syst}) in adapted coordinates,
in a neighborhood of $(\widetilde{A_{0}},A_{1},A_{2},A_{3},B_{1})=0,$ which
are $B_{0}-$ dependent. Then, in sections \ref{sec 3.4}, \ref{sec 3.5} we
eliminate $B_{1}$ (coordinate $z_{1})$ in using the first integral (\ref{Wg}%
).We now look for a 3-dimensional unstable manifold lying in the
5-dimensional manifold $\mathcal{W}_{\varepsilon ,\delta }.$ In sections \ref%
{sec 3.6}, \ref{sec 3.7}, \ref{sec 3.8} we solve the system for the unstable
manifold, the function $B_{0}(x)$ taken as a parameter, still unknown. Then,
in section \ref{sec 3.9} we solve the remaining differential equation for $%
B_{0}(x).$ In section \ref{sec 3.10} we give a quite explicit form for the
2-dimensional intersection of the 3-dimensional plane tangent to the
unstable manifold $\mathcal{W}_{\varepsilon ,\delta }^{(u)}$ with the
hyperplane $H_{0}$ defined by $B_{0}=\sqrt{1-\eta _{0}^{2}\delta ^{2}}.$

With the system of coordinates (\ref{new coord 1}), the system (\ref%
{truncated syst}) becomes%
\begin{eqnarray}
\widetilde{A_{0}}^{\prime } &=&A_{1}+\frac{(1+\delta ^{2})B_{0}}{\widetilde{%
A_{\ast }}}B_{1}  \notag \\
A_{1}^{\prime } &=&A_{2}  \notag \\
A_{2}^{\prime } &=&A_{3}  \label{newsyst 1} \\
A_{3}^{\prime } &=&-2\widetilde{A_{\ast }}^{2}\widetilde{A_{0}}-3\widetilde{%
A_{\ast }}\widetilde{A_{0}}^{2}-\widetilde{A_{0}}^{3}  \notag \\
B_{0}^{\prime } &=&B_{1}  \notag \\
B_{1}^{\prime } &=&\varepsilon ^{2}\delta ^{2}B_{0}(\widetilde{A_{\ast }}%
^{2}-B_{0}^{2})+2\varepsilon ^{2}(1+\delta ^{2})\widetilde{A_{\ast }}B_{0}%
\widetilde{A_{0}}+\varepsilon ^{2}(1+\delta ^{2})B_{0}\widetilde{A_{0}}^{2}.
\notag
\end{eqnarray}%
We expect that $\widetilde{A_{0}},A_{1},A_{2},A_{3},B_{1}$ stay small enough
for $x\in (-\infty ,0],$ so we introduce the $B_{0}-$dependent linear
operator (not really a linearization at some equilibrium) 
\begin{equation}
\boldsymbol{L}_{\delta }=\left( 
\begin{array}{cccccc}
0 & 1 & 0 & 0 & 0 & \frac{(1+\delta ^{2})B_{0}}{\widetilde{A_{\ast }}} \\ 
0 & 0 & 1 & 0 & 0 & 0 \\ 
0 & 0 & 0 & 1 & 0 & 0 \\ 
-2\widetilde{A_{\ast }}^{2} & 0 & 0 & 0 & 0 & 0 \\ 
0 & 0 & 0 & 0 & 0 & 1 \\ 
2\varepsilon ^{2}(1+\delta ^{2})\widetilde{A_{\ast }}B_{0} & 0 & 0 & 0 & 0 & 
0%
\end{array}%
\right) .  \label{new linop 1}
\end{equation}%
The idea is to find new coordinates such that we are able to give nice
estimates of the monodromy operator not forgetting that the coefficients of $%
\boldsymbol{L}_{\delta }$ are functions of $B_{0}.$

The operator $\boldsymbol{L}_{\delta }$ has a double eigenvalue $0$, and is
such that the non zero eigenvalues satisfy%
\begin{equation}
\lambda ^{4}-2\varepsilon ^{2}B_{0}^{2}(1+\delta ^{2})^{2}\lambda ^{2}+2%
\widetilde{A_{\ast }}^{2}=0,  \label{equ lambda}
\end{equation}%
with a discriminant such as 
\begin{equation*}
\Delta ^{\prime }=\varepsilon ^{4}B_{0}^{4}(1+\delta ^{2})^{4}-2\widetilde{%
A_{\ast }}^{2}.
\end{equation*}%
We have the following

\begin{lemma}
\label{Lem delta}For $B_{0}\leq \sqrt{1-\eta _{0}^{2}\delta ^{2}},$ $\alpha
\geq \frac{10}{3}\varepsilon ^{2},$ and $\varepsilon $ small enough, we have%
\begin{equation*}
-\Delta ^{\prime }\geq \widetilde{A_{\ast }}^{2}.
\end{equation*}
\end{lemma}

\begin{proof}
$-\Delta ^{\prime }\geq \widetilde{A_{\ast }}^{2}$ is equivalent to%
\begin{equation}
\varepsilon ^{2}B_{0}^{2}(1+\delta ^{2})^{2}\leq \widetilde{A_{\ast }},
\label{basic ineq a}
\end{equation}%
hence

\begin{equation*}
B_{0}^{2}(1+\delta ^{2})\leq \frac{\lbrack 1+4\varepsilon ^{4}(1+\delta
^{2})^{2}]^{1/2}-1}{2\varepsilon ^{4}(1+\delta ^{2})^{2}},
\end{equation*}%
which is satisfied when%
\begin{equation*}
B_{0}^{2}(1+\delta ^{2})\leq 1-\varepsilon ^{4}(1+\delta ^{2})^{2},
\end{equation*}%
provided that $\varepsilon ^{4}(1+\delta ^{2})^{2}<2$ which is true for $%
\varepsilon $ small enough. Now, since $B_{0}\leq \sqrt{1-\eta
_{0}^{2}\delta ^{2}},$ the above inequality is satisfied a soon as we have%
\begin{equation*}
1-\delta ^{2}\alpha ^{2}<1-\varepsilon ^{4}(1+\delta ^{2})^{2},
\end{equation*}%
which is realized when $\alpha \geq \frac{10}{3}\varepsilon ^{2}.$
\end{proof}

Then we have two pairs of complex eigenvalues 
\begin{equation*}
\lambda _{\pm }^{2}=\varepsilon ^{2}B_{0}^{2}(1+\delta ^{2})^{2}\pm i\sqrt{%
-\Delta ^{\prime }}.
\end{equation*}%
We intend to find new coordinates able to manage a new linear operator in
the form of two independent blocs%
\begin{equation}
\left( 
\begin{array}{cc}
\pm \lambda _{r} & \lambda _{i} \\ 
-\lambda _{i} & \pm \lambda _{r}%
\end{array}%
\right)  \label{block standard}
\end{equation}%
for which the eigenvalues are%
\begin{equation*}
\pm \lambda _{r}\pm i\lambda _{i},
\end{equation*}%
where%
\begin{eqnarray}
2\lambda _{r}^{2} &=&\sqrt{2}\widetilde{A_{\ast }}+\varepsilon
^{2}B_{0}^{2}(1+\delta ^{2})^{2}  \label{def lambda} \\
2\lambda _{i}^{2} &=&\sqrt{2}\widetilde{A_{\ast }}-\varepsilon
^{2}B_{0}^{2}(1+\delta ^{2})^{2}  \notag \\
\lambda _{r}^{2}-\lambda _{i}^{2} &=&\varepsilon ^{2}B_{0}^{2}(1+\delta
^{2})^{2}  \notag \\
\lambda _{r}^{2}+\lambda _{i}^{2} &=&\sqrt{2}\widetilde{A_{\ast }}  \notag \\
4\lambda _{r}^{2}\lambda _{i}^{2} &=&-\Delta ^{\prime }.  \notag
\end{eqnarray}%
A form of the linear operator as (\ref{block standard}) is such that we are
able to have good estimates for the monodromy operator associated with the
linear operator $\boldsymbol{L}_{\delta }$, the coefficients of which are
functions of $B_{0}\in \lbrack 0,\sqrt{1-\eta _{0}^{2}\delta ^{2}}]$ (see
Appendix \ref{App1}). Then we have the following

\begin{lemma}
\label{Lem estim lambda}For $B_{0}\leq \sqrt{1-\eta _{0}^{2}\delta ^{2}},$ $%
\alpha \geq \frac{10}{3}\varepsilon ^{2},$ and for $\varepsilon $ small
enough, we have%
\begin{equation*}
\lambda _{r}\lambda _{i}\geq \frac{\widetilde{A_{\ast }}}{2},
\end{equation*}%
\begin{equation}
2^{1/4}\widetilde{A_{\ast }}^{1/2}\geq \lambda _{r}\geq \frac{\widetilde{%
A_{\ast }}^{1/2}}{2^{1/4}},  \label{est lambda r}
\end{equation}%
\begin{equation}
\frac{1}{2^{5/4}}\widetilde{A_{\ast }}^{1/2}\leq \lambda _{i}\leq \frac{%
\widetilde{A_{\ast }}^{1/2}}{2^{1/4}},  \label{est lambda i}
\end{equation}%
while $\widetilde{A_{\ast }}$ varies from $1$ to $\alpha \delta $ and $B_{0}$
varies from $0$ to $\sqrt{1-\eta _{0}^{2}\delta ^{2}}.$
\end{lemma}

\begin{proof}
All these inequalities follow from (\ref{def lambda}) and (\ref{basic ineq a}%
).
\end{proof}

\subsection{New coordinates\label{sec 3.2}}

The eigenvector and generalized eigenvector for the eigenvalue 0 of $%
\boldsymbol{L}_{\delta }$ are :%
\begin{equation*}
Z_{0}=\left( 
\begin{array}{c}
0 \\ 
0 \\ 
0 \\ 
0 \\ 
\widetilde{A_{\ast }} \\ 
0%
\end{array}%
\right) ,\text{ }Z_{1}=\left( 
\begin{array}{c}
0 \\ 
-(1+\delta ^{2})B_{0} \\ 
0 \\ 
0 \\ 
0 \\ 
\widetilde{A_{\ast }}%
\end{array}%
\right) .
\end{equation*}%
Now we denote by 
\begin{equation*}
V_{r}^{+}\pm i\lambda _{i}V_{i}^{+},\text{ }V_{r}^{-}\pm i\lambda
_{i}V_{i}^{-}
\end{equation*}%
the eigenvectors belonging respectively to the eigenvalues%
\begin{equation*}
\lambda _{r}\pm i\lambda _{i},\text{ \ }-\lambda _{r}\pm i\lambda _{i}
\end{equation*}%
then we define%
\begin{equation*}
V_{r}^{\pm }=\left( 
\begin{array}{c}
\mp \frac{\lambda _{r}(\lambda _{r}^{2}-3\lambda _{i}^{2})}{2\widetilde{%
A_{\ast }}^{2}} \\ 
1 \\ 
\pm \lambda _{r} \\ 
\lambda _{r}^{2}-\lambda _{i}^{2} \\ 
\mp \frac{\lambda _{r}(\lambda _{r}^{2}-\lambda _{i}^{2})}{(1+\delta
^{2})B_{0}\widetilde{A_{\ast }}} \\ 
-\frac{(\lambda _{r}^{2}-\lambda _{i}^{2})^{2}}{(1+\delta ^{2})B_{0}%
\widetilde{A_{\ast }}}%
\end{array}%
\right) ,\text{ \ }V_{i}^{\pm }=\left( 
\begin{array}{c}
-\frac{3\lambda _{r}^{2}-\lambda _{i}^{2}}{2\widetilde{A_{\ast }}^{2}} \\ 
0 \\ 
1 \\ 
\pm 2\lambda _{r} \\ 
-\frac{(\lambda _{r}^{2}-\lambda _{i}^{2})}{(1+\delta ^{2})B_{0}\widetilde{%
A_{\ast }}} \\ 
\mp \frac{2\lambda _{r}(\lambda _{r}^{2}-\lambda _{i}^{2})}{(1+\delta
^{2})B_{0}\widetilde{A_{\ast }}}%
\end{array}%
\right) ,
\end{equation*}%
and we define new coordinates in $%
%TCIMACRO{\U{211d} }%
%BeginExpansion
\mathbb{R}
%EndExpansion
^{6}$: $(x_{1},x_{2},y_{1},y_{2},B_{0},z_{1})$ such that

\begin{equation*}
\left( 
\begin{array}{c}
\widetilde{A_{0}} \\ 
A_{1} \\ 
A_{2} \\ 
A_{3} \\ 
0 \\ 
B_{1}%
\end{array}%
\right) =B_{0}(x_{1}V_{r}^{+}+x_{2}\lambda
_{i}V_{i}^{+}+y_{1}V_{r}^{-}+y_{2}\lambda
_{i}V_{i}^{-}+z_{0}Z_{0}+z_{1}Z_{1}).
\end{equation*}%
We observe that after eliminating $z_{0},$ we still have 6 coordinates,
including $B_{0}$ as one of the new coordinates.

\begin{remark}
We notice that we put $B_{0}$ in front of the new coordinates, as this
results from the analysis, and shorten the computations.
\end{remark}

The coordinate change is non linear in $B_{0}$, given explicitely by:%
\begin{eqnarray}
\widetilde{A_{0}} &=&-B_{0}\frac{\lambda _{r}(\lambda _{r}^{2}-3\lambda
_{i}^{2})}{2\widetilde{A_{\ast }}^{2}}(x_{1}-y_{1})-B_{0}\frac{\lambda
_{i}(3\lambda _{r}^{2}-\lambda _{i}^{2})}{2\widetilde{A_{\ast }}^{2}}%
(x_{2}+y_{2})  \notag \\
A_{1} &=&B_{0}(x_{1}+y_{1})-(1+\delta ^{2})B_{0}^{2}z_{1}  \notag \\
A_{2} &=&\lambda _{r}B_{0}(x_{1}-y_{1})+\lambda _{i}B_{0}(x_{2}+y_{2})
\label{new var 1} \\
A_{3} &=&(\lambda _{r}^{2}-\lambda _{i}^{2})B_{0}(x_{1}+y_{1})+2\lambda
_{r}\lambda _{i}B_{0}(x_{2}-y_{2})  \notag \\
0 &=&-\frac{(\lambda _{r}^{2}-\lambda _{i}^{2})}{(1+\delta ^{2})B_{0}%
\widetilde{A_{\ast }}}A_{2}+\widetilde{A_{\ast }}B_{0}z_{0},  \notag
\end{eqnarray}%
\begin{equation}
B_{1}=-\varepsilon ^{2}(1+\delta ^{2})B_{0}\frac{A_{3}}{\widetilde{A_{\ast }}%
}+\widetilde{A_{\ast }}B_{0}z_{1},  \label{equB1}
\end{equation}%
which needs to be inverted. We obtain%
\begin{eqnarray}
B_{0}x_{1} &=&\frac{(\lambda _{r}^{2}+\lambda _{i}^{2})}{4\lambda _{r}}%
\widetilde{A_{0}}+\frac{3\lambda _{r}^{2}-\lambda _{i}^{2}}{4\lambda
_{r}(\lambda _{r}^{2}+\lambda _{i}^{2})}A_{2}  \label{x1 1} \\
&&+\frac{A_{1}}{2}+\frac{(1+\delta ^{2})B_{0}}{2\widetilde{A_{\ast }}}B_{1}+%
\frac{(\lambda _{r}^{2}-\lambda _{i}^{2})}{2\widetilde{A_{\ast }}^{2}}A_{3},
\notag
\end{eqnarray}%
\begin{eqnarray}
\lambda _{i}B_{0}x_{2} &=&-\frac{(\lambda _{r}^{2}+\lambda _{i}^{2})}{4}%
\widetilde{A_{0}}-\frac{\lambda _{r}^{2}-3\lambda _{i}^{2}}{4(\lambda
_{r}^{2}+\lambda _{i}^{2})}A_{2}  \label{x2 1} \\
&&-\frac{(\lambda _{r}^{2}-\lambda _{i}^{2})}{4\lambda _{r}}\left( A_{1}+%
\frac{(1+\delta ^{2})B_{0}}{\widetilde{A_{\ast }}}B_{1}\right) +\frac{1}{%
4\lambda _{r}}\left( 1-\frac{(\lambda _{r}^{2}-\lambda _{i}^{2})^{2}}{%
\widetilde{A_{\ast }}^{2}}\right) A_{3},  \notag
\end{eqnarray}%
\begin{eqnarray}
B_{0}y_{1} &=&-\frac{(\lambda _{r}^{2}+\lambda _{i}^{2})}{4\lambda _{r}}%
\widetilde{A_{0}}-\frac{3\lambda _{r}^{2}-\lambda _{i}^{2}}{4\lambda
_{r}(\lambda _{r}^{2}+\lambda _{i}^{2})}A_{2}  \label{y1 1} \\
&&+\frac{A_{1}}{2}+\frac{(1+\delta ^{2})B_{0}}{2\widetilde{A_{\ast }}}B_{1}+%
\frac{(\lambda _{r}^{2}-\lambda _{i}^{2})}{2\widetilde{A_{\ast }}^{2}}A_{3},
\notag
\end{eqnarray}%
\begin{eqnarray}
\lambda _{i}B_{0}y_{2} &=&-\frac{(\lambda _{r}^{2}+\lambda _{i}^{2})}{4}%
\widetilde{A_{0}}-\frac{\lambda _{r}^{2}-3\lambda _{i}^{2}}{4(\lambda
_{r}^{2}+\lambda _{i}^{2})}A_{2}  \label{y2 1} \\
&&+\frac{(\lambda _{r}^{2}-\lambda _{i}^{2})}{4\lambda _{r}}\left( A_{1}+%
\frac{(1+\delta ^{2})B_{0}}{\widetilde{A_{\ast }}}B_{1}\right) -\frac{1}{%
4\lambda _{r}}\left( 1-\frac{(\lambda _{r}^{2}-\lambda _{i}^{2})^{2}}{%
\widetilde{A_{\ast }}^{2}}\right) A_{3},  \notag
\end{eqnarray}%
\begin{equation*}
B_{0}z_{1}=\frac{(\lambda _{r}^{2}-\lambda _{i}^{2})}{(1+\delta ^{2})B_{0}%
\widetilde{A_{\ast }}^{2}}A_{3}+\frac{1}{\widetilde{A_{\ast }}}%
B_{1}=\varepsilon ^{2}B_{0}(1+\delta ^{2})\frac{A_{3}}{\widetilde{A_{\ast }}%
^{2}}+\frac{1}{\widetilde{A_{\ast }}}B_{1}.
\end{equation*}%
Let us now define

\begin{eqnarray*}
X &=&\left( 
\begin{array}{c}
x_{1} \\ 
x_{2}%
\end{array}%
\right) ,\text{ }Y=\left( 
\begin{array}{c}
y_{1} \\ 
y_{2}%
\end{array}%
\right) , \\
|X| &=&\sqrt{x_{1}^{2}+x_{2}^{2}},\text{ }|Y|=\sqrt{y_{1}^{2}+y_{2}^{2}}%
\text{ (norms in }%
%TCIMACRO{\U{211d} }%
%BeginExpansion
\mathbb{R}
%EndExpansion
^{2}).
\end{eqnarray*}%
Then, for $\varepsilon $ small enough, and using (\ref{def lambda}), we
obtain the following useful estimates

\begin{lemma}
\label{Lem estimates}For $B_{0}\leq \sqrt{1-\eta _{0}^{2}\delta ^{2}},$ $%
\varepsilon ^{2}\leq \frac{3\alpha }{10},$ $\varepsilon $ small enough we
have
\end{lemma}

\begin{eqnarray}
\frac{\widetilde{A_{\ast }}^{1/2}}{2^{5/4}} &\leq &\lambda _{r},\lambda
_{i}<2^{1/4}\widetilde{A_{\ast }}^{1/2},\text{ \ }\widetilde{A_{\ast }}\geq
\alpha \delta ,  \notag \\
|\widetilde{A_{0}}| &\leq &3\frac{B_{0}}{\widetilde{A_{\ast }}^{1/2}}%
(|X|+|Y|),  \notag \\
|A_{1}| &\leq &B_{0}(|X|+|Y|)+(1+\delta ^{2})B_{0}^{2}|z_{1}|,  \notag \\
|A_{2}| &\leq &2B_{0}\widetilde{A_{\ast }}^{1/2}(|X|+|Y|),
\label{estim coord 1} \\
|A_{3}| &\leq &2B_{0}\widetilde{A_{\ast }}(|X|+|Y|),  \notag \\
|B_{1}| &\leq &2\varepsilon ^{2}(1+\delta ^{2})B_{0}^{2}(|X|+|Y|)+\widetilde{%
A_{\ast }}B_{0}|z_{1}|.  \notag
\end{eqnarray}

\subsection{System with new coordinates\label{sec 3.3}}

The system (\ref{newsyst 1}) writen in the new coordinates is computed in
Appendix \ref{App1'}. It takes the following form (quadratic and higher
order terms are not explicited)%
\begin{eqnarray}
&&x_{1}^{\prime }=f_{1}+\lambda _{r}x_{1}+\lambda _{i}x_{2}  \label{x'1 1} \\
&&+B_{1}\left[ a_{1}\widetilde{A_{0}}+c_{1}A_{2}+d_{1}A_{3}+e_{1}\frac{B_{1}%
}{B_{0}}-\frac{1}{B_{0}}x_{1}\right]  \notag \\
&&-\varepsilon ^{2}\frac{(1+\delta ^{2})(2-\delta ^{2})B_{0}}{2\widetilde{%
A_{\ast }}}\widetilde{A_{0}}^{2}-\varepsilon ^{2}\frac{(1+\delta ^{2})B_{0}}{%
2\widetilde{A_{\ast }}^{2}}\widetilde{A_{0}}^{3},  \notag
\end{eqnarray}%
\begin{eqnarray}
x_{2}^{\prime } &=&f_{2}-\lambda _{i}x_{1}+\lambda _{r}x_{2}+B_{1}\left[
-a_{2}\widetilde{A_{0}}+b_{2}A_{1}+c_{2}A_{2}+d_{2}A_{3}+e_{2}B_{1}-\frac{1}{%
B_{0}}x_{2}\right]  \label{x'2 1} \\
&&-\frac{1}{4\lambda _{r}\lambda _{i}\widetilde{A_{\ast }}B_{0}}\left( 3%
\widetilde{A_{\ast }}^{2}-2\varepsilon ^{4}B_{0}^{4}(1+\delta
^{2})^{4}\right) \widetilde{A_{0}}^{2}-\frac{1}{4\lambda _{r}\lambda
_{i}B_{0}}\left( 1-\frac{(\lambda _{r}^{2}-\lambda _{i}^{2})^{2}}{\widetilde{%
A_{\ast }}^{2}}\right) \widetilde{A_{0}}^{3}.  \notag
\end{eqnarray}%
\begin{eqnarray}
y_{1}^{\prime } &=&f_{1}-\lambda _{r}y_{1}+\lambda _{i}y_{2}+  \label{y'1 1}
\\
&&+B_{1}\left[ -a_{1}\widetilde{A_{0}}-c_{1}A_{2}+d_{1}A_{3}+e_{1}\frac{B_{1}%
}{B_{0}}-\frac{1}{B_{0}}y_{1}\right]  \notag \\
&&-\varepsilon ^{2}\frac{(1+\delta ^{2})(2-\delta ^{2})B_{0}}{2\widetilde{%
A_{\ast }}}\widetilde{A_{0}}^{2}-\varepsilon ^{2}\frac{(1+\delta ^{2})B_{0}}{%
2\widetilde{A_{\ast }}^{2}}\widetilde{A_{0}}^{3},  \notag
\end{eqnarray}%
\begin{eqnarray}
y_{2}^{\prime } &=&-f_{2}-\lambda _{i}y_{1}-\lambda _{r}y_{2}+B_{1}\left[
-a_{2}\widetilde{A_{0}}-b_{2}A_{1}+c_{2}A_{2}-d_{2}A_{3}+e_{2}B_{1}-\frac{1}{%
B_{0}}y_{2}\right]  \label{y'2 1} \\
&&+\frac{1}{4\lambda _{r}\lambda _{i}\widetilde{A_{\ast }}B_{0}}\left( 3%
\widetilde{A_{\ast }}^{2}-2\varepsilon ^{4}B_{0}^{4}(1+\delta
^{2})^{4}\right) \widetilde{A_{0}}^{2}+\frac{1}{4\lambda _{r}\lambda
_{i}B_{0}}\left( 1-\frac{(\lambda _{r}^{2}-\lambda _{i}^{2})^{2}}{\widetilde{%
A_{\ast }}^{2}}\right) \widetilde{A_{0}}^{3},  \notag
\end{eqnarray}%
with%
\begin{equation*}
f_{1}=\frac{\varepsilon ^{2}\delta ^{2}B_{0}(1+\delta ^{2})(\widetilde{%
A_{\ast }}^{2}-B_{0}^{2})}{2\widetilde{A_{\ast }}},
\end{equation*}%
\begin{equation*}
f_{2}=-\frac{\varepsilon ^{2}\delta ^{2}B_{0}(1+\delta ^{2})(\lambda
_{r}^{2}-\lambda _{i}^{2})(\widetilde{A_{\ast }}^{2}-B_{0}^{2})}{4\lambda
_{r}\lambda _{i}\widetilde{A_{\ast }}}.
\end{equation*}%
Using Lemma \ref{Lem estim lambda} and (\ref{basic ineq a}) we see that we
have the estimates%
\begin{equation*}
|f_{j}|\leq \frac{B_{0}\varepsilon ^{2}\delta }{\alpha },\text{ }j=1,2.
\end{equation*}%
The coefficients $a_{j},$ $b_{j},$ $c_{j},$ $d_{j},$ $e_{j}$ are defined and
estimated in Appendix \ref{App1'} in (\ref{a1 1},\ref{a2 1}), (\ref{b2 1},%
\ref{c1 1},\ref{c2 1}), (\ref{d1 1},\ref{d2 1}), (\ref{e1 1},\ref{e2 1}).
Here $\widetilde{A_{0}},A_{1},A_{2},A_{3},B_{1}$ should be replaced by their
(linear) expressions (\ref{new var 1}) in coordinates $%
(x_{1},x_{2},y_{1},y_{2},z_{1})$ with coefficients functions of $B_{0}.$ The
system above should be completed by the differential equations for $%
z_{1}^{\prime }$ and $B_{0}^{\prime }(=B_{1})$. In fact we replace the
equation for $z_{1}^{\prime }$ by the direct resolution of the first
integral (\ref{Wg}) with respect to $z_{1}$ using the expression of $B_{1}$
in (\ref{equB1}).

\subsection{Resolution of (\protect\ref{Wg}) with respect of $%
z_{1}(X,Y,B_{0})\label{sec 3.4}$}

For extending the validity (as a graph with respect to $B_{0})$ for the
existence of the unstable manifold of $M_{-}$ we need to replace the
differential equation for $z_{1}^{\prime }$ by the expression of $z_{1}$
given by the first integral (\ref{Wg}). This leads to the following

\begin{lemma}
\label{Lem z1}For $B_{0}\leq \sqrt{1-\eta _{0}^{2}\delta ^{2}},$ $%
\varepsilon ^{2}\leq \frac{3\alpha }{10}.$ Then choose $\rho >0$ such that
for $\varepsilon $ and $\alpha $ small enough, and for%
\begin{equation}
|\overline{X}|+|\overline{Y}|\leq \rho ,\text{ }\alpha ^{3}\rho ^{4}<<1,
\label{cond ro khi}
\end{equation}%
with the scaling%
\begin{equation}
(X,Y)=\alpha ^{3/2}\delta (\overline{X},\overline{Y}),\text{ }%
z_{1}=\varepsilon \delta \overline{z_{1}},  \label{scaling a}
\end{equation}%
we have 
\begin{equation}
\overline{z_{1}}=\overline{z_{10}}(B_{0},\delta )[1+\mathcal{Z(}\overline{X},%
\overline{Y},B_{0},\varepsilon ,\alpha ,\delta )]  \label{z1 1}
\end{equation}%
where the function $\mathcal{Z(}\overline{X},\overline{Y},B_{0},\varepsilon
,\delta )$ is analytic in its arguments, and with%
\begin{equation}
\overline{z_{10}}(B_{0},\delta )\overset{def}{=}(1+\frac{\delta ^{2}B_{0}^{2}%
}{2\widetilde{A_{\ast }}^{2}})^{1/2}\leq \frac{1}{\alpha },
\label{estim z10}
\end{equation}%
and%
\begin{equation}
|\mathcal{Z(}\overline{X},\overline{Y},B_{0},\varepsilon ,\alpha ,\delta
)|\leq c\alpha ^{3}(1+\rho ^{2})(|\overline{X}|+|\overline{Y}|)^{2},
\label{estim Z}
\end{equation}%
with $c$ independent of $(\varepsilon ,\alpha ,\rho ).$
\end{lemma}

\begin{remark}
In the Lemma above, we introduce the number $\rho $ which may be large.
Precise constraints are given later.
\end{remark}

\begin{proof}
Using (\ref{equB1}) and (\ref{Wg}) we have%
\begin{eqnarray*}
B_{1}^{2} &=&\{\widetilde{A_{\ast }}B_{0}z_{1}-\varepsilon ^{2}\frac{%
B_{0}(1+\delta ^{2})}{\widetilde{A_{\ast }}}A_{3}\}^{2}=2\varepsilon
^{2}A_{1}A_{3}-\varepsilon ^{2}A_{2}^{2}+ \\
&&\frac{\varepsilon ^{2}}{2}(-\delta ^{2}B_{0}^{2}+2\widetilde{A_{\ast }}%
\widetilde{A_{0}}+\widetilde{A_{0}}^{2})^{2}+\varepsilon ^{2}\delta ^{2}(%
\widetilde{A_{\ast }}+\widetilde{A_{0}})^{2}B_{0}^{2},
\end{eqnarray*}%
hence%
\begin{eqnarray}
\widetilde{A_{\ast }}^{2}z_{1}^{2} &=&\varepsilon ^{2}\delta ^{2}\widetilde{%
A_{\ast }}^{2}(1+\frac{\delta ^{2}B_{0}^{2}}{2\widetilde{A_{\ast }}^{2}})+%
\frac{2\varepsilon ^{2}}{B_{0}}A_{3}(x_{1}+y_{1})-\frac{\varepsilon
^{4}(1+\delta ^{2})^{2}}{\widetilde{A_{\ast }}^{2}}A_{3}^{2}-\frac{%
\varepsilon ^{2}}{B_{0}^{2}}A_{2}^{2}+  \notag \\
&&+\frac{2\varepsilon ^{2}\widetilde{A_{\ast }}^{2}}{B_{0}^{2}}\widetilde{%
A_{0}}^{2}+\frac{2\varepsilon ^{2}\widetilde{A_{\ast }}}{B_{0}^{2}}%
\widetilde{A_{0}}^{3}+\frac{\varepsilon ^{2}}{2B_{0}^{2}}\widetilde{A_{0}}%
^{4},  \label{z1^2}
\end{eqnarray}%
where we may observe on the r.h.s., that 
\begin{equation*}
\varepsilon ^{2}\delta ^{2}\leq \varepsilon ^{2}\delta ^{2}(1+\frac{\delta
^{2}B_{0}^{2}}{2\widetilde{A_{\ast }}^{2}})\leq \varepsilon ^{2}\delta
^{2}(1+\frac{1}{2\alpha ^{2}}),
\end{equation*}%
which is independent of $(X,Y).$ Moreover there is no linear part in $(X,Y)$
in (\ref{z1^2}). The scaling (\ref{scaling a}), Lemma \ref{Lem delta} and
Lemma \ref{Lem estimates} imply ($c$ is a generic constant, independent of $%
\varepsilon $)%
\begin{equation*}
|\frac{2\varepsilon ^{2}}{B_{0}}A_{3}(x_{1}+y_{1})|\leq c\varepsilon
^{2}\alpha ^{3}\widetilde{A_{\ast }}(|\overline{X}|+|\overline{Y}|)^{2}
\end{equation*}%
\begin{equation*}
|\frac{\varepsilon ^{4}(1+\delta ^{2})^{2}}{\widetilde{A_{\ast }}^{2}}%
A_{3}^{2}|\leq c\varepsilon ^{4}\alpha ^{3}(|\overline{X}|+|\overline{Y}%
|)^{2}\leq c\varepsilon ^{2}\alpha ^{3}\widetilde{A_{\ast }}(|\overline{X}|+|%
\overline{Y}|)^{2},
\end{equation*}%
\begin{equation*}
\frac{\varepsilon ^{2}}{B_{0}^{2}}A_{2}^{2}\leq c\varepsilon ^{2}\alpha ^{3}%
\widetilde{A_{\ast }}(|\overline{X}|+|\overline{Y}|)^{2})
\end{equation*}%
\begin{equation*}
\frac{2\varepsilon ^{2}\widetilde{A_{\ast }}^{2}}{B_{0}^{2}}\widetilde{A_{0}}%
^{2}\leq c\varepsilon ^{2}\alpha ^{3}\widetilde{A_{\ast }}(|\overline{X}|+|%
\overline{Y}|)^{2}
\end{equation*}%
\begin{equation*}
|\frac{2\varepsilon ^{2}\widetilde{A_{\ast }}}{B_{0}^{2}}\widetilde{A_{0}}%
^{3}|\leq c\varepsilon ^{2}\alpha ^{3}\widetilde{A_{\ast }}(|\overline{X}|+|%
\overline{Y}|)^{3}
\end{equation*}%
\begin{equation*}
\frac{\varepsilon ^{2}}{2B_{0}^{2}}\widetilde{A_{0}}^{4}\leq c\varepsilon
^{2}\alpha ^{3}\widetilde{A_{\ast }}(|\overline{X}|+|\overline{Y}|)^{4},
\end{equation*}%
so that the factors in the estimates are such that%
\begin{equation*}
\frac{c\varepsilon ^{2}\alpha ^{3}\widetilde{A_{\ast }}}{\varepsilon
^{2}\delta ^{2}\widetilde{A_{\ast }}^{2}}\leq c\frac{\alpha ^{3}}{\widetilde{%
A_{\ast }}},\text{ }
\end{equation*}%
$c$ being independent of $\varepsilon ,\alpha $ and $\delta \in \lbrack
1/3,1].$ Now defining $\overline{z_{10}}$ such that%
\begin{equation}
1\leq \overline{z_{10}}(B_{0},\delta )\overset{def}{=}(1+\frac{\delta
^{2}B_{0}^{2}}{2\widetilde{A_{\ast }}^{2}})^{1/2}\leq \frac{1}{\alpha },%
\text{ for }\alpha \leq 1/\sqrt{2},  \label{defz10}
\end{equation}%
we notice that we have%
\begin{equation*}
\frac{1}{\overline{z_{10}}^{2}}=\frac{2\widetilde{A_{\ast }}^{2}}{2%
\widetilde{A_{\ast }}^{2}+\delta ^{2}B_{0}^{2}}\leq \frac{2(1+\delta ^{2})%
\widetilde{A_{\ast }}^{2}}{\delta ^{2}}\leq 20\widetilde{A_{\ast }}^{2}.
\end{equation*}%
It results that, for $|\overline{X}|+|\overline{Y}|\leq \rho $ 
\begin{equation*}
\overline{z_{1}}^{2}=\overline{z_{10}}^{2}+\mathcal{O}\left( \frac{\alpha
^{3}(1+\rho ^{2})}{\widetilde{A_{\ast }}}(|\overline{X}|+|\overline{Y}%
|)^{2}\right) .
\end{equation*}%
It is shown in \cite{BHI} that $0\leq B_{0}\leq 1$ on the heteroclinic, so
that $B_{1}(x)=B_{0}^{\prime }(x)$ is positive at least at the starting
point $B_{0}=0.$ It results that we need to take the positive square root
for $B_{1}$ (as for (\ref{Wg}))$,$ hence also for $z_{1}$ (see (\ref{equB1}%
)):%
\begin{equation*}
\overline{z_{1}}=\overline{z_{10}}(B_{0})\left\{ 1+\mathcal{O}[\alpha
^{3}(1+\rho ^{2})(|\overline{X}|+|\overline{Y}|)^{2}]\right\} ^{1/2},\text{
for }|\overline{X}|+|\overline{Y}|\leq \rho ,
\end{equation*}%
and taking the square root, we obtain (\ref{z1 1}) with estimate (\ref{estim
Z}), $\mathcal{Z(}\overline{X},\overline{Y},B_{0},\varepsilon ,\alpha
,\delta )$ being defined in the ball $|\overline{X}|+|\overline{Y}|\leq \rho 
$, $c$ independent of $\varepsilon ,\alpha $ provided that $\varepsilon
,\,\alpha $ are small enough and $\rho $ satisfies (\ref{cond ro khi}).
Moreover $\mathcal{Z}$ is analytic in its arguments and is at least
quadratic in $(\overline{X},\overline{Y})$. Notice that, in using (\ref%
{defz10}), we also have%
\begin{equation}
\overline{z_{10}}|\mathcal{Z(}\overline{X},\overline{Y},B_{0},\varepsilon
,\delta )|\leq c(1+\rho ^{2})\alpha ^{2}(|\overline{X}|+|\overline{Y}|)^{2}.
\label{estim z10Z}
\end{equation}
\end{proof}

Since $z_{1}$ contains $\overline{z_{10}}$ which is independent of $(%
\overline{X},\overline{Y}),$ the new system for $(\overline{X},\overline{Y})$
has new "constant terms" and "linear terms", appearing as perturbations of
the former ones.

\subsection{System where $z_{1}$ is eliminated\label{sec 3.5}}

Now we stay on the 5-dimensional invariant manifold (\ref{Wg}) and we need
to express the new differential system in terms of the 5 coordinates $(%
\overline{X},\overline{Y},B_{0}).$ The new system is computed in Appendix %
\ref{app eliminationz1}. We obtain (notice that $B_{0}$ is in factor of the
"constant" terms, and all operators are $B_{0}-$ dependent)%
\begin{eqnarray}
\overline{X}^{\prime } &=&\mathbf{L}_{0}\overline{X}+B_{0}\mathcal{F}_{0}+%
\mathcal{L}_{01}(\overline{X},\overline{Y})+\mathcal{B}_{01}(\overline{X},%
\overline{Y}),  \label{newsyst XY a} \\
\overline{Y}^{\prime } &=&\mathbf{L}_{1}\overline{Y}+B_{0}\mathcal{F}_{1}+%
\mathcal{L}_{11}(\overline{X},\overline{Y})+\mathcal{B}_{11}(\overline{X},%
\overline{Y}),  \notag
\end{eqnarray}%
which should be completed by an equation for $B_{0}^{\prime }$ (see (\ref%
{equB1}) in terms of $(\overline{X},\overline{Y},B_{0})$), and where%
\begin{equation*}
\mathbf{L}_{0}=\left( 
\begin{array}{cc}
\lambda _{r} & \lambda _{i} \\ 
-\lambda _{i} & \lambda _{r}%
\end{array}%
\right) ,\text{ \ }\mathbf{L}_{1}=\left( 
\begin{array}{cc}
-\lambda _{r} & \lambda _{i} \\ 
-\lambda _{i} & -\lambda _{r}%
\end{array}%
\right) ,
\end{equation*}%
with the following estimates, for terms independent of $(\overline{X},%
\overline{Y})$%
\begin{equation}
|\mathcal{F}_{0}|+|\mathcal{F}_{1}|\leq \frac{c\varepsilon ^{2}}{\alpha
^{9/2}},  \label{estim const a}
\end{equation}%
for terms which are linear in $(\overline{X},\overline{Y})$%
\begin{equation}
|\mathcal{L}_{01}(\overline{X},\overline{Y})|+|\mathcal{L}_{11}(\overline{X},%
\overline{Y})|\leq c\frac{\varepsilon }{\alpha ^{2}}(|\overline{X}|+|%
\overline{Y}|),  \label{estim linXY a}
\end{equation}%
and for terms at least quadratic in $(\overline{X},\overline{Y}),$ and for 
\begin{equation*}
|\overline{X}|+|\overline{Y}|\leq \rho ,\text{ }\rho <<\alpha ^{-3/4},
\end{equation*}%
we obtain 
\begin{eqnarray}
|\mathcal{B}_{01}(\overline{X},\overline{Y})|+|\mathcal{B}_{11}(\overline{X},%
\overline{Y})| &\leq &\alpha ^{1/2}(9/2+c\frac{\varepsilon ^{2}}{\alpha
^{7/2}})(|\overline{X}|+|\overline{Y}|)^{2}  \label{estim quad a} \\
&&+\alpha ^{1/2}(\frac{27}{2}+c\frac{\varepsilon }{\alpha })(|\overline{X}|+|%
\overline{Y}|)^{3}+\alpha ^{1/2}(c\frac{\varepsilon ^{2}}{\alpha ^{2}})(|%
\overline{X}|+|\overline{Y}|)^{4}.  \notag
\end{eqnarray}%
We are now ready to formulate the search for the unstable manifold of $%
M_{-}. $

\subsection{Integral formulation for solutions bounded as $x\rightarrow
-\infty \label{sec 3.6}$}

Let us introduce the monodromy operators associated with the linear
operators $\mathbf{L}_{0},\mathbf{L}_{1}$ which have non constant
coefficients:%
\begin{eqnarray*}
\frac{\partial }{\partial x}S_{0}(x,s) &=&\mathbf{L}_{0}S_{0}(x,s),\text{ }%
S_{0}(x,s_{1})S_{0}(s_{1},s_{2})=S_{0}(x,s_{2}),\text{ }S_{0}(x,x)=\mathbb{I}%
, \\
\frac{\partial }{\partial x}S_{1}(x,s) &=&\mathbf{L}_{1}S_{1}(x,s),\text{ }%
S_{1}(x,s_{1})S_{1}(s_{1},s_{2})=S_{1}(x,s_{2}),\text{ }S_{1}(x,x)=\mathbb{I}%
.
\end{eqnarray*}%
The coefficients of operators $\mathbf{L}_{0},\mathbf{L}_{1}$ are functions
of $B_{0},$ so we need Lemma \ref{LemMonodromy} in Appendix \ref{App1}, with
the following estimates, valid for $0\leq B_{0}\leq \sqrt{1-\eta
_{0}^{2}\delta ^{2}},$ $\varepsilon ^{2}\leq \frac{3\alpha }{10},$ $%
\varepsilon $ small enough$:$%
\begin{eqnarray}
||\boldsymbol{S}_{0}(x,s)|| &\leq &e^{\sigma (x-s)},\text{ }-\infty <x<s\leq
-x_{\ast }\leq 0,  \label{monodrom1} \\
||\boldsymbol{S}_{1}(x,s)|| &\leq &e^{-\sigma (x-s)},\text{ }-\infty
<s<x\leq -x_{\ast }\leq 0,  \label{monodrom2}
\end{eqnarray}%
with%
\begin{equation*}
\sigma =\frac{\alpha ^{1/2}\delta ^{1/2}}{2^{1/4}}.
\end{equation*}%
We are looking for solutions of (\ref{newsyst XY a}) which stay bounded for $%
x\rightarrow -\infty .$ Then, thanks to estimates (\ref{monodrom1}) (\ref%
{monodrom2}), the system (\ref{newsyst XY a}) may be formulated for $-\infty
<x\leq -x_{\ast }\leq 0$ as%
\begin{eqnarray}
\overline{X}(x) &=&\mathbf{S}_{0}(x,-x_{\ast })\overline{X}%
_{0}-\int_{x}^{-x_{\ast }}\mathbf{S}_{0}(x,s)G_{0}(s)ds
\label{integ formulation} \\
\overline{Y}(x) &=&\int_{-\infty }^{x}\mathbf{S}_{1}(x,s)G_{1}(s)ds  \notag
\end{eqnarray}

\begin{eqnarray*}
&&G_{0}(s)\overset{def}{=}B_{0 }\mathcal{F}_{0}+\mathcal{L}_{01}(\overline{X}%
,\overline{Y})+\mathcal{B}_{01}(\overline{X},\overline{Y}), \\
&&G_{1}(s)\overset{def}{=}B_{0 }\mathcal{F}_{1}+\mathcal{L}_{11}(\overline{X}%
,\overline{Y})+\mathcal{B}_{11}(\overline{X},\overline{Y})
\end{eqnarray*}%
where $\overline{X},\overline{Y}$ and $B_{0}$ are bounded and continuous
functions of $s,$ $B_{0}$ tending towards 0 as $s\rightarrow -\infty .$

\subsection{Strategy (continued)\label{sec 3.7}}

The 3-dimensional unstable manifold of $M_{-}$ is such that $(\overline{X}%
(x),\overline{Y}(x),B_{0}(x))$ should be expressed in terms of $\overline{X}%
_{0},B_{0}(-x_{\ast }).$ The idea is then

i) solve (\ref{integ formulation}) with respect to $(\overline{X},\overline{Y%
})$ in function of $(\overline{X}_{0},B_{0})$;

ii) solve the differential equation for $B_{0}(x)$ satisfying $%
B_{0}(-x_{\ast })=B_{00},$ $B_{0}(-\infty )=0$.

The result will be valid for $B_{0}(x)$, and $B_{00}$ in the interval $[0,%
\sqrt{1-\eta _{0}^{2}\delta ^{2}}]$ and it appears that $A_{0}(x)$ is then
very close to $0$ at the end point $x=-x_{\ast }$. The hope is that this
should allow to obtain an intersection with the 3-dim stable manifold of $%
M_{+}$ which computation should be valid for $B_{0}(x)$ in the interval $[%
\sqrt{1-\eta _{0}^{2}\delta ^{2}},1]$.

\subsection{Resolution with respect to $(\overline{X},\overline{Y})\label%
{subsect resol XY}\label{sec 3.8}$}

Let us define, for $\kappa >0$ the function space 
\begin{equation*}
C_{\kappa }^{0}=\{\overline{X}\in C^{0}(-\infty ,-x_{\ast }];\overline{X}%
(x)e^{-\kappa x}\text{ is bounded}\}
\end{equation*}%
equiped with the norm%
\begin{equation*}
||\overline{X}||_{\kappa }=\sup_{(-\infty ,-x_{\ast })}|\overline{X}%
(x)e^{-\kappa x}|.
\end{equation*}

In this subsection we prove the following

\begin{lemma}
\label{Lem X,Y} Given $M>0$, for $\varepsilon =\nu \alpha ^{5/2},$ with $\nu
\ $small enough, there exists $k_{0}>0,$ independent of $\nu ,$ such that
for $||\overline{X}_{0}||\leq k_{0}$, there is a unique solution $(\overline{%
X},\overline{Y})$ in $(C_{\kappa }^{0})^{2}$ such that, for $||B||_{\kappa
}\leq M,$ we have 
\begin{eqnarray*}
||\overline{X}||_{\kappa } &\leq &|\overline{X}_{0}|(1+c\nu )+c\nu ,\text{ }%
\overline{X}(-x_{\ast })=\overline{X}_{0}, \\
||\overline{Y}||_{\kappa } &\leq &c\nu +c(\nu |\overline{X}_{0}|+|\overline{X%
}_{0}|^{2}),
\end{eqnarray*}%
where $c$ is independent of $\varepsilon ,\nu ,k_{0}.$
\end{lemma}

\begin{remark}
The choice of $\kappa $ will be in agreement with the behavior of $B_{0}(x)$
as $x\rightarrow -\infty ,$ which is studied at next subsection.
\end{remark}

\begin{proof}
First we observe that, provided that $\kappa <\sigma $%
\begin{equation*}
|\int_{-\infty }^{x}\boldsymbol{S}_{1}(x,s)e^{\kappa s}ds|\leq \frac{%
e^{\kappa x}}{\kappa +\sigma },\text{ \ }x\leq -x_{\ast },
\end{equation*}%
\begin{equation*}
|\boldsymbol{S}_{0}(x,-x_{\ast })e^{-\kappa (x+x_{\ast }})|\leq e^{(\sigma
-\kappa )(x+x_{\ast })},\text{ \ }x\leq -x_{\ast },
\end{equation*}%
\begin{equation*}
|\int_{-x_{\ast }}^{x}\mathbf{S}_{0}(x,s)e^{\kappa s}ds|\leq \frac{e^{\kappa
x}}{\sigma -\kappa },\text{ \ }x\leq -x_{\ast }.
\end{equation*}%
Let us choose%
\begin{equation*}
\kappa \leq \frac{\sigma }{2},
\end{equation*}%
then%
\begin{equation*}
|\int_{-\infty }^{x}\boldsymbol{S}_{1}(x,s)e^{\kappa s}ds|\leq \frac{%
e^{\kappa x}}{\sigma }=2^{1/4}\frac{e^{\kappa x}}{\alpha ^{1/2}\delta ^{1/2}}%
,
\end{equation*}%
\begin{equation*}
|\int_{-x_{\ast }}^{x}\mathbf{S}_{0}(x,s)e^{\kappa s}ds|\leq 2^{5/4}\frac{%
e^{\kappa x}}{\alpha ^{1/2}\delta ^{1/2}},\text{ \ }x\leq -x_{\ast }.
\end{equation*}%
Let us assume that%
\begin{equation}
||B_{0}||_{\kappa }\leq M  \label{estim normB0}
\end{equation}%
holds (needs to be proved at next subsection). We wish to use the analytic
implicit function theorem (see \cite{Dieudo} section X.2) for $(\overline{X},%
\overline{Y})$ in a neighborhood of $(\overline{X}_{0},0)$ in the function
space $(C_{\kappa }^{0})^{2}$, provided that we can choose $\kappa \leq 
\frac{\sigma }{2}$ and $||\overline{X}||_{\kappa }+||\overline{Y}||_{\kappa
}\leq \rho $. Indeed, using the above estimates for coefficients, we obtain
for $x\in (-\infty ,-x_{\ast }]$%
\begin{equation*}
|\overline{X}(x)e^{-\kappa x}|\leq |\overline{X}(-x_{\ast })|e^{\kappa
x_{\ast }}+\frac{2^{5/4}}{\alpha ^{1/2}\delta ^{1/2}}||B_{0}\mathcal{F}_{0}+%
\mathcal{L}_{01}(\overline{X},\overline{Y})+\mathcal{B}_{01}(\overline{X},%
\overline{Y})||_{\kappa },
\end{equation*}%
hence%
\begin{equation}
||\overline{X}||_{\kappa }\leq |\overline{X}(-x_{\ast })|e^{\kappa x_{\ast
}}+\frac{2^{5/4}}{\alpha ^{1/2}\delta ^{1/2}}||B_{0}\mathcal{F}_{0}+\mathcal{%
L}_{01}(\overline{X},\overline{Y})+\mathcal{B}_{01}(\overline{X},\overline{Y}%
)||_{\kappa },  \label{estimX}
\end{equation}%
and in the same way%
\begin{equation}
||\overline{Y}||_{\kappa }\leq \frac{2^{1/4}}{\alpha ^{1/2}\delta ^{1/2}}%
||B_{0}\mathcal{F}_{1}+\mathcal{L}_{11}(\overline{X},\overline{Y})+\mathcal{B%
}_{11}(\overline{X},\overline{Y})||_{\kappa }.  \label{estimY}
\end{equation}%
Using estimates (\ref{estim const a}) for $\mathcal{F}_{j},$ (\ref{estim
linXY a}) for $\mathcal{L}_{j1},$ (\ref{estim quad a}) for $\mathcal{B}%
_{01}, $ $\mathcal{B}_{11}$, (\ref{estimX}), (\ref{estimY}), we obtain, for $%
S\overset{def}{=}||\overline{X}||_{\kappa }+||\overline{Y}||_{\kappa }\leq
\rho $%
\begin{eqnarray*}
S &\leq &|\overline{X}(-x_{\ast })|e^{\kappa x_{\ast }}+c\frac{\varepsilon
^{2}M}{\alpha ^{5}}+c\frac{S\varepsilon }{\alpha ^{5/2}}+\frac{27}{%
2^{3/4}\delta ^{1/2}}(1+c\frac{\varepsilon ^{2}}{\alpha ^{7/2}})S^{2} \\
&&+(\frac{81}{2^{3/4}\delta ^{1/2}}+c\frac{\varepsilon }{\alpha })S^{3}+c%
\frac{\varepsilon ^{2}}{\alpha ^{2}}S^{4},
\end{eqnarray*}%
so that choosing 
\begin{equation}
\varepsilon =\nu \alpha ^{5/2},\text{ }\nu <1/M,  \label{cond alpha}
\end{equation}%
\begin{eqnarray*}
S(1-c\nu ) &\leq &|\overline{X}(-x_{\ast })|+c\nu +\frac{27}{2^{3/4}\delta
^{1/2}}(1+c\nu ^{2}\alpha ^{3/2})S^{2} \\
&&+\frac{81}{2^{3/4}\delta ^{1/2}}(1+c\nu \alpha ^{3/2})S^{3}+c\nu
^{2}\alpha ^{3}S^{4}.
\end{eqnarray*}%
Let us choose $k_{0}$ such that%
\begin{equation*}
27k_{0}+81k_{0}^{2}<2^{3/4}\delta ^{1/2}(1-c\nu ),
\end{equation*}%
which is satisfied for%
\begin{equation*}
k_{0}<0.13\delta ^{1/2}(1-c\nu ).
\end{equation*}%
Then the estimate above shows that for $0<k_{0}$ small enough, we can apply
the implicit function theorem (its analytic version) with respect to$(%
\overline{X},\overline{Y})$ (for $|\overline{X}(-x_{\ast })|\leq
k_{0}(1-c_{1}\nu )$, and for $\varepsilon $ and $\nu $ small enough). We
find a unique $(\overline{X},\overline{Y})\in (C_{\kappa }^{0})^{2}$ such
that $||\overline{X}||_{\kappa }+||\overline{Y}||_{\kappa }$ is close to $|%
\overline{X}(-x_{\ast })|$. Moreover for $\varepsilon $ and $\nu $ small
enough%
\begin{equation*}
S\leq |\overline{X}(-x_{\ast })|(1+c\nu )+c|\overline{X}(-x_{\ast
})|^{2}+c\nu ,
\end{equation*}%
with $c$ independent of $(\varepsilon ,\nu ,k_{0}).$ This leads finally to $%
\overline{X}$ and $\overline{Y}$ in $C_{\kappa }^{0},$ depending
analytically on $(\overline{X}_{0},B_{0})\in 
%TCIMACRO{\U{211d} }%
%BeginExpansion
\mathbb{R}
%EndExpansion
^{2}\times C_{\kappa }^{0},$ and such that 
\begin{equation}
||\overline{Y}||_{\kappa }\leq c\nu +c(\nu |\overline{X}(-x_{\ast })|+|%
\overline{X}(-x_{\ast })|^{2}),  \label{estim Ybar}
\end{equation}%
\begin{equation}
||\overline{X}||_{\kappa }\leq |\overline{X}(-x_{\ast })|(1+c\nu )+c\nu ,
\label{estim Xbar}
\end{equation}%
where $c$ is a number independent of $(\varepsilon ,\nu ,k_{0}),$ $%
\varepsilon $ small enough, $S\leq k_{0},$ and $k_{0}<\rho ,$ which is
compatible with (\ref{cond ro khi}). Lemma \ref{Lem X,Y} is proved.
\end{proof}

\subsubsection{Estimate of $\protect\nu $}

From the proof of Lemma \ref{Lem X,Y}, the constraint on $\nu $ is 
\begin{equation*}
\nu <\min \{1/||B_{0}||_{\kappa },1/c\},
\end{equation*}%
where 
\begin{equation*}
c=\frac{3.2^{1/4}}{\sqrt{\delta }}c^{\prime }
\end{equation*}%
and $c^{\prime }$ comes from the estimate (\ref{estim linXY a}) appearing on
the linear term in Appendix \ref{app eliminationz1}. From Appendix \ref{app
eliminationz1} and Lemma \ref{Lem estimates}, we see that $c^{\prime }$ is
given at main order by%
\begin{equation*}
c^{\prime }=3c(a_{j})+2c(c_{j})+2c(d_{j})\text{ , }j=1,2,
\end{equation*}%
where $c(a_{j})$ is defined in (\ref{a1 1},\ref{a2 1}), $c(c_{j})$ defined
in (\ref{c1 1},\ref{c2 1}), $c(d_{j})$ defined in (\ref{d1 1},\ref{d2 1}). A
careful checking leads to%
\begin{eqnarray*}
c(a_{j}) &=&3.2^{-1/4}(1+\delta ^{2}), \\
c(c_{j}) &=&2.2^{1/4}(1+\delta ^{2}), \\
c(d_{j}) &=&4\sqrt{2}(1+\delta ^{2}),
\end{eqnarray*}%
so that%
\begin{equation*}
2^{1/4}c^{\prime }=28.11(1+\delta ^{2}).
\end{equation*}%
Hence 
\begin{equation*}
c=84.33\frac{1+\delta ^{2}}{\sqrt{\delta }},
\end{equation*}%
which needs to be compared with $||B_{0}||_{\kappa }.$

We see at subsection \ref{sol integrodiff B*} (see (\ref{estimateB*a}))that%
\begin{equation*}
||B_{0}||_{\kappa }=\underset{(-\infty ,-x_{\ast }]}{\sup }%
B_{0}(x)e^{-\kappa x}\leq B_{0}(-x_{\ast })=\left( \frac{1-\alpha ^{2}\delta
^{2}}{1+\delta ^{2}}\right) ^{1/2}<(1+\delta ^{2})^{-1/2}.
\end{equation*}%
Finally the restriction on $(\nu )_{\max }$ is%
\begin{equation}
\frac{\nu }{\sqrt{\delta }}\leq \frac{(1+\delta ^{2})^{-1}}{84.33}.
\label{restrict nu-}
\end{equation}

\subsection{Resolution for $B_{0}\label{sol integrodiff B*}\label{sec 3.9}$}

In this subsection we finish the proof of Lemma \ref{unstablemanifM-}. It
remains to solve the last part of the system (\ref{newsyst XY a}) for $B_{0}$
with $B_{0}(-x_{\ast })=B_{00}.$

We notice from (\ref{equB1}), (\ref{z1 1}) and (\ref{estim coord 1}) that%
\begin{eqnarray*}
B_{1} &=&\varepsilon \delta \widetilde{A_{\ast }}B_{0}\overline{z_{10}}%
(B_{0})\left( 1+\mathcal{Z(}\overline{X},\overline{Y},B_{0},\varepsilon
,\delta )-\frac{\varepsilon ^{3/2}}{\overline{z_{10}}}\frac{(1+\delta ^{2})}{%
\widetilde{A_{\ast }}^{2}}\overline{A_{3}}\right) \\
\overline{A_{3}} &=&B_{0}[\varepsilon ^{2}B_{0}^{2}(1+\delta ^{2})^{2}(%
\overline{x_{1}}+\overline{y_{1}})+2\lambda _{r}\lambda _{i}(\overline{x_{2}}%
-\overline{y_{2}})], \\
\frac{\varepsilon ^{3/2}}{\overline{z_{10}}}(1+\delta ^{2})\frac{\overline{%
A_{3}}}{\widetilde{A_{\ast }}^{2}} &\leq &4\varepsilon ^{7/6}(|\overline{X}%
|+|\overline{Y}|),
\end{eqnarray*}%
so that it is clear that (see above estimates for $\mathcal{Z}$)%
\begin{equation}
B_{1}>0\text{ for }B_{0}\in (0,\sqrt{1-\eta _{0}^{2}\delta ^{2}}),|\overline{%
X}|+|\overline{Y}|\leq \rho ,  \label{B1>0}
\end{equation}%
This is coherent with the study of the linearized system near $M_{-}:$
Indeed the principal part of the differential equation for $B_{0}$ is%
\begin{equation*}
B_{0}^{\prime }=\varepsilon \delta B_{0}\widetilde{A_{\ast }}\overline{z_{10}%
}(B_{0})
\end{equation*}%
which may be integrated as%
\begin{eqnarray}
B_{0}^{2}(x) &=&\frac{1}{(1+\frac{\delta ^{2}}{2})\cosh
^{2}(x_{0}-\varepsilon \delta (x+x_{\ast }))},  \label{B*principal} \\
\cosh x_{0} &=&\frac{1}{B_{0}(-x_{\ast })(1+\frac{\delta ^{2}}{2})^{1/2}}, 
\notag
\end{eqnarray}%
which satisfies $B_{0}=0$ for $x=-\infty .$ More precisely the differential
equation for $B_{0}$ is now, after replacing $(\overline{X},\overline{Y})$
by the expression found at previous subsection,%
\begin{equation}
B_{0}^{\prime }=\varepsilon \delta \widetilde{A_{\ast }}B_{0}\overline{z_{10}%
}(B_{0})[1+f(B_{0})]  \label{equB*}
\end{equation}%
where $f(B_{0})$ is a non local analytic function of $B_{0}$ in $C_{\kappa
}^{0},$ such that, for $\varepsilon $ small enough, and using $||\overline{X}%
||_{\kappa }+||\overline{Y}||_{\kappa }\leq k_{0},$ (\ref{cond ro khi}) and (%
\ref{cond alpha}),%
\begin{equation*}
||f(B_{0})||_{\kappa }\leq c(\alpha ^{3}(1+\rho ^{2})S^{2}+\varepsilon
^{7/6}S)\leq c\varepsilon k_{0}.
\end{equation*}

\begin{remark}
We may notice that we might replace $c\varepsilon k_{0}$ in the estimate
above, by%
\begin{equation*}
\varepsilon k_{0}e^{\kappa x}\rightarrow 0\text{ as }x\rightarrow -\infty ,
\end{equation*}%
since $X$ and $Y\in C_{\kappa }^{0}.$
\end{remark}

We are looking for the solution such that $B_{0}=0$ for $x=-\infty ,$ and $%
B_{0}(-x_{\ast })\leq \sqrt{1-\eta _{0}^{2}\delta ^{2}}$ for $x=0.$ We can
rewrite (\ref{equB*}) as%
\begin{equation}
\frac{2B_{0}B_{0}^{\prime }}{B_{0}^{2}\widetilde{A_{\ast }}\overline{z_{10}}%
(B_{0})}=2\varepsilon \delta \lbrack 1+f(B_{0})].  \label{integrodiffB*}
\end{equation}%
We now introduce the variable $v:$%
\begin{equation*}
v=\frac{1-\sqrt{1-(1+\frac{\delta ^{2}}{2})B_{0}^{2}}}{1+\sqrt{1-(1+\frac{%
\delta ^{2}}{2})B_{0}^{2}}},\text{ }B_{0}^{2}=\frac{1}{1+\frac{\delta ^{2}}{2%
}}\frac{4v}{(1+v)^{2}},
\end{equation*}%
so that%
\begin{equation*}
(\ln v)^{\prime }=2\varepsilon \delta \lbrack 1+f(B_{0})].
\end{equation*}%
We observe that for $x$ running from $-\infty $ to $0,$ 
\begin{equation*}
w=\ln v\text{ is increasing from }-\infty \text{ to }w_{0}=\ln v_{0}<0.
\end{equation*}%
Now let us define $h$ continuous in its argument and such that%
\begin{eqnarray*}
h(w) &=&f(B_{0}), \\
B_{0} &=&\frac{1}{\left( 1+\frac{\delta ^{2}}{2}\right) ^{1/2}}\frac{2e^{w/2}%
}{(1+e^{w})},
\end{eqnarray*}%
and let us find an a priori estimate for the solution $B_{0}(x),$ for $x\in
(-\infty ,-x_{\ast }].$ We obtain by simple integration%
\begin{equation*}
\int_{-x_{\ast }}^{x}\frac{w^{\prime }(s)}{1+h(w)(s)}ds=2\varepsilon \delta
(x+x_{\ast }).
\end{equation*}%
For $\alpha $ small enough we have%
\begin{equation*}
1-c\varepsilon k_{0}\leq \frac{1}{1+h(w)}\leq 1+c\varepsilon k_{0},
\end{equation*}%
hence (since $w<w_{0},$ and $x<0)$%
\begin{equation*}
(w_{0}-w)(1-c\varepsilon k_{0})\leq -2\varepsilon \delta (x+x_{\ast })\leq
(w_{0}-w)(1+c\varepsilon k_{0})
\end{equation*}%
so that%
\begin{equation*}
\exp (\frac{-2\varepsilon \delta (x+x_{\ast })}{1+c\varepsilon k_{0}})\leq
e^{w_{0}-w}\leq \exp (\frac{-2\varepsilon \delta (x+x_{\ast })}{%
1-c\varepsilon k_{0}})
\end{equation*}%
and%
\begin{equation*}
v_{0}\exp (\frac{2\varepsilon \delta }{1-c\varepsilon k_{0}}(x+x_{\ast
}))\leq v(x)\leq v_{0}\exp (\frac{2\varepsilon \delta }{1+c\varepsilon k_{0}}%
(x+x_{\ast })).
\end{equation*}%
It finally results that we obtain an a priori estimate for%
\begin{equation}
B_{0}(x)=\mathcal{B}_{0}(\overline{X}_{0},B_{0}(-x_{\ast }))(x)\in C_{\kappa
}^{0},  \label{solu B*}
\end{equation}%
\begin{eqnarray}
\mathcal{B}_{0}(\overline{X}_{0},B_{0}(-x_{\ast }))(x) &=&\frac{1}{\left( 1+%
\frac{\delta ^{2}}{2}\right) ^{1/2}}\frac{2\sqrt{v(x)}}{(1+v(x))},\text{ }%
x\in (-\infty ,-x_{\ast }),  \notag \\
\frac{2\sqrt{v_{0}}\exp (\frac{\varepsilon \delta }{1-c\varepsilon k_{0}}%
(x+x_{\ast }))}{1+v_{0}\exp (\frac{2\varepsilon \delta }{1-c\varepsilon k_{0}%
}(x+x_{\ast }))} &\leq &\left( 1+\frac{\delta ^{2}}{2}\right) ^{1/2}\mathcal{%
B}_{0}  \label{estimateB*a} \\
&\leq &\frac{2\sqrt{v_{0}}\exp (\frac{\varepsilon \delta }{1+c\varepsilon
k_{0}}(x+x_{\ast }))}{1+v_{0}\exp (\frac{2\varepsilon \delta }{%
1+c\varepsilon k_{0}}(x+x_{\ast }))},  \notag
\end{eqnarray}%
\begin{equation*}
v_{0}=\frac{1-\sqrt{1-(1+\frac{\delta ^{2}}{2})B_{0}^{2}(-x_{\ast })}}{1+%
\sqrt{1-(1+\frac{\delta ^{2}}{2})B_{0}^{2}(-x_{\ast })}}<1,\text{ }\left( 1+%
\frac{\delta ^{2}}{2}\right) ^{1/2}B_{0}(-x_{\ast })=\frac{2\sqrt{v_{0}}}{%
1+v_{0}}.
\end{equation*}%
It remains to notice that we can choose in the proof for $(\overline{X},%
\overline{Y})$%
\begin{equation}
\kappa =\frac{\varepsilon \delta }{1+c\varepsilon k_{0}},
\label{choice of kappa}
\end{equation}%
which needs to satisfy%
\begin{equation}
\kappa \leq \frac{\sigma }{2}=\frac{\alpha ^{1/2}\sqrt{\delta }}{2^{5/4}}.
\label{cond monodrom}
\end{equation}%
We have already chosen $\varepsilon =\nu \alpha ^{\frac{5}{2}}$ hence, for $%
\alpha $ small enough, the choice (\ref{choice of kappa}) leads to%
\begin{equation*}
\kappa <\varepsilon \delta =\delta \nu \alpha ^{\frac{5}{2}}\leq \frac{%
\alpha ^{1/2}\sqrt{\delta }}{2^{5/4}}
\end{equation*}%
and (\ref{cond monodrom}) is satisfied. The \emph{a priori estimate} (\ref%
{estimateB*a}) for $B_{0}$ allows to prove that there is a unique solution
of the differential equation (\ref{integrodiffB*}) which may be extended on
the whole interval $(-\infty ,-x_{\ast }]$, and which satisfies the estimate
(\ref{estimateB*a}) (see for example \cite{Hale}). Since $B_{0}$ is in
factor in $\widetilde{A_{0}},A_{1},A_{2},A_{3},B_{1}$ the behavior for $%
x\rightarrow -\infty $ of the coordinates of the unstable manifold, is
governed by the behavior of $B_{0}.$ The estimates indicated in Lemma \ref%
{unstablemanifM-} results from (\ref{estim coord 1}), (\ref{scaling a}), (%
\ref{defz10}), (\ref{estim Xbar}), (\ref{estim Ybar}), (\ref{cond monodrom})
with $\kappa =\varepsilon \delta ^{\ast }$. This ends the proof of Lemma \ref%
{unstablemanifM-}. The part of Corollary \ref{Corollary unstab}
corresponding to $x\in (-\infty ,-x_{\ast }]$ follows from the estimates
found at Lemma \ref{unstablemanifM-}.

For making a difference with the side treated for the stable manifold, from
now on we write $\alpha _{-}$ and $\nu _{-}$ in place of $\alpha ,\nu .$ Let
us define the hyperplane $H_{0}$ 
\begin{equation*}
B_{0}=B_{00}=\left( \frac{1-\alpha _{-}^{2}\delta ^{2}}{1+\delta ^{2}}%
\right) ^{1/2}.
\end{equation*}

\subsection{Intersection of the unstable manifold with $H_{0}\label{sec 3.10}
$}

We need to give precisely the intersection of the unstable manifold with the
hyperplane $B_{0}=B_{00}.$ This gives a two-dimensional manifold lying in
the 4-dimensional manifold $\mathcal{W}_{\varepsilon ,\delta }\cap H_{0}.$
Taking into account of 
\begin{eqnarray*}
\widetilde{A_{\ast }} &=&\delta \alpha _{-} \\
\lambda _{r},\lambda _{i} &\sim &\frac{\delta ^{1/2}\alpha _{-}^{1/2}}{%
2^{1/4}},\text{ }\varepsilon =\nu _{-}\alpha _{-}^{\frac{5}{2}},\text{ }\nu
_{-}<\frac{1}{B_{00}}, \\
\overline{z_{10}} &\sim &\frac{B_{00}}{\alpha _{-}\sqrt{2}},\text{ }%
B_{00}\sim \frac{1}{\sqrt{1+\delta ^{2}}}, \\
|\overline{Y}(-x_{\ast })| &=&\mathcal{O}(|\overline{X_{0}}|^{2}+\nu _{-}|%
\overline{X_{0}}|+\nu _{-}^{2}B_{00}),\text{ \ }|\overline{X_{0}}|\leq k_{0},%
\text{, }\overline{X_{0}}=\overline{X(-x_{\ast })},
\end{eqnarray*}%
we obtain a two-dimensional intersection which is tangent to a plane
(parameters $\overline{x_{1}},\overline{x_{2}}$), given by the following
(using (\ref{new var 1}) and Lemma \ref{Lem z1})

\begin{lemma}
\label{Lem tgunstabmanif} For $1/3\leq \delta \leq 1,$ $\varepsilon $ small
enough, $\varepsilon =\nu _{-}\alpha _{-}^{5/2},$ the 2-dimensional
intersection of the 3-dimensional plane tangent to the unstable manifold,
with the 5-dimensional hyperplane $H_{0},$ satisfies the system (rescaled
parameters are $(\overline{x_{1}},\overline{x_{2}})=\frac{\delta ^{1/2}}{%
B_{00}}\overline{X}_{0}$ with $|\overline{X}_{0}|\leq k_{0})$
\end{lemma}

\begin{eqnarray}
A_{0} &=&\delta \alpha _{-}+\frac{\delta \alpha _{-}}{2^{3/4}}(\overline{%
x_{1}}-\overline{x_{2}})+\mathcal{O}(\alpha _{-}\nu _{-}|(\overline{x_{1}},%
\overline{x_{2}})|+\alpha _{-}\nu _{-}^{2})  \notag \\
A_{1} &=&(\delta \alpha _{-})^{3/2}\overline{x_{1}}-\frac{\alpha
_{-}^{2}\delta }{\sqrt{2}}B_{00}+\mathcal{O}(\alpha _{-}^{3/2}\nu _{-}|(%
\overline{x_{1}},\overline{x_{2}})|+\alpha _{-}^{3/2}\nu_{-} ^{2})  \notag \\
A_{2} &=&\frac{(\delta \alpha _{-})^{2}}{2^{1/4}}(\overline{x_{1}}+\overline{%
x_{2}})+\mathcal{O}(\alpha _{-}^{2}\nu _{-}|(\overline{x_{1}},\overline{x_{2}%
})|+\alpha _{-}^{2}\nu _{-}^{2})  \label{trace tg unstmanif} \\
A_{3} &=&\sqrt{2}(\delta \alpha _{-})^{5/2}\overline{x_{2}}+\mathcal{O}%
(\alpha _{-}^{5/2}\nu _{-}|(\overline{x_{1}},\overline{x_{2}})|+\alpha
_{-}^{5/2}\nu _{-}^{2})  \notag \\
B_{00} &\sim &(1+\delta ^{2})^{-1/2},  \notag
\end{eqnarray}%
with%
\begin{eqnarray*}
(|\overline{x_{1}}|^{2}+|\overline{x_{2}}|^{2})^{1/2} &\leq &k_{0}[\delta
(1+\delta ^{2})]^{1/2},\text{ \ }1/3\leq \delta \leq 1,\text{ }\varepsilon
=\nu _{-}\alpha _{-}^{5/2}, \\
\alpha _{-}^{2}\delta ^{2} &=&1-B_{00}^{2}(1+\delta ^{2})>0,
\end{eqnarray*}%
and where we do not write $B_{1}$ since we know that we sit on the 5
dimensional manifold $\mathcal{W}_{\varepsilon ,\delta }.$

\section{Stable manifold of $M_{+}\label{sect stabmanifold}$}

Assuming some estimates which need to be checked at the end, the strategy
here is to first solve with respect to $B_{0}$ in using the first integral (%
\ref{Wg}), and an implicit function argument. Hence $B_{0}$ becomes function
of $(A_{j}(x),B_{0}(-x_{\ast })),$ $j=0,1,2,3$ defined on $(-x_{\ast
},+\infty ).$ Afterwards for the search of the stable manifold of $M_{+}$
(with only 2 remaining dimensions), we need to solve a 4-dimensional system,
with variable coefficients on $(-x_{\ast },+\infty ).$

\subsection{Using the first integral (\protect\ref{Wg})\label{sec 4.1}}

The first integral (\ref{Wg}) is solved with respect to $B_{0}(x).$ Since it
is shown in \cite{BHI} that $B_{0}(x)\in \lbrack 0,1]$ and $B_{0}(+\infty
)=1 $ for the heteroclinic solution of (\ref{truncated syst}) we take the
positive square-root:%
\begin{equation*}
B_{0}^{\prime }=\frac{\varepsilon }{\sqrt{2}}%
[(1-B_{0}^{2})^{2}+A_{0}^{2}(A_{0}^{2}+2\delta
^{2}B_{0}^{2}+2(B_{0}^{2}-1))-2A_{2}^{2}+4A_{1}A_{3}]^{1/2},
\end{equation*}%
which shows that $B_{0}$ is growing as soon as the bracket does not cancel.
Let us define%
\begin{equation}
-1\leq v=B_{0}-1\leq 0,  \label{def v}
\end{equation}%
then we obtain%
\begin{equation}
v^{\prime }=\frac{-\varepsilon v}{\sqrt{2}}(2+v)\left( 1+\frac{%
A_{0}^{2}(A_{0}^{2}+2\delta
^{2}B_{0}^{2}+2(B_{0}^{2}-1))-2A_{2}^{2}+4A_{1}A_{3}}{(1-B_{0}^{2})^{2}}%
\right) ^{1/2}.  \label{equa v'}
\end{equation}%
For any $\kappa \geq 0$ let us introduce a function space adapted for this
section%
\begin{equation*}
C_{\kappa }^{0}=\{X\in C^{0}(-x_{\ast },+\infty );X(x)e^{\kappa x}\text{
bounded}\},
\end{equation*}%
equiped with the norm%
\begin{equation*}
||X||_{\kappa }=\underset{(-x_{\ast },\infty )}{\sup }|X(x)e^{\kappa x}|.
\end{equation*}%
For the connection with previous section, we take%
\begin{equation}
B_{0}^{2}(-x_{\ast })=\frac{1-\alpha _{-}^{2}\delta ^{2}}{1+\delta ^{2}},%
\text{ }\varepsilon =\nu _{-}\alpha _{-}^{5/2}.  \label{hyp B0+}
\end{equation}%
Then we prove the following

\begin{lemma}
\label{Lem v(x)}For $1/3\leq \delta \leq 1,$ $\kappa >\varepsilon ,$ assume (%
\ref{hyp B0+}) holds and assume that there exists $\gamma $ such that%
\begin{eqnarray}
|A_{j}(x)| &\leq &\gamma |v(x)|,\text{ }j=0,1,2,3,  \label{Hyp Aj} \\
x &\in &(-x_{\ast },+\infty ),\text{ }\gamma \leq 1/5.  \notag
\end{eqnarray}%
Then there exists a unique solution $v=\mathcal{V}(A_{0},A_{1},A_{2},A_{3})%
\in C^{1}(-x_{\ast },+\infty )$ of (\ref{equa v'}). $\mathcal{V}$ depends
analytically on $A_{j}\in C_{\kappa }^{0},$ $j=0,1,2,3,$ with $B_{0}=1+%
\mathcal{V},$%
\begin{eqnarray}
\mathcal{V}(A_{0},A_{1},A_{2},A_{3}) &=&\mathcal{V}_{0}+\mathcal{H}%
(A_{0},A_{1},A_{2},A_{3}),\text{ }  \notag \\
||\mathcal{H}(A_{0},A_{1},A_{2},A_{3})||_{\kappa } &\leq &c\frac{\varepsilon 
}{\kappa }||(A_{0},A_{1},A_{2},A_{3})||_{\kappa }^{2},
\label{prop solu v(x)} \\
\mathcal{V}_{0}(x) &=&\frac{(1-\sqrt{1+\delta ^{2}})[1-\tanh (\varepsilon 
\sqrt{2}x)]}{\sqrt{1+\delta ^{2}}+\tanh (\varepsilon \sqrt{2}x)}<0,\text{ } 
\notag \\
1+\mathcal{V}_{0}(0) &=&\frac{1}{\sqrt{1+\delta ^{2}}},\text{ }\mathcal{V}%
_{0}(-x_{\ast })=B_{0}(-x_{\ast })-1<0,  \notag
\end{eqnarray}%
and $x_{\ast }$ defined in (\ref{def x*}), is such that $B_{0}|_{x=0}=(1+%
\delta ^{2})^{-1/2}$. Moreover we have%
\begin{equation}
\frac{v(-x_{\ast })[1-\tanh (\frac{3\varepsilon (x+x_{\ast })}{4\sqrt{2}})]}{%
1+B_{0}(-x_{\ast })\tanh (\frac{3\varepsilon (x+x_{\ast })}{4\sqrt{2}})}\leq
v(x)\leq \frac{v(-x_{\ast })[1-\tanh (\frac{5\varepsilon (x+x_{\ast })}{4%
\sqrt{2}})]}{1+B_{0}(-x_{\ast })\tanh (\frac{5\varepsilon (x+x_{\ast })}{4%
\sqrt{2}})}.  \label{apriori estim v(x)}
\end{equation}
\end{lemma}

\begin{remark}
Later we need to check that (\ref{Hyp Aj}) is indeed satisfied for the
stable manifold.
\end{remark}

\begin{proof}
Let us assume that (\ref{Hyp Aj}) holds, and using (\ref{hyp B0+}) for $%
\varepsilon $ small enough, we have 
\begin{equation*}
A_{0}^{2}+2\delta ^{2}B_{0}^{2}+2(B_{0}^{2}-1)<3,
\end{equation*}%
so that 
\begin{equation*}
|A_{0}^{2}(A_{0}^{2}+2\delta
^{2}B_{0}^{2}+2(B_{0}^{2}-1))-2A_{2}^{2}+4A_{1}A_{3}|\leq 9\gamma
^{2}|v|^{2}.
\end{equation*}%
Now, by (\ref{def v})%
\begin{equation*}
(1-B_{0}^{2})=|v|(2+v)>|v|,
\end{equation*}%
hence 
\begin{equation*}
\left\vert \frac{A_{0}^{2}(A_{0}^{2}+2\delta
^{2}B_{0}^{2}+2(B_{0}^{2}-1))-2A_{2}^{2}+4A_{1}A_{3}}{(1-B_{0}^{2})^{2}}%
\right\vert \leq 9\gamma ^{2}<\frac{1}{2},
\end{equation*}%
and the square root is analytic in $(v,A_{0},A_{1},A_{2},A_{3})$ leading to%
\begin{equation}
v^{\prime }=-\varepsilon \sqrt{2}v(1+\frac{1}{2}v)[1+\mathcal{Z}%
(v,A_{0},A_{1},A_{2},A_{3})],\text{ \ }||\mathcal{Z}||_{\kappa }\leq 1/4,
\label{diffequ v}
\end{equation}%
with%
\begin{equation}
||\mathcal{Z}(v,A_{0},A_{1},A_{2},A_{3})||_{\kappa }\leq
c||(A_{0},A_{1},A_{2},A_{3})||_{\kappa }^{2}.  \label{estim Z+}
\end{equation}%
Then we can integrate the differential equation, as in section \ref{sol
integrodiff B*}. We introduce the new variable $w$ as%
\begin{equation*}
w^{\prime } =\frac{2v^{\prime }}{v(2+v)}, \text{  }
w =\ln \left( \frac{-v}{1+v/2}\right) , \text{  }
v =-\frac{e^{w}}{1+\frac{1}{2}e^{w}},
\end{equation*}%
$w$ decreases from $w_{0}$ to $-\infty $ for $x\in (-x_{\ast },\infty ),$
while $v$ grows from $v_{0}=v(-x_{\ast })<0$ to $0.$ Then, defining $%
h(w,A_{0},A_{1},A_{2},A_{3})\overset{def}{=}\mathcal{Z}%
(v,A_{0},A_{1},A_{2},A_{3})$, we obtain, by simple integration%
\begin{equation*}
\varepsilon \sqrt{2}(x+x_{\ast })(1-1/4)\leq w_{0}-w(x)\leq \varepsilon 
\sqrt{2}(x+x_{\ast })(1+1/4),
\end{equation*}%
from which we deduce the estimate (\ref{apriori estim v(x)}). This a priori
estimate for $v$ allows to prove (see for example \cite{Hale}) the existence
and uniqueness of a solution for (\ref{diffequ v}) on the whole interval $%
x\in \lbrack -x_{\ast },\infty )$. We define $x_{\ast }$ in choosing to
satisfy%
\begin{equation*}
B_{0}(0)=\frac{1}{\sqrt{1+\delta ^{2}}},
\end{equation*}%
which gives the expression of $\mathcal{V}_{0}(x)$ given in (\ref{prop solu
v(x)}) and%
\begin{equation}
x_{\ast }\sim \frac{\sqrt{1+\delta ^{2}}}{2\sqrt{2}}\frac{\alpha _{-}^{2}}{%
\varepsilon }=\frac{\sqrt{1+\delta ^{2}}}{2\nu _{-}^{4/5}\sqrt{2}}%
\varepsilon ^{-1/5}.  \label{def x*}
\end{equation}%
The estimate in (\ref{prop solu v(x)}) results from (\ref{estim Z+}). Lemma %
\ref{Lem v(x)} is proved.
\end{proof}

\subsection{About the interval $(-x_{\ast },x_{\ast }^{+})$\label{sect x*}}

For the first part of the proof for the stable manifold, we do not start
from $x=-x_{\ast },$ but from $x=x_{\ast }^{+}$ for which $B_{0}(x_{\ast
}^{+})$ satisfies%
\begin{equation}
B_{0}^{2}(x_{\ast }^{+})=\frac{1+\alpha _{+}^{2}\delta ^{2}}{1+\delta ^{2}},
\label{initial cond B0+}
\end{equation}%
where $\alpha _{+}$ is defined at next section, with $\nu _{+}\neq \nu _{-}$
and%
\begin{equation*}
\varepsilon =\nu _{+}\alpha _{+}^{5/2}.
\end{equation*}%
In fact we obtain below an estimate on $\nu _{+}$ which is much better than
for $\nu _{-}$ and this will allow to extend later the existence of the
stable manifold on the interval 
\begin{equation*}
\left( \frac{1-\alpha _{-}^{2}\delta ^{2}}{1+\delta ^{2}}\right) ^{1/2}\leq
B_{0}\leq \left( \frac{1+\alpha _{+}^{2}\delta ^{2}}{1+\delta ^{2}}\right)
^{1/2}.
\end{equation*}%
For finding the heteroclinic we need to connect $B_{0}(x)$ found for the
unstable manifold until its upper limit 
\begin{equation*}
B_{0}(-x_{\ast })=\left( \frac{1-\alpha _{-}^{2}\delta ^{2}}{1+\delta ^{2}}%
\right) ^{1/2},
\end{equation*}%
with $B_{0}(x)$ found at Lemma \ref{Lem v(x)} valid until the same limit of $%
B_{0}(-x_{\ast })$ (now the lower limit)$.$ Then we have%
\begin{equation*}
B_{0}(x_{\ast }^{+})\sim \frac{1+\sqrt{1+\delta ^{2}}\tanh (\varepsilon 
\sqrt{2}x_{\ast }^{+})}{\sqrt{1+\delta ^{2}}+\tanh (\varepsilon \sqrt{2}%
x_{\ast }^{+})}
\end{equation*}%
which gives%
\begin{equation}
x_{\ast }^{+}\sim \frac{\sqrt{(1+\delta ^{2})}}{2\sqrt{2}}\frac{\alpha
_{+}^{2}}{\varepsilon }\sim \frac{\sqrt{1+\delta ^{2}}}{2\nu _{+}^{4/5}\sqrt{%
2}}\varepsilon ^{-1/5},  \label{x*+}
\end{equation}%
where $x_{\ast }^{+}$ is a priori different from $x_{\ast }$ since $\nu
_{+}\neq \nu _{-}$ and $B_{0}(x)$ is not invariant under the change $%
x\mapsto -x.$

\subsection{Formulation with new coordinates\label{coord stab manif}}

Here again, the strategy is first in section \ref{coord stab manif}, to
write the system (\ref{truncated syst}) in adapted coordinates, still $%
B_{0}- $ dependent, where $B_{0}$ satisfies Lemma \ref{Lem v(x)}. The
remaining two stable directions lead, for the search of the stable manifold,
to a system solved for $x\in \lbrack x_{\ast }^{+},+\infty )$ in section \ref%
{sec 4.4}. Finally we give nearly explicitely the 2-dimensional intersection
of the tangent plane of the stable manifold built on $x\in \lbrack x_{\ast
}^{+},+\infty ),$ with the hyperplane $B_{0}=\left( \frac{1+\alpha
_{+}^{2}\delta ^{2}}{1+\delta ^{2}}\right) ^{1/2}.$

As in section \ref{sect unstable manif new}, the new coordinates are such
that we are able to use monodromy operators with easy estimates in the
formulation of the search for the 3-dimensional stable manifold of $M_{+}.$

For 
\begin{equation*}
(1+\delta ^{2})B_{0}^{2}-1\geq \alpha _{+}^{2}\delta ^{2},\text{ }%
\varepsilon =\nu _{+}\alpha _{+}^{5/2},\text{ }\nu _{+}>0,
\end{equation*}%
let us define 
\begin{equation}
\widetilde{\delta }=\{(1+\delta ^{2})B_{0}^{2}-1\}^{1/4}\geq (\alpha
_{+}\delta )^{1/2},  \label{ineg deltatilde}
\end{equation}%
and choose a new basis (function of $B_{0})$%
\begin{equation*}
V_{r}^{\pm }=\left( 
\begin{array}{c}
1 \\ 
\pm \frac{\widetilde{\delta }}{\sqrt{2}} \\ 
0 \\ 
\mp \frac{\widetilde{\delta }^{3}}{\sqrt{2}} \\ 
0 \\ 
0%
\end{array}%
\right) ,\text{ }V_{i}^{\pm }=\left( 
\begin{array}{c}
0 \\ 
\pm \frac{\widetilde{\delta }}{\sqrt{2}} \\ 
\widetilde{\delta }^{2} \\ 
\pm \frac{\widetilde{\delta }^{3}}{\sqrt{2}} \\ 
0 \\ 
0%
\end{array}%
\right) ,
\end{equation*}%
\begin{equation*}
W_{1}^{-}=\left( 
\begin{array}{c}
0 \\ 
0 \\ 
0 \\ 
0 \\ 
1 \\ 
-\varepsilon \sqrt{2}%
\end{array}%
\right) ,\text{ }W_{1}^{+}=\left( 
\begin{array}{c}
0 \\ 
0 \\ 
0 \\ 
0 \\ 
1 \\ 
\varepsilon \sqrt{2}%
\end{array}%
\right) ,
\end{equation*}%
for defining new coordinates $(x_{1},x_{2},y_{1},y_{2},z_{0},z_{1})$ such
that%
\begin{eqnarray}
Z &=&(0,0,0,0,1,0)^{t}  \label{newcoord+} \\
&&+x_{1}V_{r}^{-}+x_{2}V_{i}^{-}+y_{1}V_{r}^{+}+y_{2}V_{i}^{+}+z_{0}W_{1}^{-}+z_{1}W_{1}^{+}
\notag
\end{eqnarray}%
hence 
\begin{eqnarray}
A_{0} &=&x_{1}+y_{1}  \notag \\
A_{1} &=&-\frac{\widetilde{\delta }}{\sqrt{2}}(x_{1}-y_{1}+x_{2}-y_{2}) 
\notag \\
A_{2} &=&\widetilde{\delta }^{2}(x_{2}+y_{2})  \label{variables M+} \\
A_{3} &=&\frac{\widetilde{\delta }^{3}}{\sqrt{2}}(x_{1}-y_{1}-x_{2}+y_{2}) 
\notag \\
B_{0} &=&1+v  \notag
\end{eqnarray}%
which is easy to invert, where $v$ is given by Lemma \ref{Lem v(x)}, and
coordinates $z_{0},z_{1}$ are not used, replaced by the use of $v$. Now the
system (\ref{truncated syst}) reads as%
\begin{eqnarray}
A_{0}^{\prime } &=&A_{1},  \notag \\
A_{1}^{\prime } &=&A_{2},  \notag \\
A_{2}^{\prime } &=&A_{3},  \label{syst stabmanifa} \\
A_{3}^{\prime } &=&-A_{0}[A_{0}^{2}+\widetilde{\delta }^{4}(v)],  \notag
\end{eqnarray}%
where we know (and use) that 
\begin{equation*}
v^{\prime }=-\varepsilon \sqrt{2}v(1+\frac{1}{2}v)[1+\mathcal{Z}%
(v,A_{0},A_{1},A_{2},A_{3})].
\end{equation*}%
With new variables defined in (\ref{variables M+}) this leads to the new
4-dimensional system%
\begin{eqnarray*}
X^{\prime } &=&-\boldsymbol{L}X+G_{0}(\varepsilon ,v,X,Y), \\
Y^{\prime } &=&\boldsymbol{L}Y+G_{1}(\varepsilon ,v,X,Y),
\end{eqnarray*}%
\begin{equation*}
X=\left( 
\begin{array}{c}
x_{1} \\ 
x_{2}%
\end{array}%
\right) ,\text{ }Y=\left( 
\begin{array}{c}
y_{1} \\ 
y_{2}%
\end{array}%
\right) ,\text{ }\boldsymbol{L}=\frac{\widetilde{\delta }(v)}{\sqrt{2}}%
\left( 
\begin{array}{cc}
1 & 1 \\ 
-1 & 1%
\end{array}%
\right) ,
\end{equation*}%
\begin{eqnarray*}
G_{0}(\varepsilon ,v,X,Y) &=&\frac{(x_{1}+y_{1})^{3}}{2\sqrt{2}\widetilde{%
\delta }^{3}}V_{0}+\frac{\varepsilon }{\widetilde{\delta }^{4}}\frac{%
(1+v)(\delta ^{2}-\widetilde{\delta }^{4})}{4\sqrt{2}}[1+\mathcal{Z}(v,X,Y)](%
\boldsymbol{M}_{1}X+\boldsymbol{M}_{2}Y), \\
G_{1}(\varepsilon ,v,X,Y) &=&-\frac{(x_{1}+y_{1})^{3}}{2\sqrt{2}\widetilde{%
\delta }^{3}}V_{0}+\frac{\varepsilon }{\widetilde{\delta }^{4}}\frac{%
(1+v)(\delta ^{2}-\widetilde{\delta }^{4})}{4\sqrt{2}}[1+\mathcal{Z}(v,X,Y)](%
\boldsymbol{M}_{1}Y+\boldsymbol{M}_{2}X),
\end{eqnarray*}%
\begin{equation*}
V_{0}=\left( 
\begin{array}{c}
-1 \\ 
1%
\end{array}%
\right) ,\text{ }\boldsymbol{M}_{1}=\left( 
\begin{array}{cc}
-2 & 1 \\ 
1 & -4%
\end{array}%
\right) ,\text{ }\boldsymbol{M}_{2}=\left( 
\begin{array}{cc}
2 & -1 \\ 
-1 & 0%
\end{array}%
\right) ,
\end{equation*}%
where $\mathcal{Z}(v,X,Y)$ is defined in (\ref{diffequ v}), and where we used%
\begin{equation*}
\widetilde{\delta }^{\prime }=\frac{(1+v)(\delta ^{2}-\widetilde{\delta }%
^{4})}{2\widetilde{\delta }^{3}}\frac{\varepsilon }{\sqrt{2}}(1+\mathcal{Z}).
\end{equation*}%
The above system is completed by the expression of $v$ given by Lemma \ref%
{Lem v(x)}. Notice that the coefficients of the linear part in $(X,Y)$ are
functions of $v$, where the expected part, which has the factor $\frac{%
\widetilde{\delta }(v)}{\sqrt{2}},$ is perturbed by a linear part bounded by 
$\mathcal{O}(\frac{\varepsilon }{\alpha _{+}^{2}})=\mathcal{O}(\nu
_{+}\alpha _{+}^{1/2}).$ Below, we choose $\nu _{+}$ such that the
perturbed part is really a perturbation of the first part.

For finding the stable manifold of $M_{+}$ we put the system in an integral
form, looking for solutions tending to $0$ as $x\rightarrow +\infty .$ With $%
x\geq x_{\ast }^{+}$, we obtain the system 
\begin{eqnarray}
X(x) &=&S_{0}(x,x_{\ast }^{+})X_{0}+\int_{x_{\ast
}^{+}}^{x}S_{0}(x,s)G_{0}[\varepsilon ,v(s),X(s),Y(s)]ds,  \notag \\
Y(x) &=&-\int_{x}^{+\infty }S_{1}(x,s)G_{1}[\varepsilon ,v(s),X(s),Y(s)]ds,
\label{stabmanifXY}
\end{eqnarray}%
where we notice that%
\begin{equation}
S_{0}(x,s)=e^{-\int_{s}^{x}\frac{\widetilde{\delta }(v(\tau ))d\tau }{\sqrt{2%
}}}\left( 
\begin{array}{cc}
\cos \int_{s}^{x}\frac{\widetilde{\delta }d\tau }{\sqrt{2}} & -\sin
\int_{s}^{x}\frac{\widetilde{\delta }d\tau }{\sqrt{2}} \\ 
\sin \int_{s}^{x}\frac{\widetilde{\delta }d\tau }{\sqrt{2}} & \cos
\int_{s}^{x}\frac{\widetilde{\delta }d\tau }{\sqrt{2}}%
\end{array}%
\right) ,  \label{S0+}
\end{equation}%
\begin{equation}
S_{1}(x,s)=e^{\int_{s}^{x}\frac{\widetilde{\delta }(v(\tau ))d\tau }{\sqrt{2}%
}}\left( 
\begin{array}{cc}
\cos \int_{s}^{x}\frac{\widetilde{\delta }d\tau }{\sqrt{2}} & \sin
\int_{s}^{x}\frac{\widetilde{\delta }d\tau }{\sqrt{2}} \\ 
-\sin \int_{s}^{x}\frac{\widetilde{\delta }d\tau }{\sqrt{2}} & \cos
\int_{s}^{x}\frac{\widetilde{\delta }d\tau }{\sqrt{2}}%
\end{array}%
\right) ,  \label{S1+}
\end{equation}%
and, using (\ref{ineg deltatilde})%
\begin{eqnarray*}
||S_{0}(x,s)|| &\leq &e^{-\sqrt{\frac{\alpha _{+}\delta }{2}}(x-s)},\text{ }%
x_{\ast }^{+}\leq s\leq x<\infty , \\
||S_{1}(x,s)|| &\leq &e^{-\sqrt{\frac{\alpha _{+}\delta }{2}}(s-x)},\text{ \ 
}x_{\ast }^{+}\leq x\leq s<\infty .
\end{eqnarray*}%
The 3-dimensional stable manifold is obtained in expressing $%
(X(x),Y(x),B_{0}(x))$ as a function of $(X_{0},B_{0}(x_{\ast }^{+}))$
solving (\ref{stabmanifXY}), where $B_{0}=1+v$ is given by Lemma \ref{Lem
v(x)} and $B_{0}(x_{\ast }^{+})$ has the lower bound (\ref{initial cond B0+}%
).

\subsection{The stable manifold for $x\in \lbrack x_{\ast }^{+},+\infty ) 
\label{sec 4.4}$}

We show the following

\begin{lemma}
\label{Lemma stable manifM+}For $1/3\leq \delta \leq 1,$ $0<k_{1}<\frac{1}{%
\sqrt{2}}$, $\nu _{+}>0$ small enough, and for $\varepsilon $ small enough,
the 3-dimensional stable manifold $\mathcal{W}_{\varepsilon ,\delta }^{(s)}$
of $M_{+}$ exists for $x\in \lbrack x_{\ast }^{+},+\infty ),$ is included in
the 5-dimensional manifold $\mathcal{W}_{\varepsilon ,\delta },$ is analytic
in $(\varepsilon ,\delta )$, parameterized by $(X_{0},B_{0}(x_{\ast }^{+}))$
where $X(x)$ is a 2-dimensional coordinate defined in (\ref{newcoord+}), and 
$X_{0}=X(x_{\ast }^{+})$. Moreover choosing $\delta _{\ast }$such that 
\begin{equation*}
\delta _{\ast }=\frac{1}{10}\delta ^{2/5}
\end{equation*}%
we have 
\begin{eqnarray*}
\left( \frac{1+\alpha _{+}^{2}\delta ^{2}}{1+\delta ^{2}}\right) ^{1/2}
&\leq &B_{0}(x_{\ast }^{+})\leq B_{0}(x)\leq 1,\text{ }B_{0}^{\prime }(x)>0,%
\text{ }x\in \lbrack x_{\ast }^{+},+\infty ), \\
A_{j}(x) &=&\mathcal{O}(|X_{0}|e^{-\delta _{\ast }\varepsilon ^{1/5}x}),%
\text{ }j=0,1,2,3,\text{ }\varepsilon =\nu _{+}\alpha _{+}^{5/2},
\end{eqnarray*}%
where $|X_{0}|\leq k_{1}\alpha _{+}\delta .$ As $x\rightarrow +\infty ,$ $%
(A_{0},A_{1},A_{2},A_{3})\rightarrow 0$ as $exp(-\sqrt{\frac{\delta }{2}}x),$
$(1-B_{0},B_{1})\rightarrow 0$ as $\exp (-\sqrt{2}\varepsilon x).$
\end{lemma}

\begin{proof}
Let us solve (\ref{stabmanifXY}) with respect to $(X,Y)\in C_{\kappa }^{0}$
for $|X_{0}|$ small enough and choose 
\begin{equation*}
\kappa =\frac{1}{10}\delta ^{2/5}\varepsilon ^{1/5}=\overline{\kappa }\sqrt{%
\frac{\alpha _{+}\delta }{2}},\text{ \ }\overline{\kappa }=\frac{\sqrt{2}}{10%
}\left( \frac{\nu _{+}}{\sqrt{\delta }}\right) ^{1/5}
\end{equation*}%
then $\nu _{+}$ is choosen such that $\overline{\kappa }<1.$ Then (\ref%
{stabmanifXY}) implies (below, the space $C_{\kappa }^{0}$ is built on $%
[x_{\ast }^{+},+\infty )$ instead of $[-x_{\ast },+\infty )$) 
\begin{eqnarray*}
||X||_{\kappa } &\leq &|X_{0}|+\frac{1}{1-\overline{\kappa }}\sqrt{\frac{2}{%
\alpha _{+}\delta }}||G_{0}||_{\kappa }, \\
||Y||_{\kappa } &\leq &\frac{1}{1+\overline{\kappa }}\sqrt{\frac{2}{\alpha
_{+}\delta }}||G_{1}||_{\kappa }.
\end{eqnarray*}%
Moreover we have for $j=0,1$ (see the expressions of $G_{j}),$ using $%
||M_{j}||\leq \sqrt{18},$%
\begin{equation*}
||G_{j}||_{\kappa }\leq \frac{(1+c\alpha_{+} ^{2})}{2\sqrt{2}\delta ^{3/2}\alpha_{+}
^{3/2}}(||X||_{\kappa }+||Y||_{\kappa })^{3}+\frac{3\nu_{+}\delta^2 \alpha_{+} ^{1/2}}{4}%
(||X||_{\kappa }+||Y||_{\kappa })
\end{equation*}%
with $c$ independent of $\varepsilon .$ Now making the scaling%
\begin{equation*}
(X,Y)=\alpha_{+} (\overline{X},\overline{Y}),
\end{equation*}%
we obtain%
\begin{eqnarray*}
||\overline{X}||_{\kappa } &\leq &|\overline{X_{0}}|+\frac{3}{2\sqrt{2}(1-%
\overline{\kappa })}\frac{\nu _{+}\delta^2}{\sqrt{\delta }}(||\overline{X}||_{\kappa
}+||\overline{Y}||_{\kappa })+\frac{(1+c\alpha _{+}^{2})}{(1-\overline{%
\kappa })\delta ^{2}}(||\overline{X}||_{\kappa }+||\overline{Y}||_{\kappa
})^{3}, \\
||\overline{Y}||_{\kappa } &\leq &\frac{3}{2\sqrt{2}(1+\overline{\kappa })}%
\frac{\nu _{+}\delta^2}{\sqrt{\delta }}(||\overline{X}||_{\kappa }+||\overline{Y}%
||_{\kappa })+\frac{(1+c\alpha _{+}^{2})}{(1+\overline{\kappa })\delta ^{2}}%
(||\overline{X}||_{\kappa }+||\overline{Y}||_{\kappa })^{3}.
\end{eqnarray*}%
It is then clear that, for 
\begin{equation*}
\frac{\nu _{+}\delta^2}{\sqrt{\delta }}<\frac{\sqrt{2}}{3}(1-\overline{\kappa }^{2}),
\end{equation*}%
i.e. (using the definition of $\overline{\kappa })$ 
\begin{equation}
\nu _{+}\delta^{3/2} \leq 0.46,\text{ }\overline{\kappa }\leq
0.1645  \label{cond nu+}
\end{equation}%
we can appy the implicit function theorem (in the analytic frame as in \cite%
{Dieudo} section X.2) for $|\overline{X_{0}}|\leq k_{1}$ with $0<k_{1}$
small enough and for $\varepsilon $ small enough, so that we obtain a unique
solution $(\overline{X},\overline{Y})$ in $C_{\kappa }^{0}$ satisfying%
\begin{eqnarray*}
||\overline{X}||_{\kappa } &\leq &(1+c\nu _{+})|\overline{X_{0}}|, \\
||\overline{Y}||_{\kappa } &\leq &c\nu _{+}|\overline{X_{0}}|+c|\overline{%
X_{0}}|^{3},
\end{eqnarray*}%
with $c$ independent of $\varepsilon ,\nu _{+}$ small enough, and provided
that%
\begin{equation*}
|\overline{X_{0}}|\leq k_{1},\text{ \ }k_{1}\text{ small enough.}
\end{equation*}%
By construction (see (\ref{variables M+})) of $(X,Y)$ we obtain the
estimates on $A_{j}(x)$ indicated at Lemma \ref{Lemma stable manifM+}. It
then remains to check the validity of condition (\ref{Hyp Aj}). Indeed
estimates (\ref{apriori estim v(x)}) of $v$ imply 
\begin{equation*}
|v_{0}|e^{\frac{-5\varepsilon \sqrt{2}x}{4}}\leq |v(x)|\leq \frac{|v_{0}|}{1-%
\frac{|v_{0}|}{2}}e^{\frac{-3\varepsilon \sqrt{2}x}{4}},
\end{equation*}%
where%
\begin{equation*}
|v_{0}|=1-B_{0}(-x_{\ast })\sim 1-\frac{1}{\sqrt{1+\delta ^{2}}}.
\end{equation*}%
Moreover, for $j=0,1,2,3$ and $\varepsilon $ small enough%
\begin{equation*}
|A_{j}(x)|\leq 2|X_{0}|e^{-\kappa x}\leq 2k_{1}\alpha _{+}e^{-\kappa x},%
\text{ }x>x_{\ast }^{+}.
\end{equation*}%
Since we have $\kappa =\mathcal{O}(\varepsilon ^{1/5})$, then for $%
\varepsilon $ small enough%
\begin{equation*}
e^{-\kappa x}\leq e^{\frac{-5\varepsilon \sqrt{2}x}{4}},\text{ }x>x_{\ast
}^{+}>0,
\end{equation*}%
the required condition (\ref{Hyp Aj}) is realized as soon as%
\begin{equation*}
2k_{1}\alpha _{+}\leq 1-B_{0}(-x_{\ast })
\end{equation*}%
which holds true for $\varepsilon $ small enough. The exponential behavior
declared in Lemma \ref{Lemma stable manifM+} follows from the linear study
of section \ref{sect linear} as $x\rightarrow +\infty .$ This ends the proof
of Lemma \ref{Lemma stable manifM+}, and of Corollary \ref{Corollary stab}
for the part $x\in \lbrack x_{\ast }^{+},+\infty ).$
\end{proof}

\subsubsection{Estimate of $\protect\nu _{+}$}

From the proof of Lemma \ref{Lemma stable manifM+}, and from the fact that $%
\overline{\kappa }$ may be stated as small as we need, the restriction on $%
\nu _{+}$ is%
\begin{equation}
\nu _{+}\delta^{3/2} \leq \frac{\sqrt{2}}{3}.  \label{estim nu+}
\end{equation}

\subsection{Intersection of the stable manifold with $H_{1}$}

We need to compute the intersection of the 3-dimensional stable manifold $%
\mathcal{W}_{\varepsilon ,\delta }^{(s)}$ of $M_{+}$ with the hyperplane $%
H_{1}$ defined by%
\begin{equation}
B_{0}=B_{01}\overset{def}{=}\sqrt{\frac{1+\alpha _{+}^{2}\delta ^{2}}{%
1+\delta ^{2}}}.  \label{hyperpl B0}
\end{equation}%
We then obtain a 2-dimensional sub-manifold living in the 4-dimensional
manifold $\mathcal{W}_{\varepsilon ,\delta }\cap H_{1}.$ We have the
following

\begin{lemma}
\label{intersect H1}For $1/3\leq \delta \leq 1,$ $\varepsilon =\nu
_{+}\alpha _{+}^{5/2},$ $\varepsilon $ and $\nu _{+}$ small enough, the
two-dimensional intersection of the 3-dimensional plane, tangent to the
stable manifold of $M_{+},$ with the 5-dimensional hyperplane $H_{1},$
satisfies a linear system with rescaled parameters $(\overline{x_{10}},%
\overline{x_{20}})=\frac{1}{\delta }\overline{X_{0}}$ and, by construction $%
\widetilde{\delta }|_{B_{01}}=(\alpha _{+}\delta )^{1/2},$ 
\begin{eqnarray}
A_{0} &=&\delta \alpha _{+}(\overline{x_{10}}+\overline{y_{10}}),\text{ } 
\notag \\
A_{1} &=&-\frac{(\delta \alpha _{+})^{3/2}}{\sqrt{2}}(\overline{x_{10}}+%
\overline{x_{20}}-\overline{y_{10}}-\overline{y_{20}})  \label{tr manif M+}
\\
A_{2} &=&(\alpha _{+}\delta )^{2}(\overline{x_{20}}+\overline{y_{20}}) 
\notag \\
A_{3} &=&\frac{(\alpha _{+}\delta )^{5/2}}{\sqrt{2}}(\overline{x_{10}}-%
\overline{x_{20}}-\overline{y_{10}}+\overline{y_{20}}),  \notag
\end{eqnarray}%
where $\overline{Y_{0}}$ is a linear function of $\overline{X_{0}}$ such
that $|\overline{Y_{0}}|\leq c\nu _{+}|\overline{X_{0}}|$, with the
restriction%
\begin{equation*}
(|\overline{x_{10}}|^{2}+|\overline{x_{20}}|^{2})^{1/2}\leq k_{1}.
\end{equation*}
\end{lemma}

\section{Extension and Intersection of the two manifolds \label{sect
intersection}}

\subsection{Extension of the stable manifold of $M_{+}$ for $x\in \lbrack
-x_{\ast },x_{\ast }^{+}]$}

We need to extend the definition of the stable manifold of $M_{+}$ to the
region where%
\begin{equation*}
\frac{1-\alpha _{-}^{2}\delta ^{2}}{1+\delta ^{2}}\leq B_{0}^{2}(x)\leq 
\frac{1+\alpha _{+}^{2}\delta ^{2}}{1+\delta ^{2}},
\end{equation*}%
i.e where%
\begin{equation*}
-x_{\ast }\leq x\leq x_{\ast }^{+},
\end{equation*}%
and where%
\begin{equation}
A_{0}^{(4)}=-A_{0}[A_{0}^{2}+(1+\delta ^{2})B_{0}^{2}-1],
\label{4thorderODE}
\end{equation}%
$B_{0}$ satisfying Lemma \ref{Lem v(x)}. We prove the following

\begin{lemma}
\label{connection for stablemanif} For $1/3\leq \delta \leq 1,$ and $%
\varepsilon $ small enough, the 3-dimensional stable manifold $\mathcal{W}%
_{\varepsilon ,\delta }^{(s)}$ of $M_{+}$ still exists for $B_{0}^{2}\in
\lbrack \frac{1-\alpha _{-}^{2}\delta ^{2}}{1+\delta ^{2}},\frac{1+\alpha
_{+}^{2}\delta ^{2}}{1+\delta ^{2}}],$ is analytic in $\varepsilon ,\delta ,$
parameterized by $(X_{0},B_{0}(x_{\ast }^{+})),$ where $|X_{0}|\leq
k_{1}\alpha _{+}$ with $k_{1}$ small enough, as in Lemma \ref{Lemma stable
manifM+}. Moreover $B_{0}(-x_{\ast })\leq B_{0}(x)\leq B_{0}(x_{\ast }^{+}),$
$B_{0}^{\prime }(x)>0$ where $B_{0}$ is given by Lemma \ref{Lem v(x)}, and $%
A_{0}$ and it derivatives solve the differential equation (\ref{4thorderODE}%
) on the interval $-x_{\ast }\leq x\leq x_{\ast }^{+}.$
\end{lemma}

We observe that Lemma \ref{Lem v(x)} is valid, provided that (\ref{Hyp Aj})
holds, hence%
\begin{eqnarray*}
(1+\delta ^{2})B_{0}^{2}-1 &=&\frac{2\sqrt{2}\delta ^{2}}{\sqrt{1+\delta ^{2}%
}}[\varepsilon x+\mathcal{O}(\varepsilon x)^{2}+\frac{\varepsilon }{\kappa }%
\mathcal{O}(|(A_{0},A_{1},A_{2},A_{3})|^{2})], \\
-\alpha _{-}^{2}\frac{\sqrt{1+\delta ^{2}}}{2\sqrt{2}} &\sim &-\varepsilon
x_{\ast }\leq \varepsilon x\leq \varepsilon x_{\ast }^{+}\sim \alpha _{+}^{2}%
\frac{\sqrt{1+\delta ^{2}}}{2\sqrt{2}},
\end{eqnarray*}%
where, from lemma \ref{Lemma stable manifM+} 
\begin{equation*}
|A_{j}(x_{\ast }^{+})|\leq ck_{1}\alpha _{+}^{(1+j/2)},\text{ }j=0,1,2,3,
\end{equation*}%
with $k_{1}$ small enough and $\frac{\varepsilon }{\kappa }=\frac{%
\varepsilon }{\delta _{\ast }\varepsilon ^{1/5}}=\mathcal{O}(\varepsilon
^{4/5}).$

The strategy here consists to solve backwards the principal part of the
differential equation (\ref{4thorderODE}) on $[-x_{\ast },x_{\ast }^{+}]$,
where $A_{0}$ satisfies the boundary conditions in $x_{\ast }^{+}.$ In fact
we only consider the boundary conditions coming from the tangent plane of
the stable manifold for $x=x_{\ast }^{+}$ (role of hyperplanr $H_{1})$,
since, for $A_{0}$ we stay in a neighborhoof of $0$ and we only need at the
end the trace of the tangent plane to the manifold at the extreme point for$%
x=-x_{\ast }.$ It is a sort of managing the transport of the tangent plane
to the manifold on the interval $[-x_{\ast },x_{\ast }^{+}].$ Then we need
to study the intersection of the 2-dimensional tangent plane along the
intersection $\mathcal{W}_{\varepsilon ,\delta }^{(u)}\cap H_{0}$ with the
2-dimensional tangent plane along the intersection $\mathcal{W}_{\varepsilon
,\delta }^{(s)}\cap H_{0},$ all this lying in the 4-dimensional tangent
plane to the manifold $H_{0}\cap \mathcal{W}_{\varepsilon ,\delta }.$

Let us first consider the principal part of (\ref{4thorderODE}) and rescale
as%
\begin{eqnarray*}
z &=&K\varepsilon ^{1/5}x,\text{ }A_{0}(x)=K^{2}\varepsilon ^{2/5}\overline{%
A_{0}}(z) \\
K &=&\left( \frac{2\sqrt{2}\delta ^{2}}{\sqrt{1+\delta ^{2}}}\right) ^{1/5},
\end{eqnarray*}%
Now%
\begin{equation}
z\in \lbrack -a_{-},a_{+}]\text{ \ with }a_{\pm }=\left( \frac{(1+\delta
^{2})\delta }{8\nu _{\pm }^{2}}\right) ^{2/5}\text{ independent of }%
\varepsilon ,  \label{def a+-}
\end{equation}%
with 
\begin{equation*}
\delta ^{2}=(K\nu _{\pm }^{1/5})^{4}a_{\pm }.
\end{equation*}%
The principal part of the differential equation for $A_{0}$ reads now as%
\begin{equation}
\frac{d^{4}\overline{A_{0}}}{dz^{4}}=-\overline{A_{0}}(\overline{A_{0}}%
^{2}+z),  \label{limit interm equ}
\end{equation}%
with boundary conditions coming from the intersection of the unstable
manifold of $M_{-}$ with $H_{0}$ for $z=-a_{-},$ and from the intersection
of the stable manifold of $M_{+}$ with $H_{1}$ for $z=+a_{+}.$

\begin{remark}
From the definition of $a_{\pm }$ and from estimates (\ref{restrict nu-})
for $\nu _{-}$ and (\ref{estim nu+}) for $\nu _{+},$ we obtain%
\begin{equation}
(a_{-})_{\min }=34.74\left( \frac{1+\delta ^{2}}{2}\right) ^{6/5},
\label{a- min}
\end{equation}%
\begin{equation}
(a_{+})_{\min }=1.05\left( \frac{1+\delta ^{2}}{2}\right) ^{2/5}\delta^{8/5}.
\label{a+ min}
\end{equation}
\end{remark}

For the principal part of the intersection with $H_{0}$, for $z=-a_{-}$ (see
(\ref{trace tg unstmanif}))%
\begin{eqnarray}
\overline{A_{0}} &=&a_{-}^{1/2}[1+\frac{1}{2^{3/4}}(\overline{x_{1}}^{u}-%
\overline{x_{2}}^{u})]  \notag \\
\frac{d\overline{A_{0}}}{dz} &=&a_{-}^{3/4}\overline{x_{1}}^{u}  \notag \\
\frac{d^{2}\overline{A_{0}}}{dz^{2}} &=&\frac{a_{-}}{2^{1/4}}(\overline{x_{1}%
}^{u}+\overline{x_{2}}^{u})  \label{bound cond z=-1} \\
\frac{d^{3}\overline{A_{0}}}{dz^{3}} &=&\sqrt{2}a_{-}^{5/4}\overline{x_{2}}%
^{u}  \notag
\end{eqnarray}%
where $(\overline{x_{1}}^{u},\overline{x_{2}}^{u})$ is a 2-dimensional
parameter of size $k_{0}$ assumed to be small enough and $\nu _{-}$ is also
small enough, independent of $\varepsilon ,k_{0}.$

We also obtain the principal part of the intersection with $H_{1}$, for $%
z=+a $ (see (\ref{tr manif M+})) as%
\begin{eqnarray}
\overline{A_{0}} &=&a_{+}^{1/2}\overline{x_{10}}^{s},\text{ }  \notag \\
\frac{d\overline{A_{0}}}{dz} &=&-\frac{a_{+}^{3/4}}{\sqrt{2}}(\overline{%
x_{10}}^{s}+\overline{x_{20}}^{s})  \label{bound cond z=1} \\
\frac{d^{2}\overline{A_{0}}}{dz^{2}} &=&a_{+}\overline{x_{20}}^{s}  \notag \\
\frac{d^{3}\overline{A_{0}}}{dz^{3}} &=&\frac{a_{+}^{5/4}}{\sqrt{2}}(%
\overline{x_{10}}^{s}-\overline{x_{20}}^{s}),  \notag
\end{eqnarray}%
where $(\overline{x_{10}}^{s},\overline{x_{20}}^{s})$ is a 2-dimensional
parameter assumed to be bounded in $%
%TCIMACRO{\U{211d} }%
%BeginExpansion
\mathbb{R}
%EndExpansion
^{2}$ by $k_{1}$. For (\ref{limit interm equ}), the heteroclinic curve is
obtained when we find a solution $\overline{A_{0}}$ satisfying the boundary
conditions with principal parts given above above in $z=\pm a_{\pm }.$ We
observe that $\varepsilon $ has completely disappeared from this
formulation, and that we have 4 parameters for this 4th order differential
equation on the interval $[-a_{-},+a_{+}].$ It is clear that in satisfying
the boundary conditions at $z=+a_{+},$ we obtain a two-parameter family of
solutions of (\ref{limit interm equ}) which needs to exist until $z=-a_{-},$
where $\varepsilon $ does not play any role. In Appendix \ref{app lem} we
prove the following

\begin{lemma}
\label{conjecture} For the initial conditions (\ref{bound cond z=1}) in $%
z=+a_{+},$ and for $k_{1}$ small enough, and $|(\overline{x_{10}}^{s},%
\overline{x_{20}}^{s})|\leq k_{1}$ the differential equation (\ref{limit
interm equ}) has a 2-dimensional (parameter $(\overline{x_{10}}^{s},%
\overline{x_{20}}^{s})$) family of solutions for $z\in \lbrack
-a_{-},+a_{+}].$ These solutions depend analytically of $\delta \in \lbrack
1/3,1].$
\end{lemma}

\begin{proof}
The proof made in Appendix \ref{app lem} uses the standard way starting from 
$a_{+}$, with the fixed point theorem in a function space of bounded
continuous functions. The only restriction for its use is that 
\begin{equation*}
a_{+}^{5}<3/2.
\end{equation*}%
This is verified with our suitable choice of $a_{+}$ subjected to the
restriction (\ref{a+ min}).
\end{proof}

Now, the complete system (\ref{4thorderODE}), which depends on $\varepsilon
, $ is a small regular perturbation of size $\ \mathcal{O}(\varepsilon
^{4/5})$ (after scaling) of (\ref{limit interm equ}), as well for the
perturbed boundary conditions, so that the solution of (\ref{4thorderODE})
with its given boundary conditions at $x=x_{\ast }^{+},$ exists until $%
x=-x_{\ast }$, and corresponds to a small perturbation of the solution of (%
\ref{limit interm equ}) satisfying the principal part of the boundary
condition at $z=a_{+}$.

\begin{remark}
The differential equation (\ref{limit interm equ}) is mentionned in \cite%
{Man-Pom}, where a matching asymptotic method is used, it also appears in 
\cite{Buff} where a variational method is used. Both studies are for an
infinite interval. Here we have the big advantage to deal with a finite
interval, allowing to prove Lemma \ref{conjecture}.
\end{remark}

\subsection{Intersection of the two manifolds}

In this section we prove the following:

\begin{lemma}
\label{intersection of manifolds} For $\varepsilon $ small enough, and for $%
1/3\leq \delta \leq 1$, then, except maybe for a finite number of values,
the unstable manifold $\mathcal{W}_{\varepsilon ,\delta }^{(u)}$ of $M_{-}$
intersects transversally the stable manifold $\mathcal{W}_{\varepsilon
,\delta }^{(s)}$ of $M_{+}$ along the heteroclinic solution. Moreover, for $%
x=0,$ we have the estimates $A_{j}(0)=\mathcal{O}(\varepsilon ^{(2+j)/5}),$ $%
j=0,1,2,3.$
\end{lemma}

\begin{proof}
For proving the intersection it is sufficient to prove that the manifolds
intersect in the hyperplane $H_{0}.$

For the differential equation (\ref{limit interm equ}),$\ \overline{A_{0}}%
(-a_{-})$ and its 3 first derivatives should satisfy (\ref{bound cond z=-1}%
). Considering the 2-dimensional tangent plane to the manifold formed by the
2-parameter family of solutions coming from the solution satisfying the
conditions in $z=a_{+},$ we then obtain a linear system in $%
%TCIMACRO{\U{211d} }%
%BeginExpansion
\mathbb{R}
%EndExpansion
^{4}$ for the 4 unknowns $(\overline{x_{10}}^{s},\overline{x_{20}}^{s},%
\overline{x_{1}}^{u},\overline{x_{2}}^{u}).$ Looking at the system, which is
independent of $\varepsilon ,$ we see that the solution is independent of $%
\varepsilon .$

If we ignore the relationship coming from the differential equation (\ref%
{limit interm equ}) and just equate both sides, we obtain the unique solution%
\begin{eqnarray}
\overline{x_{1}}^{u} &=&-2\rho \frac{(2^{1/4}\rho -1)}{2\rho ^{4}-1},\text{ }%
\rho =\left( \frac{a_{-}}{a_{+}}\right) ^{1/4},  \label{solu without equ} \\
\overline{x_{2}}^{u} &=&2^{3/4}\frac{(2^{1/4}\rho -1)}{2\rho ^{4}-1},  \notag
\\
\overline{x_{10}}^{s} &=&\frac{\sqrt{2}\rho (\sqrt{2}\rho ^{3}-1)}{2\rho
^{4}-1},  \notag \\
\overline{x_{20}}^{s} &=&-\rho ^{4}\frac{\sqrt{2}(2^{1/4}\rho -1)^{2}}{2\rho
^{4}-1}.  \notag
\end{eqnarray}%
Now the action of the differential equation (\ref{limit interm equ}) on $%
[-a_{-},+a_{+}]$ "rotates" the 2-dimensional tangent plane, starting from
the plane (\ref{bound cond z=1}) in $z=+a_{+}$, and arriving for $z=-a_{-}$
to a system different from the above, hence giving a solution different from
(\ref{solu without equ}). The 2-parameter family of solutions depends
analytically of $\delta $, so that if the solution is not degenerated, the
solution of the 4-dimensional linear system is unique, except maybe for a
finite number of values of $\delta .$ The solution (\ref{solu without equ})
which is obtained above, shows that the linear system should need to
"rotate" suitably for becoming degenerate, where the "rotation" is function
of $\delta $. This can only happen for a discrete set of values of $\delta .$

It remains to consider the full equation (\ref{4thorderODE}) which is a
regular perturbation of (\ref{limit interm equ}), of order $\varepsilon
^{4/5}$ including the bounds of the finite interval. It is then clear that
the same result holds true for $\varepsilon $ small enough, i.e the
existence of a unique solution, meaning the transversality of the two
manifolds, except maybe for a finite number of values of $\delta $. However
we observe that the conditions on the norms of $\overline{x_{10}}^{s},%
\overline{x_{20}}^{s}$ and on $\overline{x_{1}}^{u},\overline{x_{2}}^{u}$
cannot be checked directly, even though they all are of order 1, as
expected. Here, we need to use the result of \cite{BHI} which asserts the
existence of a heteroclinic where $B_{0}$ starts from $0$ for $x=-\infty ,$
arriving at $B_{0}=1$ for $x=+\infty ,$ and where $B_{0}(x)\in \lbrack 0,1].$
Indeed, such a heteroclinic necessarily belongs to the stable manifold of $%
M_{-}$ which is constructed here for $x\in (-\infty ,-x_{\ast }]$, and also
belongs to the stable manifold of $M_{+}$ which we have contructed for $x\in
\lbrack -x_{\ast },+\infty ).$ This means that the intersection of tangent
planes for $x=-x_{\ast }$, given, for the principal part, by the above
mentioned 4-dimensional system has a solution. Such a solution is unique,
from the proof made above, except maybe for a finite number of values of $%
\delta .$ Theorem \ref{theorem} is proved.
\end{proof}

\begin{remark}
Since the size of $A_{j}(x)$ with respect to $\varepsilon$ is not modified
on $[-x_{\ast },x_{\ast }^{+}],$ the result above implies the validity of
the parts of Corollaries \ref{Corollary unstab} and \ref{Corollary stab} for 
$-x_{\ast }\leq x\leq 0$ and $0\leq x\leq x_{\ast }^{+}$ respectively.
\end{remark}

\section{Study of the linearized operator}

\label{studylinop}

Let us redefine the heteroclinic connection we found at Theorem \ref{theorem}
as%
\begin{equation*}
(A_{\ast }(x),B_{\ast }(x))\subset 
%TCIMACRO{\U{211d} }%
%BeginExpansion
\mathbb{R}
%EndExpansion
^{2}
\end{equation*}%
with%
\begin{equation*}
1+1/9\leq g=1+\delta ^{2}\leq 2,
\end{equation*}%
and where we know that, for $\varepsilon $ small enough%
\begin{eqnarray*}
B_{\ast }(x) &>&0,\text{ }B_{\ast }^{\prime }(x)>0 \\
(A_{\ast }(x),B_{\ast }(x)) &\rightarrow &\left\{ 
\begin{array}{c}
(1,0)\text{ as }x\rightarrow -\infty \\ 
(0,1)\text{ as }x\rightarrow +\infty%
\end{array}%
\right. ,
\end{eqnarray*}%
at least as $e^{\varepsilon \delta x}$ for $x\rightarrow -\infty ,$ and at
least as $e^{-\sqrt{2}\varepsilon x}$ \ for $x\rightarrow +\infty $.

At section \ref{sect origin} we show that the perturbed system (\ref{new
reduced syst}) leading to system (\ref{truncated syst}) is now considered
with $B_{0}$ complex valued, so in (\ref{truncated syst}) $B^{2}$ is
replaced by $|B|^{2}.$

For being able to prove any persistence result under reversible
perturbations of system (\ref{truncated syst}) in $%
%TCIMACRO{\U{211d} }%
%BeginExpansion
\mathbb{R}
%EndExpansion
^{4}\times 
%TCIMACRO{\U{2102} }%
%BeginExpansion
\mathbb{C}
%EndExpansion
^{2}$, as it appears in (\ref{new reduced syst}), we need to study the
linearized operator at the above heteroclinic solution. We follow the lines
of \cite{HI Arma}.

The linearized operator is given by%
\begin{eqnarray*}
A^{(4)} &=&(1-3A_{\ast }^{2}-gB_{\ast }^{2})A-gA_{\ast }B_{\ast }(B+%
\overline{B}), \\
B^{\prime \prime } &=&\varepsilon ^{2}(-1+gA_{\ast }^{2}+2B_{\ast
}^{2})B+2\varepsilon ^{2}gA_{\ast }B_{\ast }A+\varepsilon ^{2}B_{\ast }^{2}%
\overline{B}.
\end{eqnarray*}%
Taking real and imaginary parts for $B:$%
\begin{equation*}
B=C+iD,
\end{equation*}%
we then obtain the linearized system%
\begin{eqnarray*}
-A^{(4)}+(1-3A_{\ast }^{2}-gB_{\ast }^{2})A-2gA_{\ast }B_{\ast }C &=&0, \\
\frac{1}{\varepsilon ^{2}}C^{\prime \prime }+(1-gA_{\ast }^{2}-3B_{\ast
}^{2})C-2gA_{\ast }B_{\ast }A &=&0, \\
\frac{1}{\varepsilon ^{2}}D^{\prime \prime }+(1-gA_{\ast }^{2}-B_{\ast
}^{2})D &=&0.
\end{eqnarray*}%
Notice that the equation for $D$ decouples, so that we can split the linear
operator in an operator $\mathcal{M}_{g}$ acting on $(A,C)$ and an operator $%
\mathcal{L}_{g}$ acting on $D:$%
\begin{equation*}
\mathcal{M}_{g}\left( 
\begin{array}{c}
A \\ 
C%
\end{array}%
\right) =\left( 
\begin{array}{c}
-A^{(4)}+(1-3A_{\ast }^{2}-gB_{\ast }^{2})A-2gA_{\ast }B_{\ast }C \\ 
\frac{1}{\varepsilon ^{2}}C^{\prime \prime }+(1-gA_{\ast }^{2}-3B_{\ast
}^{2})C-2gA_{\ast }B_{\ast }A%
\end{array}%
\right) ,
\end{equation*}%
\begin{equation*}
\mathcal{L}_{g}D=\frac{1}{\varepsilon ^{2}}D^{\prime \prime }+(1-gA_{\ast
}^{2}-B_{\ast }^{2})D.
\end{equation*}%
Let us define the Hilbert spaces%
\begin{equation*}
L_{\eta }^{2}=\{u;u(x)e^{\eta |x|}\in L^{2}(%
%TCIMACRO{\U{211d} }%
%BeginExpansion
\mathbb{R}
%EndExpansion
)\},
\end{equation*}%
\begin{eqnarray*}
\mathcal{D}_{0} &=&\{(A,C)\in H_{\eta }^{4}\times H_{\eta }^{2};A\in H_{\eta
}^{4},C\in \mathcal{D}_{1}\} \\
\mathcal{D}_{1} &=&\{C\in H_{\eta }^{2};\varepsilon ^{-2}||C^{\prime \prime
}||_{L_{\eta }^{2}}+\varepsilon ^{-1}||C^{\prime }||_{L_{\eta
}^{2}}+||C||_{L_{\eta }^{2}}\overset{def}{=}||C||_{\mathcal{D}_{1}}<\infty \}
\end{eqnarray*}%
equiped with natural scalar products. Below, we prove the following

\begin{lemma}
\label{Lem Linearoperator}Except maybe for a set of isolated values of $g,$
the kernel of $\mathcal{M}_{g}$ in $L_{\eta }^{2}$ is one dimensional, span
by $(A_{\ast }^{\prime },B_{\ast }^{\prime }),$ and its range has
codimension 1, $L^{2}$- orthogonal to $(A_{\ast }^{\prime },B_{\ast
}^{\prime }).$ $\mathcal{M}_{g}$ has a pseudo-inverse acting from $L_{\eta
}^{2}$ to $\mathcal{D}_{0}$ for any $\eta >0$ small enough, with bound
independent of $\varepsilon .$

The operator $\mathcal{L}_{g}$ has a trivial kernel, and its range which has
codimension 1, is $L^{2}$- orthogonal to $B_{\ast }$ ($B_{\ast }\notin
L^{2}).$ $\mathcal{L}_{g}$ has a pseudo-inverse acting from $L_{\eta }^{2}$
to $\mathcal{D}_{1}$ for $\eta >0$ small enough, with bound independent of $%
\varepsilon .$
\end{lemma}

\begin{remark}
For proving the above Lemma, we use the uniqueness (resulting from the
transversality of manifolds $\mathcal{W}_{\varepsilon ,\delta }^{(u)}$ and $%
\mathcal{W}_{\varepsilon ,\delta }^{(s)}$) and analyticity in $\delta $
(i.e. $g$) of the heteroclinic, proved in previous section (see subsection %
\ref{sect dim Mg}).
\end{remark}

\begin{remark}
\label{rmk persist}The above Lemma is useful for proving the persistence
under reversible perturbations, as indicated in (\ref{new reduced syst}), of
our heteroclinic. This is done in \cite{Io24} and appears to be more
difficult than the symmetric case solved in \cite{HI Arma}. Indeed, it is
needed to introduce two different wave numbers for the two systems of
convective rolls at $\pm \infty .$ In \cite{Io24} it is shown that the
component on the kernel of $\mathcal{M}_{g}$ corresponds to a sort of
adapted phase shift of rolls parallel to the wall, while the codimension 2
of the range implies that each wave number is function not only of the
amplitude of rolls but also of the above shift. This then leads to a one
parameter family of domain walls, for any fixed small amplitude $\varepsilon
^{2}$.
\end{remark}

\subsection{Asymptotic operators}

Let us define the operators obtained when $x=\pm \infty :$

\begin{equation*}
\mathcal{M}_{\infty }^{-}\left( 
\begin{array}{c}
A \\ 
C%
\end{array}%
\right) =\left( 
\begin{array}{c}
-A^{(4)}-2A \\ 
\varepsilon ^{-2}C^{\prime \prime }-(g-1)C%
\end{array}%
\right) ,
\end{equation*}%
\begin{equation*}
\mathcal{M}_{\infty }^{+}\left( 
\begin{array}{c}
A \\ 
C%
\end{array}%
\right) =\left( 
\begin{array}{c}
-A^{(4)}-(g-1)A \\ 
\varepsilon ^{-2}C^{\prime \prime }-2C%
\end{array}%
\right) ,
\end{equation*}%
\begin{eqnarray*}
\mathcal{L}_{\infty }^{-}D &=&\varepsilon ^{-2}D^{\prime \prime }-(g-1)D, \\
\mathcal{L}_{\infty }^{+}D &=&\varepsilon ^{-2}D^{\prime \prime }.
\end{eqnarray*}%
Notice that all these operators are negative. Furthermore, their spectra in $%
L^{2}(%
%TCIMACRO{\U{211d} }%
%BeginExpansion
\mathbb{R}
%EndExpansion
)$ are such that%
\begin{eqnarray*}
\sigma (\mathcal{M}_{\infty }^{-}) &=&(-\infty ,-c_{-}],\text{ }c_{-}=\max
\{2,(g-1)\}>0, \\
\sigma (\mathcal{M}_{\infty }^{+}) &=&(-\infty ,-c_{+}],\text{ }c_{+}=c_{-},
\\
\sigma (\mathcal{L}_{\infty }^{-}) &=&(-\infty ,-(g-1)], \\
\sigma (\mathcal{L}_{\infty }^{+}) &=&(-\infty ,0].
\end{eqnarray*}%
Operators $\mathcal{M}_{g}$ and $\mathcal{L}_{g}$ are respectively
relatively compact perturbations of the corresponding asymptotic operators $%
\mathcal{M}_{\infty }$ and $\mathcal{L}_{\infty }$ defined as%
\begin{equation*}
\mathcal{M}_{\infty }=\left\{ 
\begin{array}{c}
\mathcal{M}_{\infty }^{-},\text{ \ }x<0 \\ 
\mathcal{M}_{\infty }^{+},\text{ \ }x>0%
\end{array}%
\right. ,\text{ \ }\mathcal{L}_{\infty }=\left\{ 
\begin{array}{c}
\mathcal{L}_{\infty }^{-},\text{ \ }x<0 \\ 
\mathcal{L}_{\infty }^{+},\text{ \ }x>0%
\end{array}%
\right. ,
\end{equation*}%
Their essential spectrum, i.e. the set of $\lambda \in 
%TCIMACRO{\U{2102} }%
%BeginExpansion
\mathbb{C}
%EndExpansion
$ for which $\lambda -\mathcal{M}_{g}$ (resp. $\lambda -\mathcal{L}_{g}%
\mathcal{)}$ is not Fredholm with index 0, is equal to the essential
spectrum of $\mathcal{M}_{\infty }$ (resp. $\mathcal{L}_{\infty })$ (see 
\cite{Kato}). The latter spectra are found from the spectra of $\mathcal{M}%
_{\infty }^{\pm }$ and $\mathcal{L}_{\infty }^{\pm }:$%
\begin{eqnarray*}
\sigma _{ess}(\mathcal{M}_{\infty }) &=&(-\infty ,-c_{+}], \\
\sigma _{ess}(\mathcal{L}_{\infty }) &=&(-\infty ,0].
\end{eqnarray*}%
In particular, this implies that $0$ does not belong to the essential
spectrum of $\mathcal{M}_{g}$, so that the operator $\mathcal{M}_{g}$ is
Fredholm with index $0.$ Moreover operators $\mathcal{M}_{\infty }$ and $%
\mathcal{L}_{\infty }$ are self adjoint negative operators in $L^{2},$ and $%
\mathcal{M}_{\infty }$ has a bounded inverse \cite{Kato}.%
\begin{equation*}
||\mathcal{M}_{\infty }^{-1}||_{L^{2}}\leq \frac{1}{c_{+}}.
\end{equation*}%
This last property remains valid in exponentially weighted spaces, with
weights $e^{\eta |x|},$ and $\eta $ sufficiently small, since this acts as a
small perturbation of the differential operator (see \cite{Kap-Pro} section
3.1).

\subsection{Properties of $\mathcal{L}_{g}$}

Notice that $\mathcal{L}_{g}$ is self adjoint in $L^{2}(%
%TCIMACRO{\U{211d} }%
%BeginExpansion
\mathbb{R}
%EndExpansion
)$ and that%
\begin{equation*}
\mathcal{L}_{g}B_{\ast }=0,\text{ \ but }B_{\ast }\notin L^{2}(%
%TCIMACRO{\U{211d} }%
%BeginExpansion
\mathbb{R}
%EndExpansion
).
\end{equation*}%
This property allows to solve explicitely the equation $\mathcal{L}%
_{g}u=f\in L_{\eta }^{2}$ with respect to $u\in L_{\eta }^{2}$ (using
variation of constants method), and shows that it has a unique solution,
provided that 
\begin{equation*}
\int_{%
%TCIMACRO{\U{211d} }%
%BeginExpansion
\mathbb{R}
%EndExpansion
}fB_{\ast }dx=0.
\end{equation*}%
We obtain%
\begin{eqnarray*}
u(x) &=&\int_{x}^{\infty }\frac{\varepsilon ^{2}B_{\ast }(x)}{B_{\ast
}^{2}(s)}F(s)ds \\
\text{with }F(s) &=&\int_{s}^{\infty }f(\tau )B_{\ast }(\tau )d\tau \text{
for }s\geq 0 \\
&=&-\int_{-\infty }^{s}f(\tau )B_{\ast }(\tau )d\tau \text{ for }s\leq 0.
\end{eqnarray*}%
By Fubini's theorem we can write for $x\geq 0$%
\begin{equation*}
u(x)=\varepsilon ^{2}B_{\ast }(x)\int_{x}^{\infty }f(\tau )B_{\ast }(\tau
)\left( \int_{x}^{\tau }\frac{ds}{B_{\ast }^{2}(s)}\right) d\tau
\end{equation*}%
and, for $x\leq 0$%
\begin{eqnarray*}
u(x) &=&-\varepsilon ^{2}B_{\ast }(x)\int_{-\infty }^{x}f(\tau )B_{\ast
}(\tau )\left( \int_{x}^{0}\frac{ds}{B_{\ast }^{2}(s)}\right) d\tau \\
&&-\varepsilon ^{2}B_{\ast }(x)\int_{x}^{0}f(\tau )B_{\ast }(\tau )\left(
\int_{\tau }^{0}\frac{ds}{B_{\ast }^{2}(s)}\right) d\tau .
\end{eqnarray*}%
The asymptotic properties of $B_{\ast }(x)$ at $\pm \infty $ imply, for $%
x\geq 0$ 
\begin{equation*}
|u(x)|e^{\eta x}\leq C\varepsilon ^{2}\int_{x}^{\infty }|f(\tau )e^{\eta
\tau }|(\tau -x)e^{-\eta (\tau -x)}d\tau ,
\end{equation*}%
and for $x\leq 0$%
\begin{eqnarray*}
|u(x)|e^{-\eta x} &\leq &\frac{C\varepsilon ^{2}}{2\varepsilon \delta }%
\int_{-\infty }^{x}|f(\tau )e^{-\eta \tau }|e^{-(\eta +\varepsilon \delta
)(x-\tau )}d\tau \\
&&+\frac{C\varepsilon ^{2}}{2\varepsilon \delta }\int_{x}^{0}|f(\tau
)e^{-\eta \tau }|e^{(\eta -\varepsilon \delta )(\tau -x)}d\tau .
\end{eqnarray*}%
The bound 
\begin{equation*}
||u||_{L_{\eta }^{2}}\leq c_{2}||f||_{L_{\eta }^{2}}
\end{equation*}%
follows from classical convolution results between functions in $L^{2}$ and
functions in $L^{1},$ since%
\begin{eqnarray*}
\int_{-\infty }^{0}e^{(\eta -\varepsilon \delta )\tau }d\tau &=&\frac{1}{%
\eta -\varepsilon \delta }, \\
\int_{0}^{\infty }\tau e^{-\eta \tau }d\tau &=&\frac{1}{\eta ^{2}}.
\end{eqnarray*}%
Then, we choose $\eta =\frac{1}{2}\varepsilon \delta ,$ so that the
pseudo-inverse of $\mathcal{L}_{g}$ has a bounded inverse in $L_{\eta }^{2}:$%
\begin{equation*}
||\widetilde{\mathcal{L}_{g}}^{-1}||\leq c_{2},
\end{equation*}%
where $c_{2}$ is independent of $\varepsilon .$ Using the form of $\mathcal{L%
}_{g}$ we obtain easily%
\begin{equation*}
||u||_{\mathcal{D}_{1}}\leq c_{3}||f||_{L_{\eta }^{2}}
\end{equation*}%
with $c_{3}$ independent of $\varepsilon .$

\begin{remark}
\label{Rem eta}The choice made for $\eta $ is such that 
\begin{equation*}
\eta <\varepsilon \delta ,\text{ \ }\eta <\varepsilon \sqrt{2},
\end{equation*}%
for values of $\delta $ for which Theorem \ref{theorem} is valid. This means
that as $x\rightarrow -\infty $ $(A_{\ast }-1,$ $B_{\ast }),$ and, as $%
x\rightarrow +\infty $ $(A_{\ast },B_{\ast }-1)$ tend exponentially to $0$
faster than $e^{-\eta |x|}.$
\end{remark}

In fact, $\mathcal{L}_{g}$ has the same properties as the operator $\mathcal{%
M}_{i}$ in the proof of Lemma 7.3 in \cite{HI Arma}, see also \cite{H-S 2012}%
: $\mathcal{L}_{g}$ is Fredholm with index -1, when acting in $L_{\eta
}^{2}, $ for $\eta $ small enough. $\mathcal{L}_{g}$ has a trivial kernel,
and its range is orthogonal to $B_{\ast },$ with the scalar product of $%
L^{2}(%
%TCIMACRO{\U{211d} }%
%BeginExpansion
\mathbb{R}
%EndExpansion
).$

\subsection{Properties of $\mathcal{M}_{g}$}

We saw that $\mathcal{M}_{g}$ is Fredholm with index 0. Furthermore the
derivative of the heteroclinic solution belongs to its kernel:%
\begin{eqnarray}
\mathcal{M}_{g}\left( 
\begin{array}{c}
A_{\ast }^{\prime } \\ 
B_{\ast }^{\prime }%
\end{array}%
\right) &=&\left( 
\begin{array}{c}
-A_{\ast }^{(5)}+A_{\ast }^{\prime }-(A_{\ast }^{3})^{\prime }-gB_{\ast
}^{2}A_{\ast }^{\prime }-gA_{\ast }(B_{\ast }^{2})^{\prime } \\ 
\varepsilon ^{-2}B_{\ast }^{\prime \prime \prime }+[B_{\ast }^{\prime
}-gA_{\ast }^{2}B_{\ast }^{\prime }-(B_{\ast }^{3})^{\prime }-gB_{\ast
}(A_{\ast }^{2})^{\prime }]%
\end{array}%
\right)  \notag \\
&=&\left( 
\begin{array}{c}
0 \\ 
0%
\end{array}%
\right) .  \label{ident kern}
\end{eqnarray}%
The part of the proof which differs from the proof made in \cite{HI Arma},
where the symmetry play an essential role, consists in showing at section %
\ref{sect dim Mg} that the kernel of $\mathcal{M}_{g}$ is one-dimensional
(except for a finite set of values of $g$), spanned by $(A_{\ast }^{\prime
},B_{\ast }^{\prime })\overset{def}{=}U_{\ast }$ with a range orthogonal to $%
U_{\ast }$ in $L^{2}$. Let us admit this result for the moment, and define
the projections $Q_{0}$ on $U_{\ast }^{\bot }$ and $P_{0}$ on $U_{\ast }$ ,
which are orthogonal projections in $L^{2},$ then we need to solve in $%
L_{\eta }^{2}$%
\begin{equation*}
\mathcal{M}_{g}u=f
\end{equation*}%
in decomposing%
\begin{equation*}
u=zU_{\ast }+v,\text{ \ }v=Q_{0}u,
\end{equation*}%
\begin{equation*}
\mathcal{M}_{g}v=\mathcal{(M}_{\infty }+\mathcal{A}_{g})v=Q_{0}f
\end{equation*}%
and we need to satisfy the compatibility condition%
\begin{equation*}
\langle f,U_{\ast }\rangle =0,
\end{equation*}%
while $z$ is arbitrary and we obtain for $v:$%
\begin{equation*}
(\mathbb{I}+\mathcal{M}_{\infty }^{-1}\mathcal{A}_{g})v=\mathcal{M}_{\infty
}^{-1}Q_{0}f,
\end{equation*}%
where the operator $\mathcal{M}_{\infty }^{-1}\mathcal{A}_{g}$ is now a
compact operator for which $-1$ is not an eigenvalue, since $v\in U_{\ast
}^{\bot }.$ It results that there is a number $c$ independent of $%
\varepsilon $ such that%
\begin{equation*}
||v||_{L_{\eta }^{2}}\leq c||f||_{L_{\eta }^{2}}.
\end{equation*}%
From the form of operator $\mathcal{M}_{g}$ and using interpolation
properties, we obtain for $v=(A,C)$%
\begin{equation*}
||(A,C)||_{\mathcal{D}_{0}}\leq c||f||_{L_{\eta }^{2}}
\end{equation*}%
with a certain $c$ independent of $\varepsilon .$

We show below (see section \ref{sect dim Mg}) that the kernel of $\mathcal{M}%
_{g},$ is one dimensional, then this implies that the range of $\mathcal{M}%
_{g}$ needs satisfy the orthogonality with only one element. In fact,
because of selfadjointness in $L^{2}$, the range of $\mathcal{M}_{g}$ is
orthogonal in $L^{2}(%
%TCIMACRO{\U{211d} }%
%BeginExpansion
\mathbb{R}
%EndExpansion
)$ to%
\begin{equation*}
(A_{\ast }^{\prime },B_{\ast }^{\prime })\in L_{\eta }^{2}\text{.}
\end{equation*}

\subsubsection{Dimension of $\ker \mathcal{M}_{g}\label{sect dim Mg}$}

Any element $\zeta (x)$ in the kernel lies, by definition, in $L_{\eta
}^{2}, $ hence $\zeta (x)$ tends towards 0 exponentially at $\pm \infty .$
Near $x=\pm \infty $ the vector $\zeta (x)\sim \zeta _{\pm }(x)$ should
verify%
\begin{equation*}
\mathcal{M}_{\infty }^{\pm }\zeta _{\pm }(x)=0
\end{equation*}%
where there are only 2 possible good dimensions (on each side). This gives a
bound $=2$ to the dimension of the kernel of $\mathcal{M}_{g}.$ Let us show
that \emph{dimension 2 of $\ker \mathcal{M}_{g}$ implies non uniqueness of
the heteroclinic}, which contradicts Theorem \ref{theorem}, hence the only
possibility is that the dimension is one.

Let us choose arbitrarily $g_{0}$ and assume that the kernel of $\mathcal{M}%
_{g_{0}}$ consists in%
\begin{equation*}
\zeta _{0}(x),\zeta _{1}(x)
\end{equation*}%
where $\zeta _{0}=(A_{\ast }^{\prime },B_{\ast }^{\prime })|_{g_{0}}$ and
let us decompose a solution of (\ref{truncated syst}) in the neighborhood of 
$g_{0}$ as%
\begin{equation}
U=\mathbf{T}_{a}(U_{\ast }^{(g_{0})}+a_{1}\zeta _{1}+Y),  \label{decomp U}
\end{equation}%
where $\mathbf{T}_{a}$ represents the shift $x\mapsto x+a,$ where $a,$ $%
a_{1}\in 
%TCIMACRO{\U{211d} }%
%BeginExpansion
\mathbb{R}
%EndExpansion
$, and $Y$ belongs to a subspace transverse to $\ker \mathcal{M}_{g_{0}}$.
Let us denote by $\mathbf{Q}_{0}$ and $\mathbf{P}_{0}=\mathbb{I}-\mathbf{Q}%
_{0}$, projections, respectively on the range of $\mathcal{M}_{g_{0}}$, and
on a complementary subspace ($\mathbf{Q}_{0}$ may be built in using the
eigenvectors $\zeta _{0}^{\ast },\zeta _{1}^{\ast }$ of the adjoint operator 
$\mathcal{M}_{g_{0}}^{\ast })$. Let us denote by 
\begin{equation*}
\mathcal{F}(U,g)=0
\end{equation*}%
the system (\ref{truncated syst}) where we look for an heteroclinic $U$ for $%
g\neq g_{0}$. Then, we have%
\begin{eqnarray*}
\mathcal{F}(U_{\ast }^{(g_{0})},g_{0}) &=&0, \\
D_{U}\mathcal{F(}U_{\ast }^{(g_{0})},g_{0}) &=&\mathcal{M}_{g_{0}},
\end{eqnarray*}%
and since 
\begin{equation*}
\mathcal{M}_{g_{0}}\zeta _{j}=0,\text{ }j=0,1,
\end{equation*}%
using the equivariance under operator $\mathbf{T}_{a},$ we obtain (denoting $%
\mathcal{F}_{0}=\mathcal{F}(U_{\ast }^{(g_{0})},g_{0})$ and $[..]^{(2)}$ the
argument of a quadratic operator) 
\begin{eqnarray*}
0 &=&\mathcal{M}_{g_{0}}Y+(g-g_{0})\partial _{g}\mathcal{F}_{0}+\frac{1}{2}%
D_{UU}^{2}\mathcal{F}_{0}[a_{1}\zeta _{1}+Y]^{(2)}+ \\
&&+\mathcal{O}(|g-g_{0}|[|g-g_{0}|+|a_{1}|+||Y||]+||Y||^{3}).
\end{eqnarray*}%
The projection $Q_{0}$ of this equation allows to use the implicit function
theorem to solve with respect to $Y$ and then obtain a unique solution%
\begin{equation*}
Y=\mathcal{Y}(a_{1},g),
\end{equation*}%
with%
\begin{eqnarray*}
\mathcal{Y} &=&-(g-g_{0})\widetilde{\mathcal{M}_{g_{0}}}^{-1}\mathbf{Q}%
_{0}\partial _{g}\mathcal{F}_{0}-\frac{1}{2}\widetilde{\mathcal{M}_{g_{0}}}%
^{-1}\mathbf{Q}_{0}D_{UU}^{2}\mathcal{F}_{0}[a_{1}\zeta _{1}]^{(2)}+ \\
&&+\mathcal{O}(|g-g_{0}|(|g-g_{0}|+|a_{1}|)+|a_{1}|^{3})).
\end{eqnarray*}%
Then projecting on the complementary space, (only one equation since we work
in the subspace orthogonal to $\zeta _{0}^{\ast }),$ we may observe (see the
proof in Appendix \ref{App3}) that $\mathbf{P}_{0}\partial _{g_{0}}\mathcal{F%
}_{0}=0$ and then obtain the "bifurcation" equation as%
\begin{equation}
q(a_{1},g-g_{0})=\mathcal{O}((|g-g_{0}|+|a_{1}|)^{3}),  \notag
\end{equation}%
where the function $q$ is quadratic in its arguments and%
\begin{equation*}
q|_{g=g_{0}}\zeta _{1}=\frac{1}{2}\mathbf{P}_{0}D_{UU}^{2}\mathcal{F}%
_{0}[a_{1}\zeta _{1}]^{(2)}.
\end{equation*}%
This equation is just at main order a second degree equation in $a_{1}$
depending on $g-g_{0}$. Provided that the discriminant is not 0, the generic
number of solutions is 2 or 0. If the discriminant is 0 for $g=g_{0},$ we
just go a little farther in $g,$ and obtain a non zero discriminant, since
the discriminant cannot stay $=0$. Indeed the heteroclinic is analytic in $g$
and if the discriminant were identically 0, this would mean that we have a
double root for any $g$, contradicting the transversality for all g, except
a finite number, of the intersection of the two manifolds (unstable one of $%
M_{-}$, stable one of $M_{+}$). Hence, this is true except for a set of
isolated values of $g.$ We can then use the implicit function theorem for
finding corresponding solutions for the system with higher order terms. In
fact we already know a solution, corresponding to $U_{\ast }^{(g)}=U_{\ast
}^{(g_{0})}+(g-g_{0})\partial _{g}U_{\ast }^{(g_{0})}+h.o.t.$ which
corresponds to specific values for $a_{1}$ and $Y,$ of order $\mathcal{O}%
(g-g_{0}).$ It then results that there is at least another solution of order 
$\mathcal{O}(g-g_{0})$, so that there exists another heteroclinic, in the
neighborhood of the known one (then in contradiction with Theorem \ref%
{theorem}).

\begin{remark}
The above proof with only 1 dimension in the Kernel, provides $Y=-(g-g_{0})%
\widetilde{\mathcal{M}_{g_{0}}}^{-1}\partial _{g}\mathcal{F}_{0}+\mathcal{O(}%
(g-g_{0})^{2})$, which gives a unique heteroclinic. Since we found only one
heteroclinic, this shows that the kernel is of dimension 1.
\end{remark}

\appendix

\section{Appendix}

\subsection{Monodromy operator\label{App1}}

Let us prove the estimate for the monodromy operators. We prove the following

\begin{lemma}
\label{LemMonodromy} For $B_{0}\leq \sqrt{1-\eta _{0}^{2}\delta ^{2}},$ $%
\alpha \geq \frac{10}{3}\varepsilon ^{2},$ and for $\varepsilon $ small
enough, the following estimates hold%
\begin{eqnarray*}
||\boldsymbol{S}_{0}(x,s)|| &\leq &e^{\sigma (x-s)},\text{ }-\infty <x<s \\
||\boldsymbol{S}_{1}(x,s)|| &\leq &e^{-\sigma (x-s)},\text{ }-\infty <s<x
\end{eqnarray*}%
with%
\begin{equation*}
\sigma =\frac{(\alpha \delta )^{1/2}}{2^{1/4}}.
\end{equation*}
\end{lemma}

\begin{proof}
We start with the system%
\begin{eqnarray*}
x_{1}^{\prime } &=&\lambda _{r}x_{1}+\lambda _{i}x_{2} \\
x_{2}^{\prime } &=&-\lambda _{i}x_{1}+\lambda _{r}x_{2}
\end{eqnarray*}%
where $\lambda _{r}$ and $\lambda _{i}$ are functions of $x$. From Lemma \ref%
{Lem estim lambda} we have, for $\varepsilon $ small enough%
\begin{equation*}
\lambda _{r}\geq \frac{(\alpha \delta )^{1/2}}{2^{1/4}}=\sigma .
\end{equation*}%
Now we obtain%
\begin{equation*}
(x_{1}^{2}+x_{2}^{2})^{\prime }=2\lambda _{r}(x_{1}^{2}+x_{2}^{2})
\end{equation*}%
hence%
\begin{equation*}
(x_{1}^{2}+x_{2}^{2})(x)=e^{\int_{s}^{x}2\lambda _{r}(\tau )d\tau
}(x_{1}^{2}+x_{2}^{2})(s),
\end{equation*}%
which, for $x<s,$ leads to 
\begin{equation*}
\sqrt{(x_{1}^{2}+x_{2}^{2})(x)}\leq e^{\sigma (x-s)}\sqrt{%
(x_{1}^{2}+x_{2}^{2})(s)}.
\end{equation*}%
The proof is then done for the operator $\boldsymbol{S}_{0}.$ The estimate
for $\boldsymbol{S}_{1}$ is obtained in the same way.
\end{proof}

\begin{remark}
We have%
\begin{equation*}
\boldsymbol{S}_{0}(x,s)=e^{\int_{s}^{x}\lambda _{r}(\tau )d\tau }\left( 
\begin{array}{cc}
\cos (\int_{s}^{x}\lambda _{i}(\tau )d\tau ) & \sin (\int_{s}^{x}\lambda
_{i}(\tau )d\tau ) \\ 
-\sin (\int_{s}^{x}\lambda _{i}(\tau )d\tau ) & \cos (\int_{s}^{x}\lambda
_{i}(\tau )d\tau )%
\end{array}%
\right) .
\end{equation*}
\end{remark}

\subsection{Computation of the system with new coordinates \label{App1'}}

Let us look for the system (\ref{newsyst 1}) writen in the new coordinates,
first in forgetting quadratic and higher orders terms%
\begin{eqnarray*}
B_{0}x_{1}^{\prime } &=&\frac{\widetilde{A_{\ast }}}{2\sqrt{2}\lambda _{r}}%
\left( A_{1}+\frac{(1+\delta ^{2})B_{0}B_{1}}{\widetilde{A_{\ast }}}\right) +%
\frac{3\lambda _{r}^{2}-\lambda _{i}^{2}}{4\sqrt{2}\lambda _{r}\widetilde{%
A_{\ast }}}A_{3} \\
&&+\frac{A_{2}}{2}+\frac{(1+\delta ^{2})}{2\widetilde{A_{\ast }}}%
B_{0}^{2}\varepsilon ^{2}\left( \delta ^{2}(\widetilde{A_{\ast }}%
^{2}-B_{0}^{2})+2(1+\delta ^{2})\widetilde{A_{\ast }}\widetilde{A_{0}}%
\right) -(\lambda _{r}^{2}-\lambda _{i}^{2})\widetilde{A_{0}} \\
&=&B_{0}f_{1}+\frac{\widetilde{A_{\ast }}}{2\sqrt{2}\lambda _{r}}%
B_{0}(x_{1}+y_{1})+\frac{A_{2}}{2}+\frac{1}{4\lambda _{r}}A_{3},
\end{eqnarray*}%
\begin{eqnarray*}
\lambda _{i}B_{0}x_{2}^{\prime } &=&-\frac{\widetilde{A_{\ast }}}{2\sqrt{2}}%
\left( A_{1}+\frac{(1+\delta ^{2})B_{0}B_{1}}{\widetilde{A_{\ast }}}\right) -%
\frac{\lambda _{r}^{2}-3\lambda _{i}^{2}}{4(\lambda _{r}^{2}-\alpha )}A_{3}
\\
&&-\frac{(\lambda _{r}^{2}-\lambda _{i}^{2})}{4\lambda _{r}}\left( A_{2}+%
\frac{(1+\delta ^{2})B_{0}^{2}\varepsilon ^{2}}{\widetilde{A_{\ast }}}\delta
^{2}(\widetilde{A_{\ast }}^{2}-B_{0}^{2})\right) \\
&&-\frac{1}{4\lambda _{r}}2\widetilde{A_{\ast }}^{2}\widetilde{A_{0}} \\
&=&\lambda _{i}B_{0}f_{2}-\frac{\widetilde{A_{\ast }}}{2\sqrt{2}}%
B_{0}(x_{1}+y_{1})-\frac{(\lambda _{r}^{2}-\lambda _{i}^{2})}{4\lambda _{r}}%
A_{2} \\
&&+\frac{1}{4}A_{3}-\frac{1}{4\lambda _{r}}2\widetilde{A_{\ast }}^{2}%
\widetilde{A_{0}},
\end{eqnarray*}%
with%
\begin{equation*}
f_{1}=\frac{\varepsilon ^{2}\delta ^{2}B_{0}(1+\delta ^{2})(\widetilde{%
A_{\ast }}^{2}-B_{0}^{2})}{2\widetilde{A_{\ast }}},
\end{equation*}%
\begin{equation*}
f_{2}=-\frac{\varepsilon ^{2}\delta ^{2}B_{0}(1+\delta ^{2})(\lambda
_{r}^{2}-\lambda _{i}^{2})(\widetilde{A_{\ast }}^{2}-B_{\ast }^{2})}{%
4\lambda _{r}\lambda _{i}\widetilde{A_{\ast }}},
\end{equation*}%
hence%
\begin{eqnarray}
x_{1}^{\prime } &=&f_{1}+\lambda _{r}x_{1}+\lambda _{i}x_{2},
\label{linsyst 1} \\
x_{2}^{\prime } &=&f_{2}-\lambda _{i}x_{1}+\lambda _{r}x_{2},  \notag
\end{eqnarray}%
as expected. In the same way we obtain%
\begin{eqnarray}
y_{1}^{\prime } &=&f_{1}-\lambda _{r}y_{1}+\lambda _{i}y_{2},  \notag \\
y_{2}^{\prime } &=&-f_{2}-\lambda _{i}y_{1}-\lambda _{r}y_{2},
\label{linsyst 2} \\
z_{1}^{\prime } &=&\frac{2\varepsilon ^{2}\delta ^{2}(\widetilde{A_{\ast }}%
^{2}-B_{0}^{2})}{\widetilde{A_{\ast }}}=\frac{2f_{1}}{(1+\delta ^{2})B_{0}},
\notag \\
B_{\ast }^{\prime } &=&-\frac{(\lambda _{r}^{2}-\lambda _{i}^{2})}{(1+\delta
^{2})B_{0}\widetilde{A_{\ast }}}A_{3}+\widetilde{A_{\ast }}B_{0}z_{1}. 
\notag
\end{eqnarray}%
We notice that the following estimates hold (using (\ref{basic ineq a}) and
Lemma \ref{Lem estim lambda})%
\begin{equation}
|f_{1}|,|f_{2}|\leq \frac{B_{0}\varepsilon ^{2}\delta ^{2}}{\widetilde{%
A_{\ast }}}\leq \frac{B_{0}\varepsilon ^{2}\delta }{\alpha },  \notag
\end{equation}

\subsubsection{Full system in new coordinates}

We intend to derive the full system (\ref{truncated syst}) with coordinates $%
(x_{1},x_{2},y_{1},y_{2},B_{0},z_{1}).$ Differentiating (\ref{x1 1}) and (%
\ref{x2 1}) we see that we respectively need to add to the previous
expressions (\ref{linsyst 1}) for $x_{1}^{\prime }$ and $x_{2}^{\prime }$%
\begin{eqnarray*}
&&\frac{1}{B_{0}}\left\{ \left( \frac{\widetilde{A_{\ast }}}{2\sqrt{2}%
\lambda _{r}}\right) ^{\prime }\widetilde{A_{0}}+\left( \frac{(3\lambda
_{r}^{2}-\lambda _{i}^{2})}{4\sqrt{2}\lambda _{r}\widetilde{A_{\ast }}}%
\right) ^{\prime }A_{2}+\varepsilon ^{2}\left( \frac{(1+\delta
^{2})^{2}B_{0}^{2}}{2\widetilde{A_{\ast }}^{2}}\right) ^{\prime
}A_{3}+\left( \frac{(1+\delta ^{2})B_{0}}{2\widetilde{A_{\ast }}}\right)
^{\prime }B_{1}\right\} \\
&&-\varepsilon ^{2}\frac{(1+\delta ^{2})^{2}B_{0}}{2\widetilde{A_{\ast }}^{2}%
}[3\widetilde{A_{\ast }}\widetilde{A_{0}}^{2}+\widetilde{A_{0}}^{3}]+\frac{%
B_{0}\varepsilon ^{2}(1+\delta ^{2})^{2}\widetilde{A_{0}}^{2}}{2\widetilde{%
A_{\ast }}}-\frac{B_{1}}{B_{0}}x_{1}.
\end{eqnarray*}%
and%
\begin{eqnarray*}
&&\frac{1}{B_{0}}\left\{ -\left( \frac{\widetilde{A_{\ast }}}{2\sqrt{2}%
\lambda _{i}}\right) ^{\prime }\widetilde{A_{0}}-\left( \frac{(\lambda
_{r}^{2}-\lambda _{i}^{2})}{4\lambda _{r}\lambda _{i}}\right) ^{\prime
}A_{1}-\left( \frac{(\lambda _{r}^{2}-3\lambda _{i}^{2})}{4\sqrt{2}\lambda
_{i}\widetilde{A_{\ast }}}\right) ^{\prime }A_{2}+\left( \frac{\varepsilon
^{2}(1+\delta ^{2})^{3}B_{0}^{3}}{4\lambda _{r}\lambda _{i}\widetilde{%
A_{\ast }}}\right) ^{\prime }B_{1}\right\} \\
&&+\frac{1}{B_{0}}\left( \frac{1}{4\lambda _{r}\lambda _{i}}\left[ 1-\frac{%
(\lambda _{r}^{2}-\lambda _{i}^{2})^{2}}{\widetilde{A_{\ast }}^{2}}\right]
\right) ^{\prime }A_{3}-\frac{1}{4\lambda _{r}\lambda _{i}B_{0}}\left( 1-%
\frac{(\lambda _{r}^{2}-\lambda _{i}^{2})^{2}}{\widetilde{A_{\ast }}^{2}}%
\right) [3\widetilde{A_{\ast }}\widetilde{A_{0}}^{2}+\widetilde{A_{0}}^{3}]
\\
&&-\frac{\varepsilon ^{4}B_{0}^{3}(1+\delta ^{2})^{4}}{4\lambda _{r}\lambda
_{i}\widetilde{A_{\ast }}}\widetilde{A_{0}}^{2}-\frac{B_{1}}{B_{0}}x_{2}.
\end{eqnarray*}

We then arrive to the system (\ref{x'1 1},\ref{x'2 1},\ref{y'1 1},\ref{y'2 1}%
).

Using Lemma \ref{Lem estim lambda} and Lemma \ref{Lem delta} we obtain%
\begin{equation*}
\widetilde{A_{\ast }}^{\prime }=-\frac{(1+\delta ^{2})B_{0}}{\widetilde{%
A_{\ast }}}B_{1}
\end{equation*}%
\begin{equation*}
(\lambda _{r}^{2})^{\prime }=-\frac{(1+\delta ^{2})B_{0}B_{1}}{\sqrt{2}%
\widetilde{A_{\ast }}}(1-\varepsilon ^{2}\sqrt{2}(1+\delta ^{2})\widetilde{%
A_{\ast }})
\end{equation*}%
\begin{equation*}
(\lambda _{i}^{2})^{\prime }=-\frac{(1+\delta ^{2})B_{0}B_{1}}{\sqrt{2}%
\widetilde{A_{\ast }}}(1+\varepsilon ^{2}\sqrt{2}(1+\delta ^{2})\widetilde{%
A_{\ast }})
\end{equation*}%
\begin{equation}
\left( \frac{\widetilde{A_{\ast }}}{2\sqrt{2}\lambda _{r}}\right) ^{\prime
}=a_{1}B_{0}B_{1},\text{ \ }|a_{1}|\leq \frac{c}{\widetilde{A_{\ast }}^{3/2}}%
,  \label{a1 1}
\end{equation}%
\begin{equation}
\left( \frac{\widetilde{A_{\ast }}}{2\sqrt{2}\lambda _{i}}\right) ^{\prime
}=a_{2}B_{0}B_{1},\text{ \ }|a_{2}|\leq \frac{c}{\widetilde{A_{\ast }}^{3/2}}%
,  \label{a2 1}
\end{equation}%
\begin{equation}
\left( -\frac{(\lambda _{r}^{2}-\lambda _{i}^{2})}{4\lambda _{r}\lambda _{i}}%
\right) ^{\prime }=b_{2}B_{0}B_{1},\text{ \ }|b_{2}|\leq \frac{c\varepsilon
^{2}}{\widetilde{A_{\ast }}^{2}},  \label{b2 1}
\end{equation}%
\begin{equation}
\left( \frac{(3\lambda _{r}^{2}-\lambda _{i}^{2})}{4\sqrt{2}\lambda _{r}%
\widetilde{A_{\ast }}}\right) ^{\prime }=c_{1}B_{0}B_{1},\text{ \ }%
|c_{1}|\leq \frac{c}{\widetilde{A_{\ast }}^{5/2}},  \label{c1 1}
\end{equation}%
\begin{equation}
\left( -\frac{(\lambda _{r}^{2}-3\lambda _{i}^{2})}{4\sqrt{2}\lambda _{i}%
\widetilde{A_{\ast }}}\right) ^{\prime }=c_{2}B_{0}B_{1},\text{ \ }%
|c_{2}|\leq \frac{c}{\widetilde{A_{\ast }}^{5/2}},  \label{c2 1}
\end{equation}%
\begin{equation}
\varepsilon ^{2}\left( \frac{(1+\delta ^{2})^{2}B_{0}^{2}}{2\widetilde{%
A_{\ast }}^{2}}\right) ^{\prime }=d_{1}B_{0}B_{1},\text{ \ }|d_{1}|\leq 
\frac{c}{\widetilde{A_{\ast }}^{3}},  \label{d1 1}
\end{equation}%
\begin{equation}
\left( \frac{1}{4\lambda _{r}\lambda _{i}}\left[ 1-\frac{(\lambda
_{r}^{2}-\lambda _{i}^{2})^{2}}{\widetilde{A_{\ast }}^{2}}\right] \right)
^{\prime }=d_{2}B_{0}B_{1},\text{ \ }|d_{2}|\leq \frac{c}{\widetilde{A_{\ast
}}^{3}},  \label{d2 1}
\end{equation}%
\begin{equation}
\left( \frac{(1+\delta ^{2})B_{0}}{2\widetilde{A_{\ast }}}\right) ^{\prime
}=e_{1}B_{1},\text{ \ }|e_{1}|\leq \frac{c}{\widetilde{A_{\ast }}^{3}}
\label{e1 1}
\end{equation}%
\begin{equation}
\left( \frac{\varepsilon ^{2}(1+\delta ^{2})^{3}B_{0}^{2}}{4\lambda
_{r}\lambda _{i}\widetilde{A_{\ast }}}\right) ^{\prime }=e_{2}B_{0}B_{1},%
\text{ \ }|e_{2}|\leq \frac{c}{\widetilde{A_{\ast }}^{3}},  \label{e2 1}
\end{equation}%
with $c$ independent of $\varepsilon ,\alpha $ and $\delta \in \lbrack
1/3,1].$

\subsection{Elimination of $z_{1}\label{app eliminationz1}$}

\subsubsection{System after scaling}

After the scaling (\ref{scaling a}) our system (\ref{x'1 1},\ref{x'2 1},\ref%
{y'1 1},\ref{y'2 1}) takes the form%
\begin{eqnarray*}
\overline{X}^{\prime } &=&\mathbf{L}_{0}\overline{X}+B_{0}\overline{F_{0}}+%
\mathbf{B}_{01}(\overline{X},\overline{Y})+\overline{z_{1}}\boldsymbol{M}%
_{01}(\overline{X},\overline{Y}) \\
&&+\overline{z_{1}}^{2}B_{0}\boldsymbol{n}_{0}+\boldsymbol{C}_{01}(\overline{%
X},\overline{Y}), \\
\overline{Y}^{\prime } &=&\mathbf{L}_{1}\overline{Y}+B_{0}\overline{F_{1}}+%
\mathbf{B}_{11}(\overline{X},\overline{Y})+\overline{z_{1}}\boldsymbol{M}%
_{11}(\overline{X},\overline{Y}) \\
&&+\overline{z_{1}}^{2}B_{0}\boldsymbol{n}_{1}+\boldsymbol{C}_{11}(\overline{%
X},\overline{Y}),
\end{eqnarray*}%
where $\overline{F_{0}},\overline{F_{1}},\boldsymbol{n}_{0},\boldsymbol{n}%
_{1}$ are two-dimensional vectors $\mathbf{M}_{01},\mathbf{M}_{11}$ are
linear operators in $(\overline{X},\overline{Y}),$ $\mathbf{B}_{01},\mathbf{B%
}_{11}$ are quadratic and $\mathbf{C}_{01},\mathbf{C}_{11}$ are cubic in $(%
\overline{X},\overline{Y}),$ all functions of $B_{0}.$ More precisely we have%
\begin{equation*}
\overline{F_{0}}=\left( 
\begin{array}{c}
\frac{f_{1}}{\alpha ^{3/2}\delta B_{0}} \\ 
\frac{f_{2}}{\alpha ^{3/2}\delta B_{0}}%
\end{array}%
\right) ,\text{ }\overline{F_{1}}=\left( 
\begin{array}{c}
\frac{f_{1}}{\alpha ^{3/2}\delta B_{0}} \\ 
-\frac{f_{2}}{\alpha ^{3/2}\delta B_{0}}%
\end{array}%
\right) ,\text{ }|\overline{F_{j}}|\leq c\frac{\varepsilon ^{2}}{\alpha
^{5/2}},
\end{equation*}%
\begin{equation*}
\boldsymbol{n}_{0}=\frac{\varepsilon ^{2}\delta }{\alpha ^{3/2}}\left( 
\begin{array}{c}
e_{1}\widetilde{A_{\ast }}^{2} \\ 
e_{2}\widetilde{A_{\ast }}^{2}B_{0}-b_{2}(1+\delta ^{2})\widetilde{A_{\ast }}%
B_{0}^{2}%
\end{array}%
\right) ,
\end{equation*}%
\begin{equation*}
\boldsymbol{M}_{01}(\overline{X},\overline{Y})=\varepsilon \delta \left( 
\begin{array}{c}
m_{01}(\overline{X},\overline{Y}) \\ 
m_{02}(\overline{X},\overline{Y})%
\end{array}%
\right) ,
\end{equation*}%
\begin{eqnarray*}
m_{01}(\overline{X},\overline{Y}) &=&\widetilde{A_{\ast }}B_{0}\left( a_{1}%
\overline{\widetilde{A_{0}}}+c_{1}\overline{A_{2}}+(d_{1}-2e_{1}(1+\delta
^{2})\varepsilon ^{2}\frac{B_{0}}{\widetilde{A_{\ast }}})\overline{A_{3}}-%
\frac{\overline{x_{1}}}{B_{0}}\right) , \\
m_{02}(\overline{X},\overline{Y}) &=&\widetilde{A_{\ast }}B_{0}\left( -a_{2}%
\overline{\widetilde{A_{0}}}+c_{2}\overline{A_{2}}+(d_{2}-2e_{2}(1+\delta
^{2})\varepsilon ^{2}\frac{B_{0}^{2}}{A_{\ast }})\overline{A_{3}}-\frac{%
\overline{x_{2}}}{B_{0}}\right) \\
&&+\widetilde{A_{\ast }}B_{0}^{2}b_{2}(\overline{x_{1}}+\overline{y_{1}}%
)+(1+\delta ^{2})^{2}\varepsilon ^{2}\frac{B_{0}^{3}}{\widetilde{A_{\ast }}}%
b_{2}\overline{A_{3}},
\end{eqnarray*}%
\begin{equation*}
\mathbf{B}_{01}(\overline{X},\overline{Y})=\alpha ^{3/2}\delta \left( 
\begin{array}{c}
b_{01}(\overline{X},\overline{Y}) \\ 
b_{02}(\overline{X},\overline{Y})%
\end{array}%
\right) ,
\end{equation*}%
\begin{eqnarray*}
b_{01}(\overline{X},\overline{Y}) &=&-\varepsilon ^{2}\frac{(1+\delta
^{2})(2-\delta ^{2})B_{0}}{2\widetilde{A_{\ast }}}\overline{\widetilde{A_{0}}%
}^{2}+e_{1}\frac{\varepsilon ^{4}(1+\delta ^{2})^{2}B_{0}}{\widetilde{%
A_{\ast }}^{2}}\overline{A_{3}}^{2} \\
&&-\mathcal{\varepsilon }^{2}\frac{(1+\delta ^{2})B_{0}}{\widetilde{A_{\ast }%
}}\overline{A_{3}}[a_{1}\overline{\widetilde{A_{0}}}+c_{1}\overline{A_{2}}%
+d_{1}\overline{A_{3}}-\frac{\overline{x_{1}}}{B_{0}}],
\end{eqnarray*}%
\begin{eqnarray*}
b_{02}(\overline{X},\overline{Y}) &=&-\frac{1}{4\lambda _{r}\lambda _{i}%
\widetilde{A_{\ast }}B_{0}}\left( 3\widetilde{A_{\ast }}^{2}-2\varepsilon
^{4}B_{0}^{4}(1+\delta ^{2})^{4}\right) \overline{\widetilde{A_{0}}}%
^{2}+e_{2}\frac{\varepsilon ^{4}(1+\delta ^{2})B_{0}^{2}}{\widetilde{A_{\ast
}}^{2}}\overline{A_{3}}^{2} \\
&&-\mathcal{\varepsilon }^{2}\frac{(1+\delta ^{2})B_{0}}{\widetilde{A_{\ast }%
}}\overline{A_{3}}[-a_{2}\overline{\widetilde{A_{0}}}+b_{2}B_{0}(\overline{%
x_{1}}+\overline{y_{1}})+c_{2}\overline{A_{2}}+d_{2}\overline{A_{3}}-\frac{%
\overline{x_{2}}}{B_{0}}],
\end{eqnarray*}%
\begin{equation*}
\boldsymbol{C}_{01}(\overline{X},\overline{Y})=\alpha ^{3}\delta ^{2}%
\overline{\widetilde{A_{0}}}^{3}\left( 
\begin{array}{c}
-\varepsilon ^{2}\frac{(1+\delta ^{2})B_{0}}{2\widetilde{A_{\ast }}^{2}} \\ 
-\frac{1}{4\lambda _{r}\lambda _{i}B_{0}}\left( 1-\frac{\varepsilon
^{4}B_{0}^{4}(1+\delta ^{2})^{4}}{\widetilde{A_{\ast }}^{2}}\right)%
\end{array}%
\right) .
\end{equation*}%
$\boldsymbol{n}_{1}$, $\boldsymbol{M}_{11},$ $\boldsymbol{B}_{11},$ $%
\boldsymbol{C}_{11}$ are deduced respectively from $\boldsymbol{n}_{0},$ $%
\boldsymbol{M}_{01},\boldsymbol{B}_{01},$ $\boldsymbol{C}_{01}$ in changing $%
(a_{1},c_{1},b_{2},d_{2},e_{2})$ into their opposite.

\subsubsection{System after elimination of $z_{1}$}

Let us replace $\overline{z_{1}}$ by $\overline{z_{10}}[1+\mathcal{Z}(%
\overline{X},\overline{Y},B_{0},\varepsilon ,\alpha ,\delta )]$ in the
differential system for $(\overline{X},\overline{Y}).$ The new system
becomes (notice that $B_{0}$ is in factor of the "constant" terms)%
\begin{eqnarray*}
\overline{X}^{\prime } &=&\mathbf{L}_{0}\overline{X}+B_{0}\mathcal{F}_{0}+%
\mathcal{L}_{01}(\overline{X},\overline{Y})+\mathcal{B}_{01}(\overline{X},%
\overline{Y}), \\
\overline{Y}^{\prime } &=&\mathbf{L}_{1}\overline{Y}+B_{0}\mathcal{F}_{1}+%
\mathcal{L}_{11}(\overline{X},\overline{Y})+\mathcal{B}_{11}(\overline{X},%
\overline{Y}),
\end{eqnarray*}%
which is (\ref{newsyst XY a}) with%
\begin{equation*}
\mathcal{F}_{0}=\overline{F_{0}}+\overline{z_{10}}^{2}\boldsymbol{n}_{0},
\end{equation*}%
\begin{equation*}
\mathcal{L}_{01}(\overline{X},\overline{Y})=\overline{z_{10}}\boldsymbol{M}%
_{01}(\overline{X},\overline{Y}),
\end{equation*}%
\begin{eqnarray*}
\mathcal{B}_{01}(\overline{X},\overline{Y}) &=&\mathbf{B}_{01}(\overline{X},%
\overline{Y})+\overline{z_{10}}\mathcal{Z}(\overline{X},\overline{Y})%
\boldsymbol{M}_{01}(\overline{X},\overline{Y})+\boldsymbol{C}_{01}(\overline{%
X},\overline{Y}) \\
&&+2\overline{z_{10}}^{2}\mathcal{Z}(\overline{X},\overline{Y})B_{0}%
\boldsymbol{n}_{0}+\overline{z_{10}}^{2}\mathcal{Z}(\overline{X},\overline{Y}%
)^{2}B_{0}\boldsymbol{n}_{0}.
\end{eqnarray*}%
In using estimates (\ref{estim coord 1}), (\ref{a1 1}) to (\ref{e2 1}), it
is straightforward to check that%
\begin{equation*}
|\mathcal{F}_{0}|+|\mathcal{F}_{1}|\leq \frac{c\varepsilon ^{2}}{\alpha
^{9/2}},
\end{equation*}%
\begin{eqnarray*}
|\boldsymbol{M}_{01}(\overline{X},\overline{Y})| &\leq &c\frac{\varepsilon
\delta }{\widetilde{A_{\ast }}}(|\overline{X}|+|\overline{Y}|), \\
|\boldsymbol{n}_{0}| &\leq &c\frac{\varepsilon ^{2}}{\alpha ^{5/2}},\text{ }%
|b_{01}|\leq c\frac{\varepsilon ^{2}}{\alpha ^{2}},\text{ }|b_{02}|\leq 
\frac{9}{2\alpha }+c\frac{\varepsilon ^{2}}{\alpha ^{2}},
\end{eqnarray*}%
hence%
\begin{equation*}
|\mathcal{L}_{01}(\overline{X},\overline{Y})|+|\mathcal{L}_{11}(\overline{X},%
\overline{Y})|\leq c\frac{\varepsilon }{\alpha ^{2}}(|\overline{X}|+|%
\overline{Y}|).
\end{equation*}%
For higher order terms we have%
\begin{eqnarray*}
|\mathbf{B}_{01}(\overline{X},\overline{Y})| &\leq &\alpha ^{1/2}[\frac{9}{2}%
+c\frac{\varepsilon ^{2}}{\alpha }](|\overline{X}|+|\overline{Y}|)^{2}, \\
|2\overline{z_{10}}^{2}\mathcal{Z}(\overline{X},\overline{Y})\boldsymbol{n}%
_{0}| &\leq &c\frac{\varepsilon ^{2}(1+\rho ^{2})}{\alpha ^{3/2}}(|\overline{%
X}|+|\overline{Y}|)^{2}, \\
|\overline{z_{10}}\mathcal{Z}(\overline{X},\overline{Y})\boldsymbol{M}_{01}(%
\overline{X},\overline{Y})| &\leq &c\alpha \varepsilon (1+\rho ^{2})(|%
\overline{X}|+|\overline{Y}|)^{3}, \\
\overline{z_{10}}^{2}|\mathcal{Z}(\overline{X},\overline{Y})^{2}\boldsymbol{n%
}_{0}| &\leq &c\alpha ^{3/2}\varepsilon ^{2}(1+\rho ^{4})(|\overline{X}|+|%
\overline{Y}|)^{4}, \\
|\boldsymbol{C}_{01}(\overline{X},\overline{Y})| &\leq &\frac{27\alpha ^{1/2}%
}{2}(|\overline{X}|+|\overline{Y}|)^{3},
\end{eqnarray*}%
hence for 
\begin{equation*}
|\overline{X}|+|\overline{Y}|\leq \rho ,
\end{equation*}%
we obtain (with $c$ independent of $\varepsilon ,\alpha ,\delta )$%
\begin{eqnarray*}
|\mathcal{B}_{01}(\overline{X},\overline{Y})|+|\mathcal{B}_{11}(\overline{X},%
\overline{Y})| &\leq &\alpha ^{1/2}[9/2+c\frac{\varepsilon ^{2}}{\alpha }+c%
\frac{\varepsilon ^{2}}{\alpha ^{2}}(1+\rho ^{2})](|\overline{X}|+|\overline{%
Y}|)^{2} \\
&&+[\frac{27\alpha ^{1/2}}{2}+c\alpha \varepsilon (1+\rho ^{2})](|\overline{X%
}|+|\overline{Y}|)^{3}+c\alpha ^{3/2}\varepsilon ^{2}(1+\rho ^{4})(|%
\overline{X}|+|\overline{Y}|)^{4},
\end{eqnarray*}%
and in using the constraint (\ref{cond ro khi})%
\begin{eqnarray*}
|\mathcal{B}_{01}(\overline{X},\overline{Y})|+|\mathcal{B}_{11}(\overline{X},%
\overline{Y})| &\leq &\alpha ^{1/2}(9/2+c\frac{\varepsilon ^{2}}{\alpha
^{7/2}})(|\overline{X}|+|\overline{Y}|)^{2} \\
&&+\alpha ^{1/2}(\frac{27}{2}+c\frac{\varepsilon }{\alpha })(|\overline{X}|+|%
\overline{Y}|)^{3}+\alpha ^{1/2}(c\frac{\varepsilon ^{2}}{\alpha ^{2}})(|%
\overline{X}|+|\overline{Y}|)^{4}.
\end{eqnarray*}

\subsection{Proof of Lemma \protect\ref{conjecture}\label{app lem}}

\subsubsection{First step: integration on $[-a_{+},a_{+}]$}

We have a clear control on $a_{+}$ while this is more complicate for $a_{-}$
which is possibly larger. Hence, we consider the 4th-order differential
equation (\ref{limit interm equ}) with boundary conditions (\ref{bound cond
z=1}) and first find the solution on the interval $[-a_{+},a_{+}].$

Integrating simply the 4th order ODE (\ref{limit interm equ}) leads to%
\begin{eqnarray*}
\overline{A_{j}}(z) &=&\overline{A_{j+}}+\int_{a_{+}}^{z}\overline{A_{j+1}}%
(s)ds,\text{ }j=0,1,2, \\
\overline{A_{3}(z)} &=&\overline{A_{3+}}-\int_{a_{+}}^{z}\overline{A_{0}}(s)(%
\overline{A_{0}}^{2}(s)+s)ds,
\end{eqnarray*}%
where%
\begin{equation*}
\overline{A_{j+}}=\frac{d^{j}\overline{A_{0}}}{dz^{j}}(a_{+}),\text{ }%
j=0,1,2,3.
\end{equation*}%
This gives for $z<a_{+}$%
\begin{eqnarray}
\overline{A_{0}}(z) &=&\overline{A_{0+}}+(z-a_{+})\overline{A_{1+}}+\frac{%
(z-a_{+})^{2}}{2}\overline{A_{2+}}+\frac{(z-a_{+})^{3}}{6}\overline{A_{3+}} 
\notag \\
&&+\int_{z}^{a_{+}}\frac{(z-s)^{3}}{6}\overline{A_{0}}(s)[\overline{A_{0}}%
^{2}(s)+s]ds.  \label{equ integ  A0+}
\end{eqnarray}%
This leads to the estimate%
\begin{eqnarray}
|\overline{A_{0}}(z)| &\leq &|\overline{A_{0+}}|+(a_{+}-z)|\overline{A_{1+}}%
|+\frac{(z+a_{+})^{2}}{2}|\overline{A_{2+}}|+\frac{(a_{+}-z)^{3}}{6}|%
\overline{A_{3-}}|  \notag \\
&&+\int_{z}^{a_{+}}\frac{(s-z)^{3}}{6}|\overline{A_{0}}(s)||\overline{A_{0}}%
^{2}(s)+s|ds,\text{ \ for }z<a_{+}.  \label{estim equ integ A0+}
\end{eqnarray}%
Let us define%
\begin{equation*}
||\overline{A_{0}}||_{0}=\underset{z\in (-a_{+},a_{+})}{\sup }|\overline{%
A_{0}}(z)|,
\end{equation*}%
and look for a solution $\overline{A_{0}}\in C^{0}[-a_{+},a_{+}]$ of (\ref%
{equ integ A0+}). The estimate (\ref{estim equ integ A0+}) leads to%
\begin{equation*}
||\overline{A_{0}}||_{0}\leq |\overline{A_{0+}}|+2a_{+}|\overline{A_{1+}}%
|+2a_{+}^{2}|\overline{A_{2+}}|+\frac{4a_{+}^{3}}{3}|\overline{A_{3+}}|+%
\frac{2a_{+}^{5}}{3}||\overline{A_{0}}||_{0}+\frac{2a_{+}^{4}}{3}||\overline{%
A_{0}}||_{0}^{3}.
\end{equation*}%
The fixed point theorem applies in a small ball for $\overline{A_{0}},$
provided that $\frac{2a_{+}^{5}}{3}<1.$ From (\ref{def a+-}) this condition
is equivalent to 
\begin{equation}
\frac{\nu _{+}}{\sqrt{\delta }}>\frac{\sqrt{1+\delta ^{2}}}{2.6^{1/4}}.
\label{cond nu delta}
\end{equation}%
We notice that (\ref{cond nu delta}) is compatible with (\ref{cond nu+})
since for $\delta <1$%
\begin{equation*}
\frac{\sqrt{1+\delta ^{2}}}{2.6^{1/4}}\leq 0.452<0.46.
\end{equation*}%
Moreover, we have the estimates%
\begin{equation*}
|\overline{A_{0+}}|\leq a_{+}^{1/2}k_{1},\text{ }|\overline{A_{1+}}|\leq
a_{+}^{3/4}k_{1},\text{ }|\overline{A_{2+}}|\leq a_{+}k_{1},\text{ }|%
\overline{A_{3+}}|\leq a_{+}^{5/4}k_{1},
\end{equation*}%
so that%
\begin{equation*}
|\overline{A_{0+}}|+2a_{+}|\overline{A_{1+}}|+2a_{+}^{2}|\overline{A_{2+}}|+%
\frac{4a_{+}^{3}}{3}|\overline{A_{3+}}|\leq
k_{1}a_{+}^{1/2}(1+2a_{+}^{5/4}+2a_{+}^{5/2}+\frac{4a_{+}^{15/4}}{3}).
\end{equation*}%
Choosing $k_{1}$ small enough, and assuming that (\ref{cond nu delta})
holds, we then find a unique fixed point $\overline{A_{0}}$ in $%
C^{0}[-a_{+},a_{+}]$ function of the two parameters $(\overline{x_{10}}^{s},%
\overline{x_{20}}^{s}).$

\subsubsection{Second step: integration on $[-a_{-},-a_{+}]$}

For extending the solution on $[-a_{-},-a_{+}]$ we need to solve with
respect to $\overline{A_{0}}\in C^{0}[-a_{-},-a_{+}]$%
\begin{eqnarray*}
\overline{A_{0}}(z) &=&\overline{A_{0}}(-a_{+})+(z+a_{+})\overline{A_{1}}%
(-a_{+})+\frac{(z+a_{+})^{2}}{2}\overline{A_{2}}(-a_{+})+\frac{(z+a_{+})^{3}%
}{6}\overline{A_{3}}(-a_{+}) \\
&&+\int_{z}^{-a_{+}}\frac{(z-s)^{3}}{6}\overline{A_{0}}(s)[\overline{A_{0}}%
^{2}(s)+s]ds,
\end{eqnarray*}%
where $z\in \lbrack -a_{-},-a_{+}].$ For applying the fixed point argument
for $\overline{A_{0}}$ close to $0$, as above, we need to satisfy%
\begin{equation*}
\int_{-a_{-}}^{-a_{+}}\frac{|(s+a_{-})^{3}s|}{6}ds<1,
\end{equation*}%
i.e.%
\begin{equation*}
\frac{(a_{-}-a_{+})^{5}}{120}+\frac{(a_{-}-a_{+})^{4}a_{+}}{24}<1,
\end{equation*}%
where we notice that%
\begin{equation}
\frac{a_{+}}{a_{-}}=\left( \frac{\nu _{-}}{\nu _{+}}\right) ^{4/5}.
\label{nu+/nu-}
\end{equation}%
Using (\ref{nu+/nu-}) and (\ref{cond nu delta}) the above condition holds if
the following condition on $\nu _{+}/\nu _{-}$ holds%
\begin{equation*}
\frac{1}{80}\left( (\frac{\nu _{+}}{\nu _{-}})^{4/5}-1\right) ^{5}+\frac{1}{%
16}\left( (\frac{\nu _{+}}{\nu _{-}})^{4/5}-1\right) ^{4}<1.
\end{equation*}%
We may check that%
\begin{equation*}
(1/80)X^{5}+(1/16)X^{4}<1
\end{equation*}%
holds for 
\begin{equation*}
X<1.84
\end{equation*}%
hence we need%
\begin{equation*}
\frac{\nu _{+}}{\nu _{-}}<3.69.
\end{equation*}%
If this is realized, we are done! If not, we need to iterate as follows:

i) the first step is as above, then we reach $-a_{1}$ such that%
\begin{equation*}
a_{1}=a_{+}(\frac{\nu _{+}}{\nu _{1}})^{4/5},\text{ \ }\nu _{1}=\frac{\nu
_{+}}{3.69},
\end{equation*}%
and the solution is obtained on $[-a_{1},+a_{+}]$ for $k_{1}$ small enough.

ii) The second step starts at $z=-a_{1}$ and proceeds as above. We then
reach $-a_{2}$ such that%
\begin{equation*}
a_{2}=a_{1}(\frac{\nu _{1}}{\nu _{2}})^{4/5},\ \nu _{2}=\frac{\nu _{1}}{3.69}%
,
\end{equation*}%
and the solution is obtained on $[-a_{2},+a_{+}]$ for $k_{1}$ small enough.

iii) We iterate the process $n$ times, until%
\begin{equation*}
\nu _{n}=\frac{\nu _{+}}{(3.69)^{n}}\leq \nu _{-}.
\end{equation*}%
Then, the solution is obtained, for $k_{1}$ small enough, on $%
[-a_{n},+a_{+}] $ where 
\begin{equation*}
-a_{n}=-a_{+}\left( \frac{\nu _{+}}{\nu _{n}}\right) ^{4/5}\leq -a_{-}.
\end{equation*}%
The Lemma is proved.

\subsection{Proof of $\mathbf{P}_{0}\partial _{g}\mathcal{F}_{0}=0\label%
{App3}$}

\begin{lemma}
Any $(u,v)$ in the kernel of $\mathcal{M}_{g}$ satisfies%
\begin{equation*}
\int_{%
%TCIMACRO{\U{211d} }%
%BeginExpansion
\mathbb{R}
%EndExpansion
}A_{\ast }B_{\ast }(B_{\ast }u+A_{\ast }v)dx=0,
\end{equation*}%
and $\partial _{g}\mathcal{F}_{0}(U_{\ast },g)=(A_{\ast }B_{\ast
}^{2},A_{\ast }^{2}B_{\ast })$ belongs to the range of $\mathcal{M}_{g},$
hence $\mathbf{P}_{0}\partial _{g}\mathcal{F}_{0}=0.$
\end{lemma}

\emph{Proof.}

Differentiating with respect to $g$ the system (\ref{truncated syst})
verified by the heteroclinic, we obtain 
\begin{equation*}
\mathcal{M}_{g}\left( 
\begin{array}{c}
\partial _{g}A_{\ast } \\ 
\partial _{g}B_{\ast }%
\end{array}%
\right) =\left( 
\begin{array}{c}
A_{\ast }B_{\ast }^{2} \\ 
A_{\ast }^{2}B_{\ast }%
\end{array}%
\right) =\partial _{g}\mathcal{F}_{0}(U_{\ast },g),
\end{equation*}%
hence $(A_{\ast }B_{\ast }^{2},A_{\ast }^{2}B_{\ast })$ belongs to the range
of $\mathcal{M}_{g}.$ When $(u,v)\in \ker \mathcal{M}_{g},$ then $(u,v)\in
\ker \mathcal{M}_{g}^{\ast }$ where $\mathcal{M}_{g}=\mathcal{M}_{g}^{\ast
}, $ when the adjoint is computed with the scalar product of $L^{2},$ hence 
\begin{equation}
\int_{%
%TCIMACRO{\U{211d} }%
%BeginExpansion
\mathbb{R}
%EndExpansion
}A_{\ast }B_{\ast }(B_{\ast }u+A_{\ast }v)dx=0.  \label{prop3 kernel}
\end{equation}%
\ Hence, the eigenvectors $\zeta _{0}^{\ast },\zeta _{1}^{\ast }$ of the
adjoint $\mathcal{M}_{g}^{\ast }$ (the orthogonal of this 2-dimensional
eigenspace is the range of $\mathcal{M}_{g}),$ are orthogonal to $\partial
_{g}\mathcal{F}_{0}=(A_{\ast }B_{\ast }^{2},A_{\ast }^{2}B_{\ast })|_{g_{0}}$
in $L^{2}.$


\begin{thebibliography}{99}
\bibitem{BHI} B.Buffoni, M.Haragus, G.Iooss. Heteroclinic orbits for a
system of amplitude equations for orthogonal domain walls. J.Diff.Equ,2023.
https://doi.org/10.1016/j.jde.2023.01.026.

\bibitem{Buff} B.Buffoni. On minimizers of an integral functional arising in
the B\'{e}nard-Rayleigh convection problem. Preprint 2023.

\bibitem{Dieudo} J.Dieudonn\'{e}. El\'{e}ments d'Analyse. vol1.
Gauthier-Villars, Paris 1969.

\bibitem{Feni} N.Fenichel. Geometric singular perturbation theory for
ordinary differential equations. J. Diff. Equ. 31 (1979), 1, 53-98.

\bibitem{Hale} J.Hale. Ordinary differential equations. Wiley, New York,
1969.

%\bibitem{HIbook} M.Haragus, G.Iooss. Local bifurcations, Center manifolds,
%and Normal forms in infinite-dimensional dynamical systems. Universitext.
%Springer-Verlag London, EDP Sciences, Les Ulis, 2011.

\bibitem{HI Arma} M.Haragus, G.Iooss. Bifurcation of symmetric domain walls
for the B\'{e}nard-Rayleigh convection problem. Arch. Rat. Mech. Anal.
239(2), 733-781, 2020.

\bibitem{HIbook} M.Haragus, G.Iooss. Local Bifurcations, Center Manifolds,
and Normal forms in Infinite Dimensional Dynamical Systems. Universitext.
Springer-Verlag London, Ltd., London; EDP Sciences, Les Ulis, 2011.

\bibitem{H-S 2012} M.Haragus, A.Scheel. Grain boundaries in the
Swift-Hohenberg equation. Europ. J. Appl. Math. 23 (2012), 737-759.

%\bibitem{Horn} J.Horn. Pattern formation at a fluid-ferrofluid interface.
%PHD dissertation. Saarbr\"{u}cken, 2023.

%\bibitem{Io-Pe} G.Iooss, M.C.P\'erou\`eme. Perturbed homoclinic solutions in
%reversible 1:1 resonance vector fields. J.Diff. Equ. 102 (1993), 62-88.

\bibitem{Io24} G.Iooss. Existence of orthogonal domain walls in B\'{e}%
nard-Rayleigh convection. J. Math. Fluid Mech. 2025, 27,4.  https://doi.org/10.1007/s00021-024-00891-2

\bibitem{IoMiDe} G.Iooss, A.Mielke, Y.Demay. Theory of steady
Ginzburg-Landau equation in hydrodynamic stability problems. Eur. J. Mech.
B/fluids 8, 1989, 229-268.

\bibitem{Kap-Pro} T.Kapitula, K.Promislow. Spectral and Dynamical Stability
of Nonlinear Waves. Springer series, Appl. Math. Sci. 185. 2013.

\bibitem{Kato} T.Kato. Perturbation theory for linear operators. Classics in
Maths. Springer-Verlag, Berlin, 1995 (1st ed. in 1966).

\bibitem{Kirch82} K.Kirchg\"{a}ssner. Wave-solutions of reversible systems
and applications. J. Diff. Equ. 45, 1982, 113-127.

\bibitem{Krupa} M.Krupa, P.Szmolyan. Geometric analysis of the singularly
perturbed planar fold. Multiple -time-scale Dyn. Syst. 2001, 89-116.

\bibitem{Man-Pom} P.Manneville, Y.Pomeau. A grain boundary in cellular
structures near the onset of convection. Phil. Mag. A, 1983, 48, 4, 607-621.
\end{thebibliography}
\end{document}